\documentclass[10pt]{article}
\usepackage{geometry}  
\usepackage{authblk}
\usepackage[english]{babel}
\usepackage{amsthm}
\usepackage[utf8]{inputenc}
\usepackage{graphicx}
\usepackage{subcaption}
\usepackage{amssymb}
\usepackage{amsmath}
\usepackage{algorithm}
\usepackage{algorithmic}
\usepackage{bbm}
\usepackage{bm}
\usepackage{mathtools}
\usepackage{yhmath}
\usepackage{hyperref}
\usepackage{natbib}
\usepackage{cleveref} 
\usepackage{multirow}
\usepackage{caption}
\usepackage{array}
\usepackage{longtable} 
\usepackage{xcolor}
\usepackage{booktabs}
\usepackage{diagbox}
\usepackage{multicol}
\usepackage{multirow}

\newtheorem{proposition}{Proposition}
\newtheorem{lemma}{Lemma}
\newtheorem{remark}{Remark}
\newtheorem{corollary}{Corollary}
\theoremstyle{remark}

\DeclareMathOperator*{\argmin}{arg\,min}
\newcommand{\indep}{\perp \!\!\! \perp}
\newcommand{\bX}{\mathbf{X}}
\newcommand{\bY}{\mathbf{Y}}
\newcommand{\bR}{\mathbf{R}}
\newcommand{\bD}{\mathbf{D}}
\newcommand{\bI}{\mathbf{I}}
\newcommand{\bSigma}{\mathbf{\Sigma}}
\newcommand{\bE}{\mathbf{E}}
\newcommand{\bG}{\mathbf{G}}
\newcommand{\bP}{\mathbf{P}}
\newcommand{\bV}{\mathbf{V}}
\newcommand{\bH}{\mathbf{H}}
\newcommand{\bs}{\mathbf{s}}
\newcommand{\bQ}{\mathbf{Q}}
\newcommand{\bZ}{\mathbf{Z}}
\newcommand{\bS}{\mathbf{S}}
\newcommand{\bW}{\mathbf{W}}
\newcommand{\bU}{\mathbf{U}}
\newcommand{\bC}{\mathbf{C}}

\newcommand{\bL}{\mathbf{L}}
\newcommand{\bd}{\mathbf{d}}
\newcommand{\br}{\mathbf{r}}

\newcommand{\bx}{\mathbf{x}}

\newcommand{\cG}{\mathcal{G}}

\newcommand{\bzero}{\mathbf{0}}
\newcommand{\bLambda}{\bm{\Lambda}}
\newcommand{\bTheta}{\bm{\Theta}}
\newcommand{\bgamma}{\bm{\gamma}}
\newcommand{\bPi}{\bm{\Pi}}
\newcommand{\bbeta}{\bm{\beta}}
\newcommand{\bepsilon}{\bm{\epsilon}}

\newcommand{\bbP}{\mathbb{P}}
\newcommand{\bbR}{\mathbb{R}}

\newcommand{\Lp}{\left(}
\newcommand{\Rp}{\right)}
\newcommand{\Ls}{\left[}
\newcommand{\Rs}{\right]}

\DeclarePairedDelimiter\abs{\lvert}{\rvert}
\DeclarePairedDelimiter\norm{\lVert}{\rVert}

\providecommand{\keywords}[1]
{
  \small	
  \textbf{\textit{Keywords---}} #1
}

\DeclareFontFamily{U}{mathx}{\hyphenchar\font45}
\DeclareFontShape{U}{mathx}{m}{n}{
      <5> <6> <7> <8> <9> <10>
    <10.95> <12> <14.4> <17.28> <20.74> <24.88>
    mathx10
}{}
\DeclareSymbolFont{mathx}{U}{mathx}{m}{n}
\DeclareMathAccent{\widecheck}{0}{mathx}{"71}

\begin{document}

\title{Controlled Variable Selection from Summary Statistics Only? \\A Solution via GhostKnockoffs and Penalized Regression}

	%\date{\today}
	\author[1$^*$]{Zhaomeng Chen}
	\author[2,3$^*$]{Zihuai He}
    \author[4]{Benjamin B. Chu}
    \author[2]{Jiaqi Gu}
    \author[1]{Tim Morrison} 
    \author[1,4]{\hspace{9mm}Chiara Sabatti} 
    \author[1,5]{Emmanuel Candès}
        
    \affil[1]{Department of Statistics, Stanford University}
	\affil[2]{Department of Neurology and Neurological Sciences, Stanford University}
	\affil[3]{Department of Medicine (Biomedical Informatics Research), Stanford University}
	\affil[4]{Department of Biomedical Data Science, Stanford University}
	\affil[5]{Department of Mathematics, Stanford University}
	\date{}                     %% if you don't need date to appear
	\setcounter{Maxaffil}{0}
	\renewcommand\Affilfont{\itshape\small}
	\maketitle
 
	\def\thefootnote{*}\footnotetext{Equal contribution.}

\newcommand{\ejc}[1]{\textcolor{red}{[EJC: #1]}}
\newcommand{\jg}[1]{\textcolor{purple}{[JG: #1]}}
\newcommand{\zc}[1]{\textcolor{blue}{[ZC: #1]}}

\begin{abstract}
Identifying which variables do influence a response while controlling false positives pervades statistics and data science. In this paper, we consider a scenario in which we only have access to summary statistics, such as the values of marginal empirical correlations between each dependent variable of potential interest and the response. This situation may arise due to privacy concerns, e.g., to avoid the release of sensitive genetic information. We extend GhostKnockoffs \cite{ghostknockoffs} and introduce variable selection methods based on penalized regression achieving false discovery rate (FDR) control. We report empirical results in extensive simulation studies, demonstrating enhanced performance over previous work. We also apply our methods to genome-wide association studies of Alzheimer's disease, and evidence a significant improvement in power.
\end{abstract}

\section{Introduction} \label{section:intro}
\subsection{Background and contributions}
Modern large-scale studies frequently involve a multitude of explanatory variables potentially associated with an outcome we would like to better understand. Oftentimes, the goal is to select those explanatory variables that are meaningfully associated with the response variable. For instance, with recent advances in genome sequencing technologies and genotype imputation techniques, one can now gather tens of millions of variants from hundreds of thousands of samples in large-scale genetic studies, with the aim of pinpointing which genetic variants are biologically associated with specific diseases. This information could provide mechanistic insights and potentially aid the development of targeted drugs. In statistics, this challenge is typically framed as a multiple testing problem. Further, due to the sheer number of hypotheses considered and the cost of following false leads, it is generally required to control some form of error rate on the false positives.

In this paper, we focus on controlling the false discovery rate (FDR), which is the expected proportion of false selections among all selected variables. Compared to the more stringent familywise error rate (FWER) control, keeping the FDR under a nominal level allows for more discoveries while maintaining a reasonable statistical guarantee on the rate of false positives. Several methods for FDR control have been proposed in the literature, with the Benjamini-Hochberg procedure being particularly popular \citep{bh}. However, these approaches often assume a parametric model or the existence of valid $p$-values, which remains difficult, and even problematic, in high-dimensional settings. 

\cite{model-x} proposed the model-X knockoffs, a broad and flexible framework which allows the statistician to select variables that retain dependence with the response conditional on all other covariates while maintaining FDR control. Model-X knockoffs differs from previous approaches in that (1) it makes no modeling assumptions on the distribution of the response $Y$ we wish to study conditional on the family of covariates $X$, and (2) it does not require the construction of valid $p$-values. Instead, the crucial assumption is that the distribution of $X$ is known. The main idea in \cite{model-x} is to generate fake variables $\widetilde{X}$, \textit{knockoffs}, which we can view as negative controls and can be used to tease apart variables that do influence the response from those who do not. Model-X knockoffs has proved effective in a number of real-world applications, particularly in GWAS; see \cite{bates2020causal}, \cite{sesia2021false} and \cite{ghostknockoffs} for examples.

To deploy model-X knockoffs, researchers must have in hand the covariates and responses from all samples. However, in certain situations, individual-level data that may reveal sensitive personal information is not readily accessible. For example, due to privacy concerns, many GWAS studies only publish summary statistics of the original data \citep{pasaniuc2017dissecting}. Yet in such cases, we would still like to develop controlled variable selection methods that rely solely on summary statistics. In genetic studies, this would enable us to utilize available summary data from different data centers to conduct meta-analysis, enhancing the effective sample size and improving variable selection power. On this front, \cite{ghostknockoffs} proposed the framework of GhostKnockoffs, which implements the knockoffs procedure with the marginal correlation difference feature importance statistic directly from summary statistics. As we shall review next, the main idea is to generate knockoff $Z-$scores directly without creating knockoff variables; all that is needed are marginal correlations between the response and the features under study. In details, with $n$ being the sample size and $p$ the number of variables being assayed, the method operates with only $\bX^\top \bY$ and $\norm{\bY}_2^2$, where $\bX$ is the $n \times p$ matrix of covariates, and $\bY$ is the $n \times 1$ response vector. 

In this paper, we extend the family of GhostKnockoffs methods to incorporate feature importance statistics obtained from penalized regression. We first consider in Section \ref{section:partialghost} the situation in which the empirical covariance of the covariate-response pair $(X,Y)$ is available; with the above notation, this means that the summary statistics $\bX^\top \bX$, $\bX^\top \bY$, $\norm{\bY}_2^2$ are available along with the sample size $n$. Unsurprisingly, we observe substantial power improvement over the method of \cite{ghostknockoffs} because we can now employ far more effective test statistics. Next, in Section \ref{section:fullghost}, we consider the case where the empirical covariance $\bX^\top \bX$ of the features is not available. There, we propose new imputation methods that consistently outperform \cite{ghostknockoffs} in comprehensive synthetic and semi-synthetic simulations and rigorously control the FDR under suitable conditions. Finally, in Section \ref{section:gwas} we apply our methods to a meta-analysis of nine large-scale array-based genome-wide association and whole-exome/-genome sequencing studies of Alzheimer’s disease, in which our methods yield more discoveries than \cite{ghostknockoffs}. We note that existing work in the genetics literature has implemented variable selection methods based on penalized regression with summary statistics, e.g., \cite{lassosum} and \cite{susie_rss}. However, none of these provide any guarantee of FDR control. In fact, as we note in the main text, these methods can be leveraged in our approach to create knockoffs versions that do control the FDR. 

\subsection{Code availability and reproducibility}
The software and example code that reproduce the results presented in this paper can be found at \url{https://github.com/biona001/ghostknockoff-gwas-reproducibility/tree/main/chen_et_al}. Simulation results in Section \ref{simulation_partial}, Section \ref{subsubsec:full_ghost_indep} and Section \ref{subsubsec:full_ghost_auto} can be exactly reproduced. Due to data accessibility issue, we only provide code without real data for Section \ref{subsubsec:full_ghost_real} and Section \ref{section:gwas}.

\section{Model-X Knockoffs and GhostKnockoffs} \label{section:model-x}

To begin with, we define the controlled variable selection problem and give a brief review of model-X knockoffs and GhostKnockoffs. For a more detailed exposition, we refer readers to \cite{model-x}, \cite{fixed-x}, and \cite{ghostknockoffs}. In the following, we use boldface letters for vectors and matrices.\footnote{As an exception, we use $X$, $\widetilde{X}$, and $Y$ to represent generic covariates, their knockoffs, and the response.} We use $\bX_j\in \bbR^{n}$ and $\bx_i\in \bbR^{p}$ to respectively represent the $j$th column and $i$th row of the covariate matrix $\bX$. 

\subsection{Problem statement} \label{sub section:intro-ps}

Given covariates $X \in \bbR^{p}$ and a response $Y\in \bbR$, we are interested in understanding which variables influence $Y$. We formulate this selection problem as testing the \textit{conditional independence hypotheses} $\mathcal{H}_0^j:X_j\indep Y\mid X_{-j}$ for $1\le j \le p$, where $X_{-j}$ is a shorthand for all the variables except the $j$th; that is $X_{-j} =\{X_1,...,X_{j-1},X_{j+1},...,X_n\}$. In words, we should reject $\mathcal{H}_0^j$ if we believe that  $X_j$ can help better predict the outcome than if we only had available the values of all the other variables. Put differently, $X_j$ has information about $Y$ which cannot be subsumed by the information contained in all the other variables.  By conditioning on $X_{-j}$, these hypothesis tests aim to weed out variables whose relationship to $Y$ is driven by residual correlations with other covariates. 

Let $\mathcal{H}_0 \subset [p]$ be the set of indices for which the null conditional independence hypothesis $\mathcal{H}_0^j$ is true, and let ${\mathcal{S}} \subset [p]$ be the set of indices of the hypotheses rejected by a selection procedure. The false discovery rate (FDR) is the expected fraction of false positives among the selected, defined as 
$$\text{FDR}:=\mathbb{E}\Ls \frac{|{\mathcal{S}}\cap\mathcal{H}_0|}{|\hat{\mathcal{S}}|}\Rs$$
with the convention that $0/0 = 0$. 
Our goal is to make as many rejections as possible while controlling the FDR below a user-specified level $q$. 

In this paper, we consider the setting in which, instead of observing i.i.d.~samples from the distribution of ($X,Y$), we only have some summary statistics of the i.i.d.~samples. In particular, we will show how one can, quite remarkably, perform tests of conditional independence when we do not directly observe the i.i.d.~samples. Throughout this paper, we assume that $X\sim \mathcal{N}(\bzero,\bSigma)$ where $\bSigma$ is known (or, in practice, can be estimated).

\subsection{Model-X knockoffs} \label{sub section:intro-kf}
\subsubsection{The procedure}

Suppose we observe $n$ i.i.d.~samples $(X_i,Y_i)$, $1\le i\le n$, arranged in a data matrix $\bX \in \bbR^{n\times p}$ and response vector $\bY \in \bbR^{n}$. In the model-X knockoffs framework \cite{model-x}, we assume we know the distribution $P_X$ of the covariates $X$ while having no knowledge of the conditional distribution $Y\mid X$. The model-X approach is well-suited to genetic applications where reference panels may be available to estimate $P_X$ or where we have good models of linkage disequilibrium.

To implement model-X knockoffs, we first generate a matrix $\widetilde{\bX} 
\in \mathbb{R}^{n\times p}$ of knockoffs such that the following two conditions hold: 
\begin{align} \textbf{(Exchangeability):}& \; (\bX_j, \widetilde{\bX}_j, \bX_{-j}, \widetilde{\bX}_{-j})\stackrel{d}{=} (\widetilde{\bX}_j, \bX_j, \bX_{-j}, \widetilde{\bX}_{-j}),\; \forall \; 1 \le j \le p \label{exchangeability}\\ 
\textbf{(Conditional independence):}& \;\widetilde{\bX} \indep \bY\mid \bX. \label{Conditional independence}\end{align} 
Roughly, the first says that we cannot distinguish between $[\bX\;\widetilde{\bX}]$ and $[\bX\;\widetilde{\bX}]_{\text{swap}(j)}$, where $[\bX\;\widetilde{\bX}]_{\text{swap}(j)}$ is obtained from $[\bX\;\widetilde{\bX}]$ by swapping the $j$th and $(j+p)$th columns. The second condition implies that $\widetilde{\bX}$ does not provide any new information about $Y$ conditional on $X$ and is guaranteed if $\widetilde{\bX}$ is constructed without looking at $\bY$. If these properties hold, it can be shown that $\bX_j$ and $\widetilde{\bX}_j$ are indistinguishable conditional on $\bY$ for each $j \in \mathcal{H}_0$.

Next, we define feature importance statistics $\bW = w([\bX, \widetilde{\bX}], \bY) \in \mathbb{R}^p$ to be any function of $\bX$, $\widetilde{\bX}$ and $\bY$ such that a flip-sign property holds; namely, switching a column $\bX_j$ with its knockoff $\widetilde{\bX}_j$ flips the sign of  the $j$th component of the output; formally, $w_j([\bX, \widetilde{\bX}]_{\text{swap}(j)}, \bY) = -  w_j([\bX, \widetilde{\bX}], \bY)$. 
Common choices include $W_j = |\bX_j^\top \bY| - |\widetilde{\bX}_j^\top \bY|$ (marginal correlation difference statistic) and $W_j = |\hat{\beta}_j(\lambda_{\text{CV}})| - |\hat{\beta}_{j+p}(\lambda_{\text{CV}})|$ (Lasso coefficient difference statistic), where $\hat{\bbeta}(\lambda_{\text{CV}})$ is the solution to the Lasso problem 
$$\argmin_{\bbeta\in\mathbb{R}^{2p}} \frac{1}{2}||\bY-[\bX\;\widetilde{\bX}]\bbeta||_2^2 + \lambda_{\text{CV}}||\mathbf{\bbeta}||_1,$$ 
and $\lambda_{\text{CV}}$ is usually chosen by cross-validation.

Finally, the knockoff filter selects the variables $\mathcal{S}=\{j:W_j \ge T\}$, where 
\begin{equation} \label{eqn:knockoffs_threshold} T = \text{min}\left\{t \in \mathcal{W}: \frac{1 +\#\{j:\; W_j \le -t \}}{\#\{j:\; W_j \ge t\}\vee 1}\le q\right\}.\end{equation} 
Here, $\mathcal{W}=\{|W_j|:j=1,...,p\} \backslash \{0\}$, and $T=+\infty$ if $\mathcal{W}$ is empty. Intuitively, the threshold $T$ is chosen to be the most liberal one such that an estimate of FDP is bounded by $q$. \cite{model-x} showed that this procedure controls the FDR of the conditional testing problem at level $q$.

\subsubsection{Gaussian knockoff sampler} \label{subsub section:gaussian-kf}

Under the assumption that the rows of the data matrix $\bX$ are i.i.d.~from the Gaussian distribution $\mathcal{N}(\bzero,\bSigma)$, we can generate a knockoff vector $\widetilde{\bx}_i$ for each row $\bx_i$ of the data matrix $\bX$ by sampling $\widetilde{\bx}_i\sim{\mathcal{N}(\bP^\top \bx_i,\bV)}$ independently across rows, where $\bP = \bI - \bSigma^{-1}\bD$, $\bV = 2\,\bD - \bD\bSigma^{-1}\bD$, $\bD=\text{diag}\{\bs\}$, and $\bs\in \mathbb{R}^p$ is a vector of free parameters usually obtained by solving a convex optimization problem that depends on $\bSigma$ \citep{model-x}. See Appendix \ref{app:s} for details of computing $\bs$. Concatenating all the knockoff vectors then gives a valid matrix $\widetilde{\bX} \in \bbR^{n\times p}$ of knockoffs. In matrix form, the construction above is 
\begin{equation}\label{CRT_all}
    \widetilde{\bX}=\bX \bP+\bE\bV^{1/2},
\end{equation} where $\bE$ is an $n$ by $p$ matrix with i.i.d.~standard Gaussian entries, independent of $\bX$ and $\bY$. For later reference, we summarize the Gaussian knockoff sampler in Algorithm \ref{alg:gaussiankf} and denote it as $\cG$.

\begin{algorithm}[H]
\caption{Gaussian Knockoff Sampler $\cG$}\label{alg:gaussiankf}
\begin{algorithmic}[1]
\STATE \textbf{Input}: $\bX$ and $\bSigma$.
\STATE Compute $\bs$ by solving a convex optimization problem as defined in \eqref{eqn:sdp}.
\STATE Compute $\bD=\text{diag}\{\bs\}, \bP = \bI - \bSigma^{-1}\bD$, and $\bV = 2\,\bD - \bD\bSigma^{-1}\bD$.
\STATE Simulate $\bE \in \bbR^{n \times p}$ whose entries are i.i.d.~standard Gaussian variables.
\STATE \textbf{Output}: $\widetilde{\bX}=\bX \bP+\bE\bV^{1/2}.$
\end{algorithmic}
\end{algorithm}

\subsection{GhostKnockoffs with marginal correlation difference statistic} \label{sub section:intro-gk_marginal}

The original model-X knockoffs procedure relies on having access to the covariates and responses from all data points, i.e., the matrix of covariates $\bX$ and the response vector $\bY$. Henceforth, we call these \textit{individual-level data}. In many application scenarios, however, individual-level data are not available due to privacy concerns. Instead, we only have access to some summary statistics of $\bX$ and $\bY$, e.g., the empirical covariance matrix of the covariaties and the empirical covariance between each covariate and the response.

\cite{ghostknockoffs} proposed GhostKnockoffs, which implements the knockoffs procedure with marginal correlation difference statistic when only $\bX^\top \bY$ and $||\bY||_2^2$ are available. The key idea of \cite{ghostknockoffs} is to sample the knockoff $Z$-score $\widetilde{\bZ}_s$ from $\bX^\top \bY$ and $||\bY||_2^2$ directly, in a way such that \begin{equation} \label{eqn:Z_s_required} \widetilde{\bZ}_s\mid \bX, \bY \stackrel{d}{=}\widetilde{\bX}^\top \bY\mid \bX, \bY, \end{equation} where $\widetilde{\bX}=\cG(\bX,\bSigma)$ is the knockoff matrix generated by the Gaussian knockoff sampler (Algorithm \ref{alg:gaussiankf}). If we use $\bW=\bZ_s-\widetilde{\bZ}_s$ (where $\bZ_s=\bX^\top \bY$) as the feature importance statistic and run the knockoff filter, the resulting rejection set will have the same distribution as that of the knockoffs procedure with marginal correlation difference statistic. Therefore, the two procedures are statistically identical. In particular, they both control the FDR.

Specifically, \cite{ghostknockoffs} showed that for $\bP$ and $\bV$ computed in step 3 of Algorithm \ref{alg:gaussiankf},  \begin{equation} \label{eqn:Z_s} \widetilde{\bZ}_s = \bP^\top \bX^\top \bY + ||\bY||_2\bZ  \ \text{where} \ \bZ \sim \mathcal{N}(\bzero,\bV) \;\text{is independent of} \;\bX \;\text{and} \;\bY \end{equation} satisfies \eqref{eqn:Z_s_required} as detailed in Appendix \ref{ghostknockoffs_proof}. All this is summarized in Algorithm \ref{alg:ghostkf}. In the following sections, we refer to Algorithm \ref{alg:ghostkf} as GhostKnockoffs with marginal correlation difference statistic (\textit{GK-marginal}).\\

\begin{algorithm}[H]
\caption{GhostKnockoffs with Marginal Correlation Difference Statistic (\textit{GK-marginal})}\label{alg:ghostkf}
\begin{algorithmic}[1]
\STATE \textbf{Input}: $\bX^\top \bY$, $||\bY||_2^2$, and $\bSigma$.
\STATE Compute $\bs$, $\bP$, and $\bV$ as in Algorithm \ref{alg:gaussiankf}.
\STATE Compute the feature importance statistics $\bW = \abs{\bZ_s} - \abs{\widetilde{\bZ}_s}$, where $\widetilde{\bZ}_s$ is generated according to \eqref{eqn:Z_s}.
\STATE Input $\bW$ into the knockoffs selection procedure.
\STATE \textbf{Output}: Knockoffs selection set.
\end{algorithmic}
\end{algorithm}

\section{GhostKnockoffs with Penalized Regression: Known Empirical Covariance} \label{section:partialghost} 
\subsection{Setting} \label{full_ghost_setting}

As we have just seen, GhostKnockoffs-marginal gives a way to test conditional hypotheses while maintaining FDR control when only the summary statistics $\mathbf{\bX}^\top \bY$ and $\norm{\bY}_2^2$ are available to the analyst. Now, we consider the setting in which we have knowledge of the empirical covariance matrix $\bX^\top \bX$ and the sample size $n$, in addition to $\mathbf{\bX}^\top \bY$ and $\norm{\bY}_2^2$. These quantities only reveal sample averages of relevant quantities, as opposed to all the individual-level information.

In this section, we propose a variable selection method that utilizes only $\bX^\top \bX$, $\mathbf{\bX}^\top \bY$, $\norm{\bY}_2^2$, and $n$. Our method achieves FDR control and  power comparable to the knockoffs procedure with the cross-validated Lasso coefficient difference statistic defined in Section \ref{section:model-x}. This is interesting because the latter usually outperforms GhostKnockoffs with the marginal correlation difference statistic by a significant margin. Notably, for a fixed tuning parameter $\lambda$, we show that our procedure is equivalent to a corresponding knockoffs method using the Lasso coefficient difference statistic with the same penalty level $\lambda$.

\subsection{GhostKnockoffs with the Lasso}  \label{procedure2}
Recall that in the knockoffs procedure with the Lasso coefficient difference statistic, we solve the optimization problem  
\begin{equation} \label{Lasso}
\hat{\bm{\beta}}(\lambda) \in \argmin_{\bm{\beta}\in\mathbb{R}^{2p}} \frac{1}{2}||\bY-[\bX\;\widetilde{\bX}]\bm{\beta}||_2^2 + \lambda||\bm{\beta}||_1, \end{equation} where $\widetilde{\bX}=\cG(\bX,\bSigma).$
We then define the \textit{Lasso coefficient difference} feature importance statistics by $W_j=\abs{\hat{\beta}_j(\lambda)}-\abs{\hat{\beta}_{j+p}(\lambda)}$ for $1\le j\le p$. If we have access to individual-level data, $\lambda$ is usually chosen by cross-validation 
(\cite{model-x} and \cite{lcd_power}).\footnote{In the case that $Y$ is binary, one may think that utilizing (penalized) logistic regression would give much better power than Lasso. In Appendix \ref{app:binary response}, we show that this intuition may not be correct through simulations, even when $Y$ is generated according to a logistic regression model.}

As a first step, we would like to run a statistically equivalent procedure using $\bX^\top \bX, \bX^\top \bY$, $\norm{\bY}_2^2$, and $n$ for a \textit{fixed} $\lambda$. Note that, with $\lambda$ fixed, \eqref{Lasso} depends on the data only through $$\begin{bmatrix}
\bX^\top \bX & \bX^\top \widetilde{\bX} \\
\widetilde{\bX}^\top \bX & \widetilde{\bX}^\top \widetilde{\bX}
\end{bmatrix}$$ and $$\begin{bmatrix} \bX^\top \bY \\
\widetilde{\bX}^\top \bY
\end{bmatrix}.$$ 
Define the Gram matrix of $[\bX, \widetilde{\bX}, \bY]$ 
$$\mathcal{T}(\bX,\widetilde{\bX},\bY) =  [\bX, \widetilde{\bX}, \bY]^\top \, [\bX, \widetilde{\bX}, \bY]. 
$$ 
The Gram matrix can of course be equivalently reconstructed from $(\norm{\bY}_2^2,\bX^\top \bY,\widetilde{\bX}^\top \bY,\bX^\top \bX,\widetilde{\bX}^\top \bX,\widetilde{\bX}^\top \widetilde{\bX})$. The main idea is to sample from the joint distribution of $\mathcal{T}(\bX,\widetilde{\bX},\bY)$ using the Gram matrix of $[\bX, \bY]$ only. Based on this, we can then generate the solution to the Lasso problem \eqref{Lasso} (in distribution) for a fixed $\lambda$.\footnote{Careful readers may realize that the solution of the Lasso problem does not depend on $\norm{\bY}_2^2$. Here we include $\norm{\bY}_2^2$ as an input of to be able to make a more general statement later that goes beyond the Lasso.} This is achieved via the following Proposition \ref{prop:partial_ghost}, which says in words that if we generate `fake' data matrices $\widecheck{\bX}$ and $\widecheck{\bY}$ that lead to the same Gram matrix as that of $\bX$ and $\bY$, then the distribution of $\mathcal{T}$ remains unchanged if we replace the original data matrices by the fake data matrices.

\begin{proposition} \label{prop:partial_ghost}
Suppose $\widecheck{\bX} \in \mathbb{R}^{n \times p}$ and $\widecheck{\bY} \in \mathbb{R}^{n}$ are constructed such that $[\widecheck{\bX}\ \widecheck{\bY}]^\top [\widecheck{\bX}\ \widecheck{\bY}] = [\bX\ \bY]^\top [\bX\ \bY]$. Setting $\widetilde{\bX} = \cG({\bX},\bSigma)$ and $\widetilde{\widecheck{\bX}} = \cG(\widecheck{\bX},\bSigma)$ as the outputs of Algorithm \ref{alg:gaussiankf},\footnote{Note that  $\widecheck{\bX}$ may not be a data matrix with i.i.d.~rows and covariance matrix $\bSigma$ and we should call $\widetilde{\widecheck{\bX}}$ the pseudo-Gaussian knockoff data matrix.} we have $$\mathcal{T}(\bX,\widetilde{\bX},\bY)\mid \bX,\bY  \stackrel{d}{=} \mathcal{T}(\widecheck{\bX},\widetilde{\widecheck{\bX}},\widecheck{\bY})\mid \bX,\bY.$$
\end{proposition}
Proof of Proposition \ref{prop:partial_ghost} is provided in Appendix \ref{proof:partial_ghost}. Specifically, Proposition \ref{prop:partial_ghost} suggests that summary statistics ($\bX^\top \bX, \bX^\top \bY, ||\bY||_2^2$, $\bSigma$) are sufficient for sampling the Gram matrix $\mathcal{T}(\bX,\widetilde{\bX},\bY)$.

% A couple remarks are in order.
% \begin{itemize}
%     \item 
%     \item Naturally, Proposition \ref{prop:partial_ghost} allows us to sample from the distribution of $$\widetilde{\bX}^\top \bY \mid\bX, \bY$$ based on the available summary statistics. With that, we will be able to run GK-marginal as in Section \ref{sub section:intro-gk_marginal}.
% \end{itemize}

\begin{algorithm}[H]
\caption{GhostKnockoffs with Penalized Regression: Known Empirical Covariance}\label{alg:partialghost}
\begin{algorithmic}[1]
\STATE \textbf{Input}: $\bX^\top \bX, \bX^\top \bY, ||\bY||_2^2$, $\bSigma$, and $n$. \vspace{1mm}
\STATE Find $\widecheck{\bX}$ and $\widecheck{\bY}$ such that $[\widecheck{\bX}\ \widecheck{\bY}]^\top [\widecheck{\bX}\ \widecheck{\bY}] = [\bX\ \bY]^\top [\bX\ \bY]$ by eigen-decomposition or Cholesky decomposition. 
\STATE Generate $\widetilde{\widecheck{\bX}} = \cG(\widecheck{\bX},\bSigma)$ via Algorithm \ref{alg:gaussiankf}.
\STATE Run the standard knockoffs procedure (at level $q$) with the Lasso coefficient difference statistic on $\widecheck{\bX}$ and $\widetilde{\widecheck{\bX}}$ for a fixed penalty level $\lambda$ or use the methods from Sections \ref{GK-sqrtlasso} and \ref{GK-lassomax}. 
\STATE \textbf{Output}: Knockoffs selection set.
\end{algorithmic}
\end{algorithm}

We are now able to write down a procedure, namely, Algorithm \ref{alg:partialghost}, which is statistically equivalent to the corresponding individual-level knockoffs procedure using the Lasso coefficient difference statistic (or any statistic defined in Sections \ref{GK-sqrtlasso} and \ref{GK-lassomax}). In step 2, $\widecheck{\bX}$ and $\widecheck{\bY}$ can be obtained by performing the eigen-decomposition or Cholesky decomposition of $[\bX\ \bY]^\top [\bX\ \bY]$. Brief procedures to construct $\widecheck{\bX}$ and $\widecheck{\bY}$ via eigen-decomposition are provided in Appendix \ref{construction}. 
All we need to do is to run the knockoffs procedure with $\widecheck{\bX}$ and $\widetilde{\widecheck{\bX}}$ in lieu of ${\bX}$ and $\widetilde{{\bX}}$. We say that the procedure is equivalent since the rejection sets have the same distribution. In particular, this proves that Algorithm \ref{alg:partialghost} controls the FDR.

\begin{corollary} \label{corr:partial_ghost}
Consider a knockoffs feature importance statistic $\bW=\mathbf{f}(\mathcal{T}(\bX,\widetilde{\bX},\bY),\bU) \in \mathbb{R}^p$, which is a deterministic function of $\mathcal{T}(\bX,\widetilde{\bX},\bY)$ and an independent random variable $\bU$. Define $\widehat{\bW}=\mathbf{f}(\mathcal{T}(\widecheck{\bX},\widetilde{\widecheck{\bX}},\bY),\bU)$. Let $\mathcal{S}_1$ (resp.~$\mathcal{S}_2$) be the rejection set obtained from applying the knockoffs filter on $\bW$ (resp.~$\widehat{\bW}$). 
Then $\mathcal{S}_1 \mid \bX, \bY \stackrel{d}{=} \mathcal{S}_2 \mid \bX, \bY$. Thus, if $\bW$ obeys the flip-sign property, both procedures have equal FDR at most equal to $q$. 
\end{corollary} 

\begin{proof}
Proposition \ref{prop:partial_ghost} gives $\bW \mid \bX, \bY \stackrel{d}{=} \widehat{\bW} \mid \bX, \bY$. Since the selection set is uniquely determined by the values of $\bW$ (or $\widehat{\bW}$), it follows that $\mathcal{S}_1 \mid \bX, \bY \stackrel{d}{=} \mathcal{S}_2 \mid \bX, \bY$. Therefore, the procedures have the same FDR.
\end{proof}

We can easily adapt the method above to accommodate other types of regularization, such as Ridge regression and Elastic Net.

\subsection{GhostKnockoffs with the square-root Lasso} \label{GK-sqrtlasso}
In Section \ref{procedure2}, we assumed that the tuning parameter $\lambda$ in \eqref{Lasso} is fixed. In practice, one may choose the penalty level using information from the Gram matrix of $[\bX, \bY]$, and the sample size $n$. Since individual-level data is not available, we are unable to use data-splitting approaches such as cross-validation. 

An alternative way to define feature importance is to use the square-root Lasso \citep{sqrtlasso}, for which the choice of a reasonable tuning parameter is convenient. The square-root Lasso applied to the knockoffs setting solves \begin{equation} \label{sqrtlasso}
\hat{\bbeta}(\lambda) \in \argmin_{\bbeta\in\mathbb{R}^{2p}} ||\bY-[\bX\;\widetilde{\bX}]\bbeta||_2 + \lambda||\bbeta||_1,
\end{equation}
and a good choice of $\lambda$ is given by \begin{equation} \label{lambda}
\lambda=\kappa \cdot \mathbb{E} 
\Ls\frac{\norm{[\bX\;\widetilde{\bX}]^\top \bepsilon}_\infty}{\norm{\bepsilon}_2}|\bX,\widetilde{\bX}\Rs,
\end{equation} where $\bepsilon \sim \mathcal{N}(\bzero,\bI_n)$ and $\kappa$ is a unitless hyperparameter \citep{sqrtlasso2}. This value is a scalar multiple of the expected value of the minimal penalty level required such that all the coefficients are shrunk to zero under the global null model. The square-root Lasso has the benefit that the value of the hyperparameter does not depend on the details of the distribution of $Y$ conditional on $X$. We also found that the performance of our procedure does not depend very sensitively on the choice of $\kappa$. In our data examples, we take $\kappa = 0.3$.

In the setting where we only know about values of the summary statistics, we simply replace ($\bX,\widetilde{\bX},\bY$) by ($\widecheck{\bX}, \widetilde{\widecheck{\bX}}, \widecheck{\bY})$ in \eqref{sqrtlasso}. Further, we note that for any orthogonal matrix $\bQ$,  \begin{align*}
([\bX\;\widetilde{\bX}]^\top \bQ^\top \bepsilon,\bepsilon^\top \bepsilon)\mid \bX, \widetilde{\bX}  &\stackrel{d}{=} ([\bX\;\widetilde{\bX}]^\top \bQ^\top \bepsilon,\bepsilon^\top \bQ\bQ^\top \bepsilon)\mid \bX, \widetilde{\bX} 
 \\ &\stackrel{d}{=} ([\bX\;\widetilde{\bX}]^\top \bepsilon,\bepsilon^\top \bepsilon)\mid \bX, \widetilde{\bX}, 
\end{align*} where the second equality follows from $\bQ^\top \bepsilon \stackrel{d}{=} \bepsilon$. Therefore, the value of the hyperparameter in \eqref{lambda} remains unchanged if we multiply $[
\bX \; \widetilde{\bX}]$ by $\bQ$ on the left. This implies that \eqref{lambda} is a deterministic function of $[\bX\;\widetilde{\bX}]^\top [\bX\;\widetilde{\bX}]$. Hence, the feature importance statistic is a function of $\mathcal{T}(\bX,\widetilde{\bX},\bY)$. Following Corollary \ref{corr:partial_ghost}, we can apply the knockoffs procedure with the square-root Lasso and matrices 
$(\widecheck{\bX}, \widetilde{\widecheck{\bX}})$ in lieu of $({\bX}, \widetilde{{\bX}})$. Upon choosing 
\begin{equation} \label{lambda_used}
\lambda=\kappa \; \mathbb{E} 
\Ls\frac{\norm{[
\widecheck{\bX}\;\widetilde{\widecheck{\bX}}]^\top \bepsilon}_\infty}{\norm{\bepsilon}_2}\mid \widecheck{\bX},\widetilde{\widecheck{\bX}}\Rs,
\end{equation} 
we get a procedure, which is statistically indistinguishable from that we would get if we were performing all the same steps with $\bX$ and $\widetilde{\bX}$. (In practice, we compute the value in \eqref{lambda_used} via Monte Carlo simulation.)  In the sequel,  we call the resulting procedure summary statistics GhostKnockoffs with square-root Lasso importance statistic (\textit{GK-sqrtlasso}). Note that GK-sqrtlasso controls the FDR as the flip-sign property of the feature importance statistic holds. This is because swapping a variable with its knockoff does not change the value of the hyperparameter. Therefore, by Corollary \ref{corr:partial_ghost}, applying the knockoff filter to the square-root Lasso feature importance statistics yields FDR control. 

\subsection{GhostKnockoffs with the Lasso-max} \label{GK-lassomax}
In the standard fixed-X knockoffs setting, cross-validation is also not feasible, since doing so would violate the sufficiency condition required for the feature importance statistics. As one possible alternative, \cite{fixed-x} considered using as the feature importance statistic the value of $\lambda$ on the Lasso path at which feature $X_j$ first enters the model. Formally, they define the feature importance statistic 
$$W_j = \text{sup}\{\lambda:\hat{\beta}_j(\lambda) \neq 0\}-\text{sup}\{\lambda:\hat{\beta}_{j+p}(\lambda) \neq 0\},$$
where $\hat{\bbeta}(\lambda)$ is as in \eqref{Lasso}. We call this statistic the Lasso-max statistic. Intuitively, a larger penalty level is required to shrink an important feature to zero, so we should expect $W_j$ to be large and positive for non-nulls.

By Corollary \ref{corr:partial_ghost}, with the Lasso-max statistic Algorithm \ref{alg:partialghost} produces a rejection set that has the same distribution as the rejection set obtained from the corresponding individual-data-based knockoffs procedure. We call this summary-statistic-based procedure GhostKnockoffs with Lasso-max statistic (\textit{GK-lassomax}).

We remark that choices of other tuning parameters and feature importance statistics are also possible. For instance, we may choose $\lambda$ to minimize the Stein's unbiased risk estimate (SURE) associated with \eqref{Lasso}. We shall however focus on the two approaches we have described.

\subsection{Numerical simulations} \label{simulation_partial}

We consider a variety of simulation settings in which we compare the performance of the proposed GhostKnockoffs with square-root Lasso and Lasso-max statistics (GK-sqrtlasso and GK-lassomax, defined in Sections \ref{GK-sqrtlasso} and \ref{GK-lassomax}), GhostKnockoffs with marginal correlation difference statistic (GK-marginal, defined in Section \ref{section:model-x}), and the knockoffs procedure with (cross-validated) Lasso coefficient difference statistic with individual-level data (KF-lassocv). Note that the first three are statistically equivalent to the corresponding knockoffs procedures with individual-level data.

\subsubsection{Independent features} \label{simulation_partial_indep}

In the first set of simulations (Figure \ref{fig:partial_ghost_indep}), we generate random samples $\bx_i \stackrel{iid} {\sim} \mathcal{N}(\bzero,\bI_p)$ and $Y_i = \bbeta^\top \bx_i + \sqrt{n}\epsilon_i$, where $\epsilon_i \stackrel{iid}{\sim} \mathcal{N}(0,1)$ for $i 
\in \{1,2,...,n\}$.\footnote{The simulation setting is designed in a way that the signal-to-noise ratio has the same scale as $n$ varies.} We consider three settings of varying dimensionality measured by the ratio $p/n$: $(n,p) \in \{(600,200), (400,400), (200, 600)\}$. In each of the three settings, we create a sparse vector $\bbeta$ by selecting 30 coordinates to be non-zero uniformly at random. The signs of these non-zero coordinates are assigned to be either positive or negative with equal probability. We vary the signal amplitudes such that we explore a wide power range below. For the square-root Lasso, we average over 200 Monte Carlo samples to calculate $$\lambda=\kappa \cdot \mathbb{E} 
\big[\frac{\norm{[\bX\;\widetilde{\bX}]^\top \bepsilon}_\infty}{\norm{\bepsilon}_2}\mid \bX,\widetilde{\bX}\big].$$ The target FDR is 20\%. Each point on the curves represents the average of the results from 200 replications.

\begin{figure}[htbp] 
  \centering
  \includegraphics[width=0.8\textwidth]{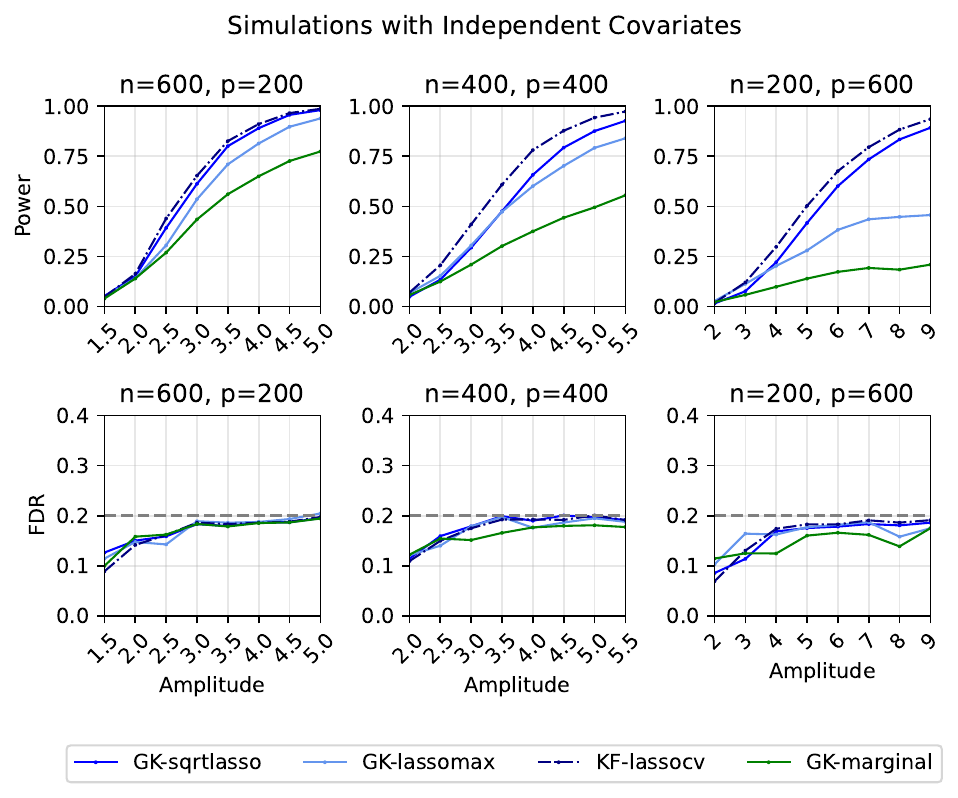}
  \caption{Power and FDR plots for independent features and a Gaussian linear model with varying dimensions. Each point is an average over 200 replications.}
  \label{fig:partial_ghost_indep}
\end{figure}

We observe that GK-sqrtlasso and GK-lassomax generally demonstrate greater power than GK-marginal. This enhanced performance is not surprising, as GK-sqrtlasso and GK-lassomax (1) have access to additional information via $\bX^\top \bX$, and (2) employing a joint modeling algorithm such as Lasso generally provides a better assessment of variable importance for understanding conditional (in)dependence since such a model explicitly adjusts for the effects from all the other variables. We also note the presence of power gaps between GK-lassocv and GK-sqrtlasso/GK-lassomax, likely due to the fact that we are unable to perform cross-validation without individual-level data. All methods control the FDR at the desired level.

\subsubsection{AR(1) features} \label{simulation_partial_auto}
In the second set of simulations (Figures \ref{fig:partial_ghost_auto}), we generate $\bx_i \stackrel{iid} {\sim} \mathcal{N}(\bzero,\bSigma_{\rho})$ for $i \in \{1,2,...,n\}$, where $\Ls\bSigma_{\rho}\Rs_{s,t} = \rho^{|s-t|}$ for $1 \le s,t \le p$. As before, we generate $Y_i = \bbeta^\top \bx_i + \sqrt{n}\epsilon_i$, where $\epsilon_i \stackrel{iid}{\sim} N(0,1)$ for $i 
\in \{1,2,...,n\}$. We consider the same three $(n,p)$ combinations. In each of the three cases, we create a sparse vector $\bbeta$ exactly as before, except that we fix the signal amplitudes to 4, 4, and 7 respectively to explore a wide power range. We vary $\rho$ in $\{0,0.1,0.2,...,0.8\}$ The target FDR is set to be 20\%. Each point represents the average of the results from 200 replications.

\begin{figure}[htbp]
  \centering
  \includegraphics[width=0.8\textwidth]{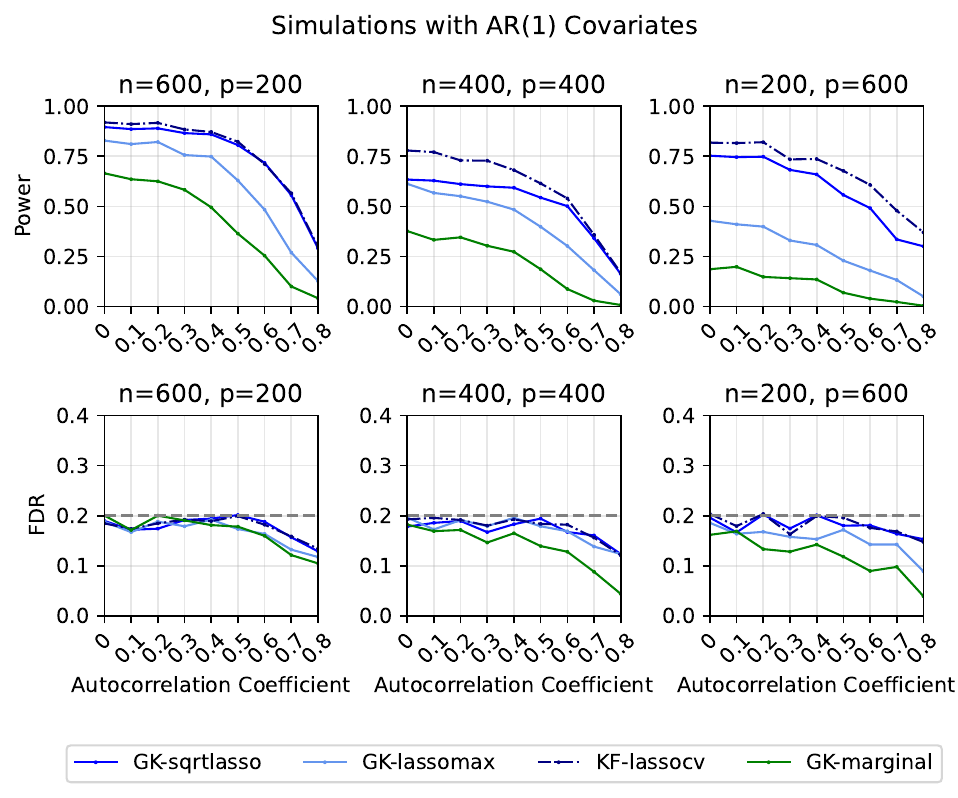}
   \caption{Power and FDR plots for AR(1) features and a Gaussian linear model with varying dimensions. Each point is an average over 200 replications.}
   \label{fig:partial_ghost_auto}
\end{figure}

Again, we observe that GK-sqrtlasso and GK-lassomax generally have greater power than GK-marginal. All methods have (almost) decreasing power as the autocorrelation coefficient increases, since it becomes harder to separate true signals from null variables that are correlated with them. All methods control the FDR at the desired level.

\section{GhostKnockoffs with Penalized Regression: Missing Empirical Covariance} \label{section:fullghost}

\subsection{Setting} \label{sec:pseudo-lasso-setting}

Thus far, we have discussed how incorporating the additional information from $\bX^\top \bX$ and $n$ could enhance our ability to detect significant features. However, in applications such as genetics, $\bX^\top \bX$ may not be  available. In this section, we propose alternative procedures when the scientist only knows about $\bX^\top \bY$,  $\norm{\bY}^2$ and the sample size $n$. As before, we assume that $X\sim \mathcal{N}(\bzero,\bSigma)$, where the covariance matrix $\bSigma$ is known (or can be estimated from other data sources).

\subsection{GhostKnockoffs with pseudo-lasso} \label{sec:pseudo-lasso}

% In Section \ref{alg:fullghost}, we saw that Lasso-style statistics give better power when employed in the knockoffs framework in all the examples we considered. 

The idea of our method is to modify the Lasso objective function so that it can be constructed from the available summary statistics. It turns out that the solution of our modified objective function is proportional to that of the scout procedure (with known precision matrix) proposed by \cite{scout}. We will see through simulation studies that our procedure improves the power of the original GhostKnockoffs method of \citep{ghostknockoffs} while maintaining FDR control.

\subsubsection{The procedure}
Recall that in the knockoffs procedure with the Lasso statistic, we solve the following optimization problem: $$\hat{\bbeta}(\lambda) = \argmin_{\bbeta\in\mathbb{R}^{2p}} \frac{1}{2n} \bbeta^\top  \begin{bmatrix} \bX^\top \bX & \bX^\top \widetilde{\bX} \\
 \widetilde{\bX}^\top \bX & \widetilde{\bX}^\top \widetilde{\bX}
\end{bmatrix}\bbeta - \frac{1}{n} \bbeta^\top \begin{bmatrix}
\bX^\top \bY\\
\widetilde{\bX}^\top \bY
\end{bmatrix} + \lambda||\bbeta||_1.$$ 

To mimic the form of the loss function when we do not observe the empirical covariance of the features, we may want to substitute them with their population version: i.e.~we swap $\bX^\top \bX/n$ and $\widetilde{\bX}^\top \widetilde{\bX}/n$ with $\bSigma$ and $\bX^\top \widetilde{\bX}/n$ with $\bSigma - \bD$. As usual, $\bD=\text{diag}\{\bs\}$ is obtained by solving the convex optimization problem \eqref{eqn:sdp}. In the language of fixed-X knockoffs \citep{fixed-x}, this is equivalent to regarding $\widetilde{\bX}$ as a fixed-X knockoff of $\bX$ and replacing $\bX^\top \bX/n$ by $\bSigma$.\footnote{We remark that similar objective functions have been used in, for example, \cite{lassosum} and \cite{susie_rss}.}  This yields Algorithm \ref{alg:fullghost}.

\begin{algorithm}[H] 
\caption{GhostKnockoffs with Penalized Regression: Missing Empirical Covariance}\label{alg:fullghost}
\begin{algorithmic}[1]
\STATE \textbf{Input}: $\bX^\top \bY, ||\bY||_2^2, \bSigma$ and $n$. \vspace{1mm}
\STATE Simulate $\bZ\sim \mathcal{N}(\bzero, \bV)$, where $\bV$ is defined as in Algorithm \ref{alg:ghostkf}.
\STATE Solve $\hat{\bbeta}(\lambda) = \argmin_{\bbeta\in\mathbb{R}^{2p}} \frac{1}{2} \bbeta^\top  \begin{bmatrix}
\bSigma & \bSigma-\bD \\
\bSigma-\bD & \bSigma
\end{bmatrix}\bbeta - \frac{1}{n}\bbeta^\top \begin{bmatrix}
\bX^\top \bY \vspace{1mm} \\ 
\bP^\top \bX^\top \bY+\norm{\bY}_2\bZ
\end{bmatrix} + \lambda||\bbeta||_1,$ where $\bD$ and $\bP$ are defined as in Section \ref{subsub section:gaussian-kf} and $\lambda$ is fixed or as chosen in Section \ref{section:full_tuning} \vspace{1mm}
\STATE Run the standard knockoffs procedure (at level $q$) with importance statistic $W_j=\abs{\hat{\beta}_j(\lambda)}-\abs{\hat{\beta}_{j+p}(\lambda)}.$
\STATE \textbf{Output}: Knockoffs selection set.
\end{algorithmic}
\end{algorithm}

We call this procedure GhostKnockoffs with pseudo-lasso statistic (\textit{GK-pseudolasso}). We show below that Algorithm \ref{alg:fullghost} controls the FDR of selections at level $q$. Before doing so, we first state a general proposition that includes GK-marginal as a special case.

\begin{proposition} \label{proposition:full_ghost}
Suppose $\bV$ and $\bP$ are defined as in Algorithm \ref{alg:ghostkf}, $\bZ\sim \mathcal{N}(\bzero, \bV)$ is independent of $\bX$ and $\bY$, and $\widetilde{\bX}=\cG(\bX,\bSigma)$. Consider a knockoffs feature importance statistic $\bW=\mathbf{g}(\norm{\bY}_2^2,\bX^\top \bY,\widetilde{\bX}^\top \bY,\bU) \in \mathbb{R}^p$, which is a deterministic function of $\norm{\bY}_2^2,\bX^\top \bY, \widetilde{\bX}^\top \bY$ and an independent random variable $\bU$. Define $\widehat{\bW}=\mathbf{g}(\norm{\bY}_2^2,\bX^\top \bY,\bP^\top \bX^\top \bY+\norm{\bY}_2\bZ,\bU)$.  Let $\mathcal{S}_1$ (resp. $\mathcal{S}_2$) be the rejection set obtained from applying the knockoffs filter on $\bW$ (resp. $\widehat{\bW}$). Then $\mathcal{S}_1 \mid \bX, \bY \stackrel{d}{=} \mathcal{S}_2 \mid \bX, \bY$. Thus, if $\bW$ obeys the flip-sign property,  both procedures have equal FDR at most equal to $q$.
\end{proposition}

\begin{proof}
In Appendix \ref{ghostknockoffs_proof}, we prove that $$\widetilde{\bX}^\top \bY\mid\bX,\bY \stackrel{d}{=} \bP^\top \bX^\top \bY + ||\bY||_2\bZ\mid \bX,\bY.$$ As a result, $\bW \mid \bX, \bY \stackrel{d}{=} \widehat{\bW} \mid \bX, \bY$. Since the selection set is uniquely determined by the values of $\bW$ (or $\widehat{\bW}$), it follows that $\mathcal{S}_1 \mid \bX, \bY \stackrel{d}{=} \mathcal{S}_2 \mid \bX, \bY$. Therefore, the procedures have the same FDR.
\end{proof}

Set $\lambda$ to be a fixed numerical constant. Consider the feature importance statistics $\bW$ defined by $W_j=\abs{\hat{\beta}_j(\lambda)}-\abs{\hat{\beta}_{j+p}(\lambda)},$ where $\hat{\bbeta}(\lambda)$ is the solution to \begin{equation}
\label{eqn:fullghost}
\argmin_{\bbeta\in\mathbb{R}^{2p}} \frac{1}{2} \bbeta^\top  \begin{bmatrix}
\bSigma & \bSigma-\bD \\
\bSigma-\bD & \bSigma
\end{bmatrix}\bbeta - \frac{1}{n} \bbeta^\top \begin{bmatrix}
\bX^\top \bY \vspace{1mm} \\ 
\widetilde{\bX}^\top \bY
\end{bmatrix} + \lambda||\bbeta||_1, 
\end{equation} and $\widetilde{\bX}=\cG(\bX,\bSigma)$ is the Gaussian knockoff data matrix. The feature importance statistic in Algorithm \ref{alg:fullghost} is thus obtained by replacing $\widetilde{\bX}^\top \bY$ by $\bP^\top \bX^\top \bY+\norm{\bY}_2\bZ$ in \eqref{eqn:fullghost}. Since $\bW$ is determined  by $\norm{\bY}_2^2, \bX^\top \bY$ and $\widetilde{\bX}^\top \bY$, it follows from  Proposition \ref{proposition:full_ghost} that the rejection set of Algorithm \ref{alg:fullghost} has the same distribution as that obtained from running the knockoff filter on $\bW$.

Thus to prove that Algorithm \ref{alg:fullghost} controls the FDR of rejections at level $q$, it suffices to verify the flip-sign property of the feature importance statistic for $\bW$ (see Section \ref{section:model-x}). This is a consequence of the following lemma:

\begin{lemma} \label{lemma:Lasso_lemma}
Consider the problem 
\begin{equation} \label{Lasso_lemma}
\argmin_{\bbeta\in\mathbb{R}^{2p}} \frac{1}{2}\bbeta^\top \bC\bbeta - \bd^\top \bbeta + \lambda||\bbeta||_1 + \gamma\norm{\bbeta}^2_2.
\end{equation}
Let $\bPi_S$ be any permutation matrix which swaps the jth and (j+p)th entries of a 2p-dimensional vector for each $j\in S  \subset \{1,...,p\}$. Assume that $\bC$ is $S$-swap invariant in the sense that $\bPi_S^\top \bC\bPi_S = \bC$. Then $\hat{\bbeta}$ is a solution to \eqref{Lasso_lemma} if and only if $\Pi_S\hat{\bbeta}$ is a solution to the same problem with $\bd$ and $\bPi_S \bd$ swapped. In other words, swapping the entries of $\bd$ has the effect of swapping the corresponding entries of the solution.
\end{lemma}

\begin{proof}
Consider the objective with problem data $\Pi_S \bd$:
$$\frac{1}{2}\bbeta^\top \bC\bbeta - (\bPi_S \bd)^\top \bbeta +\lambda\norm{\bbeta}_1+\gamma\norm{\bbeta}_2^2 = \frac{1}{2}\bbeta^\top \bC\bbeta - \bd^\top \bPi_S^\top \bbeta +\lambda\norm{\bbeta}_1+\gamma\norm{\bbeta}_2^2.$$ Set $\bbeta'=\bPi_S^\top \bbeta$ so that $\bbeta=\bPi_S \bbeta'$. Upon changing variables, the objective takes the form $$\frac{1}{2}(\bbeta')^\top \bPi_S^\top \bC\Pi_S\bbeta'-\bd^\top \bbeta'+\lambda\norm{\bPi_S \bbeta'}_1+\gamma\norm{\bPi_S \bbeta'}_2^2 = \frac{1}{2}(\bbeta')^\top \bC\bbeta'-\bd^\top \bbeta'+\lambda\norm{\bbeta'}_1+\gamma\norm{ \bbeta'}_2^2,$$ where the equality follows because $\bPi_S^\top \bC\bPi_S=\bC$ and because the 1-norm and 2-norm are invariant under permutation. Now, the objective on the right-hand side is the objective with data $\bd$. If $\hat{\bbeta}$ is the solution with data $\bd$, it follows that $\bPi_S \hat{\bbeta}$ is the solution with data $\bPi_S \bd$, and vice versa. This proves the lemma.
\end{proof}

\begin{corollary}
Algorithm \ref{alg:fullghost} with a fixed $\lambda$ controls the FDR of rejections at level $q$.
\end{corollary}

\begin{proof}
It is easy to show that $\begin{bmatrix}
\bSigma & \bSigma-\bD \\
\bSigma-\bD & \bSigma
\end{bmatrix}$ is $S$-swap invariant for any $S\subset \{1,...,p\}$. Taking $$\bC = \begin{bmatrix}
\bSigma & \bSigma-\bD \\
\bSigma-\bD & \bSigma
\end{bmatrix}$$ and $$\bd = \frac{1}{n} \begin{bmatrix}
\bX^\top \bY \vspace{1mm} \\ 
\widetilde{\bX}^\top \bY
\end{bmatrix}$$ in Lemma \ref{lemma:Lasso_lemma}
establishes the flip-sign property of $\bW$ and, therefore, the FDR control of Algorithm \ref{alg:fullghost} for a fixed $\lambda$.
\end{proof}

In practice, to ensure numerical stability, we add a small positive constant multiple of the identity matrix to $$\begin{bmatrix}
\bSigma & \bSigma-\bD \\
\bSigma-\bD & \bSigma
\end{bmatrix}$$ when solving for $\hat{\bbeta}$. This is equivalent to incorporating a small Ridge penalty into the objective function. It is easy to see that the lemma proved above guarantees that this modification does not compromise the FDR control as $$\begin{bmatrix}
\bSigma+c\bI & \bSigma-\bD \\
\bSigma-\bD & \bSigma + c\bI
\end{bmatrix}$$ is also $S$-swap invariant for any $c\in \mathbb{R}$ and any $S\subset \{1,...,p\}.$

\subsubsection{Choice of tuning parameter} \label{section:full_tuning}

Several methods can be used to tune the value of the hyperparameter $\lambda$. We here consider two approaches.

\paragraph{Method 1 (lasso-min)} Pretend a homogeneous Gaussian linear model holds, i.e. $\bY = \bX\bbeta^*+\sigma\bepsilon$ for some $\bbeta^* \in \mathbb{R}^p$, $\sigma > 0$ and $\bepsilon \sim N(\bzero,\bI_n)$. 

Focus on  \eqref{eqn:fullghost} first and imagine that we have a method for computing $\lambda$ that depends on data only through $\norm{\bY}_2^2,\bX^\top \bY$, and $\widetilde{\bX}^\top \bY$. Note that the objective in Algorithm \ref{alg:fullghost} only substitutes $\widetilde{\bX}^\top \bY$ in \eqref{eqn:fullghost} with $\bP^\top \bX^\top \bY+\norm{\bY}_2\bZ$. Therefore, by Proposition \ref{proposition:full_ghost} if we set $\lambda$ via the same functional and work with $\bP^\top \bX^\top \bY+\norm{\bY}_2\bZ$ in lieu of $\widetilde{\bX}^\top \bY$, we shall  achieve FDR control with this data-driven value of the hyperparameter $\lambda$. This holds of course with the proviso that our selection of hyperparameter is symmetric in the sense that it produces  feature importance statistic obeying the flip-sign property. 

To set the tuning parameter $\lambda_0$ in \eqref{eqn:fullghost}, we use the common choice of taking a constant multiple of the expected value of the minimum $\lambda$ value such that $\hat{\bbeta}(\lambda) =\bzero_{2p}$ under the null model $\bY=\sigma\bepsilon$.  By the Karush–Kuhn–Tucker (KKT) conditions \citep{convex_opt}, this results in a tuning parameter of the form $$\lambda_0 = \kappa \cdot \frac{\sigma}{n} \cdot \mathbb{E}[\norm{\begin{bmatrix}
\bX &
\widetilde{\bX}
\end{bmatrix}^\top \bepsilon}_\infty],$$ where $\kappa$ is a hyperparameter between 0 and 1. Since $\begin{bmatrix}
\bX &
\widetilde{\bX}
\end{bmatrix}$ is a data matrix whose rows are iid samples from $$\mathcal{N}\Lp\bzero,\begin{bmatrix}
\bSigma & \bSigma-\bD \\
\bSigma-\bD & \bSigma
\end{bmatrix}\Rp,$$ $\mathbb{E}[\norm{\begin{bmatrix}
\bX &
\widetilde{\bX}
\end{bmatrix}^\top \bepsilon}_\infty]$ is a numerical constant, which can be estimated arbitrarily well via Monte Carlo simulations. We use the approach from \cite{variance-estimation} to give an estimate of $\sigma$, which crucially requires knowing only $\norm{\bY}_2^2,\bX^\top \bY$, and $\widetilde{\bX}^\top \bY$. \cite{variance-estimation} showed that the estimator is consistent and asymptotic normal in the high-dimensional regime. Specifically, in our setting, we estimate $\sigma$ by $$ \widehat{\sigma}_0=\sqrt{\text{max}\Lp\frac{2p+n+1}{n(n+1)}\norm{\bY}_2^2-\frac{1}{n(n+1)}\bY^\top \begin{bmatrix} \bX & \widetilde{\bX} \end{bmatrix}\begin{bmatrix}
\bSigma & \bSigma-\bD \\
\bSigma-\bD & \bSigma
\end{bmatrix}^{-1}\begin{bmatrix} \bX & \widetilde{\bX} \end{bmatrix}^\top \bY,0\Rp}.$$

In sum, a choice for $\lambda$ in Algorithm \ref{alg:fullghost} is this: 
\begin{enumerate}
  \item Approximate $\mathbb{E}[\norm{\bR^\top \bepsilon}_\infty]$ via Monte Carlo simulations, where $\bR \in \mathbb{R}^{n\times 2p}$ has iid $\mathcal{N}\Lp\bzero,\begin{bmatrix}
\bSigma & \bSigma-\bD \\
\bSigma-\bD & \bSigma
\end{bmatrix}\Rp$ rows, $\bepsilon \sim N(\bzero,\bI_n)$ is independent.
  \item Compute  $$\widehat{\sigma}_0=\sqrt{\text{max}\Lp\frac{2p+n+1}{n(n+1)}\norm{\bY}_2^2-\frac{1}{n(n+1)}\begin{bmatrix} \bY^\top \bX & \bY^\top \bX\bP+\norm{\bY}_2\bZ^\top  \end{bmatrix}\begin{bmatrix}
\bSigma & \bSigma-\bD \\
\bSigma-\bD & \bSigma
\end{bmatrix}^{-1}\begin{bmatrix} \bX^\top \bY \\ \bP^\top \bX^\top \bY+\norm{\bY}_2\bZ \end{bmatrix},0\Rp},$$
where $\bZ$ is independent of everything else. 
  \item Output $\lambda \approx \kappa \cdot\frac{\widehat{\sigma}_0}{n}\cdot\mathbb{E}[\norm{\bR^\top \bepsilon}_\infty]$ where the approximation sign $\approx$ reminds us that the expectation is only approximate. 
\end{enumerate}

As in the square-root Lasso case, we observe that the power of our method is not very sensitive to the choice of $\kappa$. We use $\kappa=0.6$ in our simulations below. % We call this method of choosing $\lambda$ the \textit{lasso-min} method. 
In Appendix \ref{app:tuning_parameter}, we provide details of computation of $\lambda$ and prove that Algorithm \ref{alg:fullghost} maintains FDR control with the computed $\lambda$.

\paragraph{Method 2 (pseudo-sum)} An alternative way of choosing $\lambda$ is to adapt the pseudo-summary statistics approach
proposed by \cite{pseudo_summary}. Set $\br= \bX^\top \bY/n$ and $\widetilde{\br}=\bP^\top \br+\norm{\bY}_2\bZ/n$. The main idea of \cite{pseudo_summary} is to generate training summary statistics $\br_t$ and validation summary statistics $\br_v$ from $\br$ and $\widetilde{\br}$ based on the training and validation sample sizes $n_t$ and $n_v$ respectively (in this paper we take $n_t=0.8n$ and $n_v=0.2n$). Following \cite{pseudo_summary}, we generate the training summary statistics
$$\begin{bmatrix}
\br\\
\widetilde{\br}
\end{bmatrix}_t = \begin{bmatrix}
\br\\
\widetilde{\br}
\end{bmatrix} + \sqrt{\frac{n_v}{n\times n_t}}\bR,$$ where $$\bR\sim \mathcal{N}\left(\bzero,\begin{bmatrix}
\bSigma & \bSigma-\bD \\
\bSigma-\bD & \bSigma
\end{bmatrix}\right),$$
and the validation summary statistics$$\begin{bmatrix}
\br\\
\widetilde{\br}
\end{bmatrix}_v =\frac{1}{n_v}\Ls n\begin{bmatrix}
\br\\
\widetilde{\br}
\end{bmatrix} - n_t\begin{bmatrix}
\br\\
\widetilde{\br}
\end{bmatrix}_t\Rs.$$ 

Given a sequence of candidate $\lambda$ values, we choose that which  maximizes an approximation $f(\lambda)$ of the correlation between the predicted values and the true values on the pseudo-validation set.\footnote{Unlike the previous approach, this tuning parameter choice will not induce the exact flip-sign property. However, we observe empirically that our method is robust to this issue, and no FDR inflation occurred. In theory, one could randomly swap all the variables with their corresponding knockoffs and compute the average of all the $\lambda$ values obtained. In the limit, the average will give a data-driven value of $\lambda$ that is invariant to swapping variables with their knockoffs due to symmetry.}  Specifically, \cite{pseudo_summary} considered the approximation

\begin{equation} \label{eqn:fullghost_pseudo_f}
f(\lambda)=\frac{ \hat{\bbeta}^\top _{t,\lambda} \begin{bmatrix}
\br\\ \widetilde{\br} \end{bmatrix}_v}{\sqrt{\hat{\bbeta}^\top _{t,\lambda}\begin{bmatrix}
\bSigma & \bSigma-\bD \\
\bSigma-\bD & \bSigma
\end{bmatrix}\hat{\bbeta}_{t,\lambda}}},\end{equation} where 
\begin{equation}
\label{eqn:fullghost_pseudo}
\hat{\bbeta}_{t,\lambda}   = \argmin_{\bbeta\in\mathbb{R}^{2p}} \frac{1}{2} \bbeta^\top  \begin{bmatrix}
\bSigma & \bSigma-\bD \\
\bSigma-\bD & \bSigma
\end{bmatrix}\bbeta - \bbeta^\top \begin{bmatrix}
\br\\
\widetilde{\br}
\end{bmatrix}_t + \lambda||\bbeta||_1.
\end{equation}

Therefore, we choose the $\lambda$ value that maximizes \eqref{eqn:fullghost_pseudo_f} among a set of candidate values. Since the objective function \eqref{eqn:fullghost} is convex in $\bbeta$, we may employ the BASIL framework proposed by \cite{basil}, which implements a batch version of the strong rules introduced in \cite{strong_rules}. BASIL can be directly applied to compute the solution path of \eqref{eqn:fullghost_pseudo} efficiently. % We call this second way of choosing $\lambda$ the \textit{pseudo-sum} method. 

Note that there exist other ways to choose the penalty level $\lambda$ using $\bX^\top \bY,\norm{\bY}_2$ and $n$ (for example, the Lassosum by \cite{lassosum}). We do not attempt to claim an optimal strategy.

\paragraph{Connection with the scout procedure} 
It turns out that step 3 of Algorithm \ref{alg:fullghost} is closely related to the scout procedure \citep{scout}. The scout procedure defines a family of covariance-regularized regression methods that achieve superior prediction via shrinking the inverse covariance matrix. It includes the Lasso, Ridge and Elastic Net as special cases. In Appendix \ref{app:scout}, we show that the solution of objective function \eqref{eqn:fullghost} is proportional to that of the scout procedure (with known precision matrix $\bSigma^{-1}$). This connection provides a justification on why the objective function \eqref{eqn:fullghost} is effective.

\subsubsection{GhostKnockoffs with other feature importance statistics} \label{section:susie_rss}

In the previous sections, we presented a feature importance statistic based on summary statistics that leads to better power than the marginal correlation difference statistic. By Proposition \ref{proposition:full_ghost}, GhostKnockoffs techniques can be combined with any other feature importance statistics that i) are based on the summary statistics $\bX^\top \bY$,  $\norm{\bY}^2$ and the sample size $n$ and ii) satisfy the flip-sign property. The procedures generated will still guarantee FDR control. In our simulation studies, we found that using the posterior inclusion probability (PIP) produced by the SuSiE-RSS model \citep{susie_rss} as the feature importance statistic also results in consistent power improvement over GK-marginal. SuSiE-RSS is based on the Sum of Single Effects (SuSiE) model proposed by \cite{susie}, which assumes a Bayesian linear model with true coefficients $\bbeta$ represented as the sum of multiple one-hot (random) individual effect vectors. \cite{susie_rss} combines SuSiE with a modified likelihood function to accommodate applications in which only summary statistics are available (see \cite{susie_rss} for details).\footnote{We used the \textit{susie\_rss} function inside the R package \texttt{susieR} in our simulations.} We call the resulting procedure GhostKnockoffs with SuSiE-RSS statistic and denote it by \textit{GK-susie-rss}. We include this method in the simulation section below.

\subsection{Variants of GhostKnockoffs}

The methods we presented so far can be adapted to work with various related procedures. We give three examples below for illustration. 

\subsubsection{Multi-knockoffs} 
\label{section:multi-knockoffs}

The knockoffs procedure is a randomized procedure which could produce very different selection sets on different runs. This is especially true when the knockoffs rejection set is small. In fact, the offset on the numerator in \eqref{eqn:knockoffs_threshold} implies that knockoffs either rejects more than $\lceil\frac{1}{q}\rceil$ hypotheses, where $q$ is the target FDR level, or rejects nothing. To improve the stability of the knockoffs procedure, \cite{multipleknockoffs} proposed simultaneous multi-knockoffs, which is substantially more stable and powerful than knockoffs when the rejection set is small and maintains FDR control in general.

The idea of \cite{multipleknockoffs} is to create $M$ (instead of one) knockoff copies for every feature so that they jointly satisfy an extended exchangeability condition.\footnote{Specifically, the extended exchangeability condition says that if we permute variables with their corresponding (multiple) knockoffs arbitrarily, the joint distribution remains unchanged.} If $X \sim \mathcal{N}(\bzero,\bSigma)$, \cite{multipleknockoffs} showed that $\widetilde{X} \in \bbR^{pM}$ is a valid $M$ multi-knockoff for $X\in \bbR^{p}$ if $\begin{bmatrix} X & \widetilde{X} \end{bmatrix} \sim \mathcal{N}(\bzero,\bG)$, where $$\bG=\begin{bmatrix}
    \bSigma & \bSigma - \bD & \cdots & \bSigma - \bD\\
    \bSigma - \bD & \bSigma & \cdots & \bSigma - \bD \\
    \vdots & \vdots & \ddots & \vdots\\
    \bSigma - \bD & \cdots & \cdots & \bSigma
\end{bmatrix} \in \mathbb{R}^{(M+1)p \times (M+1)p},$$ 
Here, $\bD=\text{diag}\{\bs\}$, and $\bs$ is obtained by solving a more restrictive convex optimization problem than in \eqref{eqn:sdp} which guarantees that $\bG$ is positive semi-definite (see \cite{multipleknockoffs} for details). In data matrix form, we generate valid $M$ multi-knockoffs by 
$$\widetilde{\bX}=\bX \bP+\bE\bV^{1/2},$$ 
where $\bP=\begin{bmatrix} \bI-\bSigma^{-1}\bD & \cdots & \bI-\bSigma^{-1}\bD \end{bmatrix} \in \mathbb{R}^{p\times Mp},$ $\bE\in \mathbb{R}^{n\times Mp}$ has i.i.d.~standard normal entries, and
\[
    \bV = \begin{bmatrix}
    2\bD - \bD \bSigma^{-1}\bD & \bD - \bD \bSigma^{-1} \bD & \cdots &  \bD - \bD\bSigma^{-1} \bD\\
    \bD - \bD \bSigma^{-1} \bD & 2\bD - \bD \bSigma^{-1}\bD & \cdots & \bD - \bD \bSigma^{-1} \bD \\
    \vdots  & \vdots & \ddots & \vdots\\
    \bD - \bD \bSigma^{-1} \bD & \bD - \bD \bSigma^{-1} \bD & \cdots & 2\bD - \bD \bSigma^{-1}\bD
    \end{bmatrix}.
\] 

\cite{multipleknockoffs} generalized the knockoffs threshold \eqref{eqn:knockoffs_threshold} and the flip-sign property to produce FDR-controlling rejection sets after generating multiple knockoffs via this procedure. 

In the summary statistics settings, upon redefining $\bP$, $\bV$ and $\bs$ as above and replacing the standard knockoffs filter by the multi-knockoffs filter, Algorithms \ref{alg:ghostkf} and \ref{alg:partialghost} produce rejection sets that have the same distribution as those produced by their corresponding versions with individual-level data. For Algorithm \ref{alg:fullghost}, we simply need to further replace $$\begin{bmatrix}
\bSigma & \bSigma-\bD \\
\bSigma-\bD & \bSigma
\end{bmatrix}$$ by $\bG$.

\subsubsection{Group knockoffs}\label{sec:group_knockoff}

When variables are highly correlated, selection procedures become conservative. For example, if a non-null variable $X_j$ is highly correlated with a null variable $X_k$, it becomes difficult to reject $X_j \indep Y | X_{-j}$. This is an important practical concern because highly correlated features are ubiquitous in many settings, particularly GWAS datasets. To overcome this challenge, \textit{group knockoffs} \citep{group_knockoffs} can be useful; please see \cite{chu2023second}, whose algorithms we employ in the data analyses of Section \ref{section:gwas}. In group knockoffs, the object of inference is shifted from single variables to groups of highly correlated variables. Specifically, suppose we partition $p$ features into $g$ groups and reorder all features such that features of the same group are in adjacent columns of $\bX$. The objective is to test group conditional independence hypothesis:
\begin{align*}
H_{\gamma}^0: X_{\gamma} \indep Y \mid X_{-\gamma}
\end{align*} 
where $\gamma \in \{1,...,g\}$ denotes a group and $X_\gamma$ is the vector of features in group $\gamma$. When these groups have strong correlation, single-variable knockoffs may struggle to identify signals, but group knockoffs retain power to identify significant groups. As in Section \ref{section:multi-knockoffs}, all methods described in this paper apply to group knockoffs after redefining $\bD$ to the equivalent version in group knockoffs. In Appendix \ref{sec:group_knockoff_appendix}, we 
detail the construction of group knockoffs and examples of importance scores at the group level for inference.

\subsubsection{Conditional randomization test}
The conditional randomization test (CRT) \citep{model-x} is an alternative method to test the conditional independence hypotheses $H_j:X_j\indep Y\mid X_{-j}$ for $1\le j \le p$. By generating a valid `CRT $p$-value' $p_j$ for each hypothesis ${H}_j$, existing multiple testing procedures, including the Benjamini-Hochberg procedure \citep{bh} and the selective SeqStep+ filter \citep{crt_shuangning}, can be used to simultaneously test $H_1,\ldots,H_p$ with FDR control.\footnote{In general, CRT $p-$values may not be independent of each other or satisfy the PRDS property \citep{by}. Therefore, applying the Benjamini-Hochberg procedure on CRT $p-$values does not guarantee FDR control theoretically. However, as noted in \cite{model-x}, the FDR is usually under control empirically.} As shown in \cite{model-x} and \cite{wang}, doing so can improve the power of multiple testing with greater computational complexity.

In Appendix \ref{ghostcrt}, we introduce Ghostknockoffs for CRT (\textit{GhostCRT}), which adopts techniques introduced in this paper to the framework of CRT.

\subsection{Numerical simulations} \label{sec:full_ghost_sim}

We conduct simulations on synthetic data as well as semi-synthetic data generated from a real-world genetic dataset. Specifically, we apply GhostKnockoffs with pseudo-lasso statistic (GK-pseudolasso, defined in Algorithm \ref{alg:fullghost} with tuning parameter $\lambda$ chosen by either lasso-min or pseudo-sum from Section \ref{section:full_tuning}) and GhostKnockoffs with SuSiE-RSS statistic (GK-susie-rss, defined in Section \ref{section:susie_rss}). We  compare their performance with GhostKnockoffs with marginal correlation difference statistic (GK-marginal, defined in Section \ref{section:model-x}) and the knockoffs procedure with (cross-validated) Lasso coefficient difference statistic based on individual-level data (KF-lassocv). We also demonstrate empirically the robustness of our procedures by showing the FDR control when only an estimate of the true covariance matrix $\bSigma$ is available and when the features are discrete.

\subsubsection{Simulations based on real-world genetic data} \label{subsubsec:full_ghost_real}

To mimic the dependency structure among features in real-world applications, we generate synthetic data based on the whole genome sequencing (WGS) data from the Alzheimer's Disease Sequencing Project (ADSP). The data are obtained from the ADSP consortium following the SNP/Indel Variant Calling Pipeline and data management tool (VCPA) \citep{leung2019vcpa}. The ADSP WGS data records counts of minor alleles of genetic variants over 16,906 individuals. Using reference populations from the 1000 Genomes Consortium \citep{10002015global}, we estimate ancestry rates of each individual by SNPWeights v2.1 \citep{chen2013improved} and extract 6,952 individuals with estimated European ancestry rate greater than 80\%. We further restrict our simulations to 2,000 randomly selected genetic variants within 0.5Mb distance to the APOE gene (chr19:44909011-45912650; hg38), whose $\varepsilon$2 allele and $\varepsilon$4 allele are known to be respectively the strongest genetic protective factor and the strongest genetic risk factor for Alzheimer's disease \citep{SerranoPozo2021,Belloy2023}, and with minor allele frequency (MAF) larger than $ 0.01$. Since our simulations focus on performance at identifying relevant clusters of tightly linked variants, we simplify the simulation design by pruning variants to eliminate pairs with absolute correlation greater than $0.75$. To do so, we first compute the correlation matrix $[\text{cor}(X_j,X_k)]_{2000\times 2000}$ of the 2,000 selected variants over the 6,952 extracted individuals using the shrinkage estimate in the R package \texttt{corpcor} \citep{schafer2005shrinkage} and apply hierarchical clustering (single-linkage with cutoff value $0.25$) on the distance matrix $[1-|\text{cor}(X_j,X_k)|]_{2000\times 2000}$. As a result, we obtain 512 variant clusters such that pairwise correlation between any pair of variants from different clusters is in  $[-0.75,0.75]$. By randomly choosing one representative variant from each cluster, we include $p=512$ tested genetic variants in the simulation study.

For each replicate, we obtain synthetic data by randomly sampling $n=3,000$ individuals without replacement and collecting the sampled individuals’ records on the $p=512$ tested genetic variants as the $n\times p$ covariate 
matrix $\bX$. We further sample another $n=3,000$ individuals without replacement as the reference panel on which we compute the correlation matrix $\bSigma$ using the shrinkage estimate in the R package \texttt{corpcor} \citep{schafer2005shrinkage}. Based on the covariate matrix $\bX$, we generate the response vector $\bY=(Y_1,\ldots,Y_n)^\top $ from either the linear model ({continuous response}),
 $$Y_i = \beta_1X_{i1}+...+\beta_{p}X_{ip}+\epsilon^C_i,\quad\text{where }\epsilon^C_i\sim N(0,3^2),$$
or the mixed-effect logit model ({binary response}),
 $$Y_i\sim \text{Bernounli}(\mu_i),\quad \text{where }g(\mu_i) = \beta_0 +\beta_1X_{i1}+...+\beta_{p}X_{ip}+\epsilon^B_i, \text{ } \epsilon^B_i\sim N(0,1^2)\text{ and }g(x) = \log\Big(\frac{x}{1-x}\Big).$$
Specifically, $\beta_0$ under the mixed-effect logit model is $-\log(9)$ so that the prevalence (or the expected proportion of $Y_i=1$) is $10\%$. $\epsilon^C_i$'s and $\epsilon^B_i$'s reflect variation due to unobserved covariates. Only $10$ randomly selected coefficients $\beta_j$ are nonzero, with value $\beta_j = \frac{1}{\sqrt{20\cdot m_j(1-m_j)}}$, where $m_j$ is the MAF of the $j$-th variant. 

% and $a$ is obtained such that $\beta_1^2\text{var}(X_1)+...+\beta_p^2\text{var}(X_p)=1$. 
With the relevant summary statistics computed, we apply GK-pseudolasso and GK-susie-rss and compare their performances with GK-marginal and KF-lassocv. 

Over 1000 replicates under both the linear model and the mixed-effect logit model, average power and FDR of different methods with respect to different target FDR levels are visualized in Figure \ref{fig:semi_synthetic}. 
Under both models, we observe that GK-pseudolasso with both ways of selecting the tuning parameter and GK-susie-rss are uniformly more powerful than GK-marginal. The performance of the proposed methods is very close to that of KF-lassocv. Despite the covariance matrix being estimated using an independent sample and the entries of $X$ being discrete, the FDRs of our proposed methods are controlled in both settings, suggesting the robustness of our methods. 

\paragraph{GhostKnockoffs with discrete features} We note that discrete covariates do not follow a Gaussian distribution. However, the knockoffs procedure ensures FDR control whenever the feature importance statistics $W_j = w(T_j, T_{p+j})$, where $w$ is an anti-symmetric function, and $\mathbf{T} \in \mathbb{R}^{2p}$ is distributionally invariant upon swapping $T_j$ with $T_{j+p}$ for each null $j$. Using Lemma \ref{lemma:Lasso_lemma}, we know that Algorithm \ref{alg:fullghost} controls the FDR if swapping the $j-$th entry of $\bZ = \bX^\top \bY$ and the $j-$th entry of $\tilde{\bZ} =\bP^\top \bX^\top \bY+\norm{\bY}_2\bZ$ does not change their joint distribution for each null $j$. In Appendix \ref{app:supplementary_semi_plots}, we visually demonstrate the approximate preservation of this distributional invariance. This, along with the robustness of knockoffs \citep{model-x, barber2020robust}, helps in explaining why we have not observed FDR inflation with discrete covariates.

\begin{figure}[htbp] 
  \centering
  \includegraphics[width=0.7\textwidth]{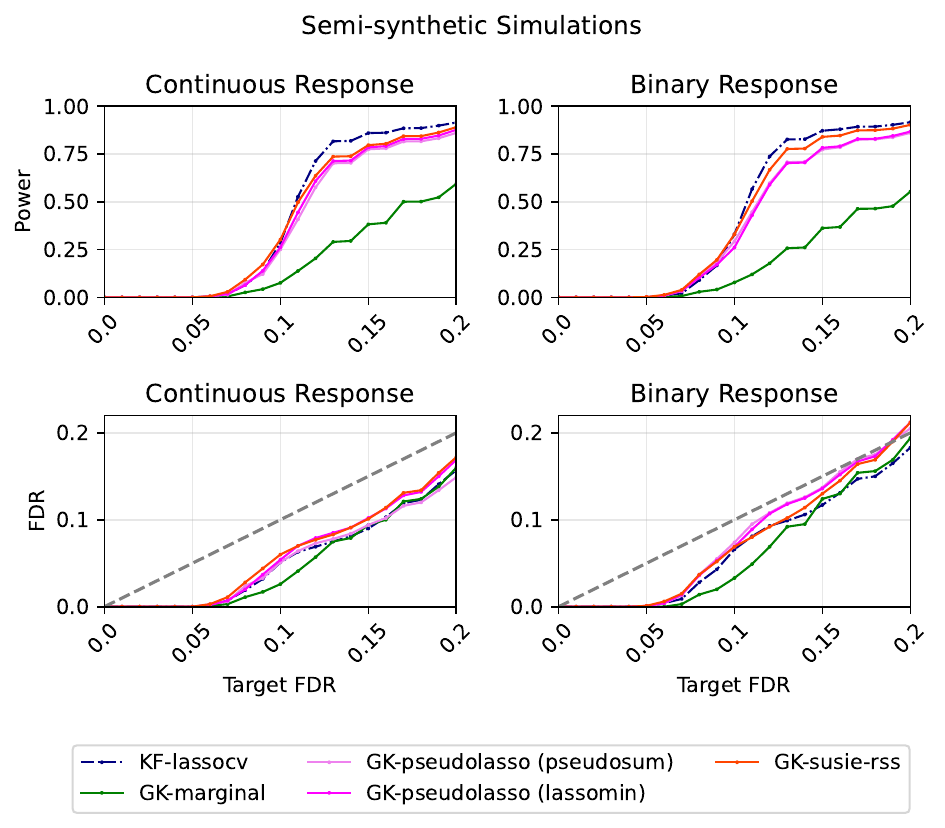}
  \caption{Average power and FDR over 1000 replications with respect to different target FDR levels in simulations based on genetic data, where features are genotypes of existing patients, and the response is simulated from a linear model ({continuous response}) or a mixed-effect logit model ({binary response}).}
  \label{fig:semi_synthetic}
\end{figure}

\subsubsection{Independent features} \label{subsubsec:full_ghost_indep}
We revisit  the setting from Section \ref{simulation_partial_indep} in which $\Sigma = I_p$. For the pseudo-sum method for GK-pseudolasso, we optimize over $\lambda$ using a grid of 100 candidate values interpolating between $\lambda_{\text{max}}$ and $\lambda_{\text{max}}/1000$ linearly in log scale, and $$\lambda_{\text{max}}=\frac{1}{n}\mathbb{E}\Ls\left\| \begin{bmatrix}
\bX^\top \bY \\
\bP^\top \bX^\top \bY+\norm{\bY}_2\bZ
\end{bmatrix}\right\|_{\infty}\Rs$$ is the minimal $\lambda$ value that shrinks all the coefficients to zero. To calculate $\mathbb{E}[\norm{\bR^\top \bepsilon}_\infty]$ for the lasso-min parameter method, we use a Monte Carlo estimate averaged over 200 samples. The target FDR is 20\%. Each point represents an average over 200 replications.

Note that when $\bSigma=\bI_p$, the solution to \eqref{eqn:sdp} is $\bD=\bI_p$. It is easy to see that \eqref{eqn:fullghost} gives $$\hat{\bbeta} = \frac{1}{n}S_{\lambda}\Lp\begin{bmatrix} \bX & \widetilde{\bX}\end{bmatrix}^\top \bY\Rp,$$ where the soft-threshold operator $S_{\lambda}(x)=sign(x)(\abs{x}-\lambda)_{+}$ is applied coordinate-wise. Therefore, the method in Section \ref{sec:pseudo-lasso} soft-thresholds the marginal correlation of $\bX$ and $\bY$.

\begin{figure}[htbp] 
  \centering
  \includegraphics[width=0.8\textwidth]{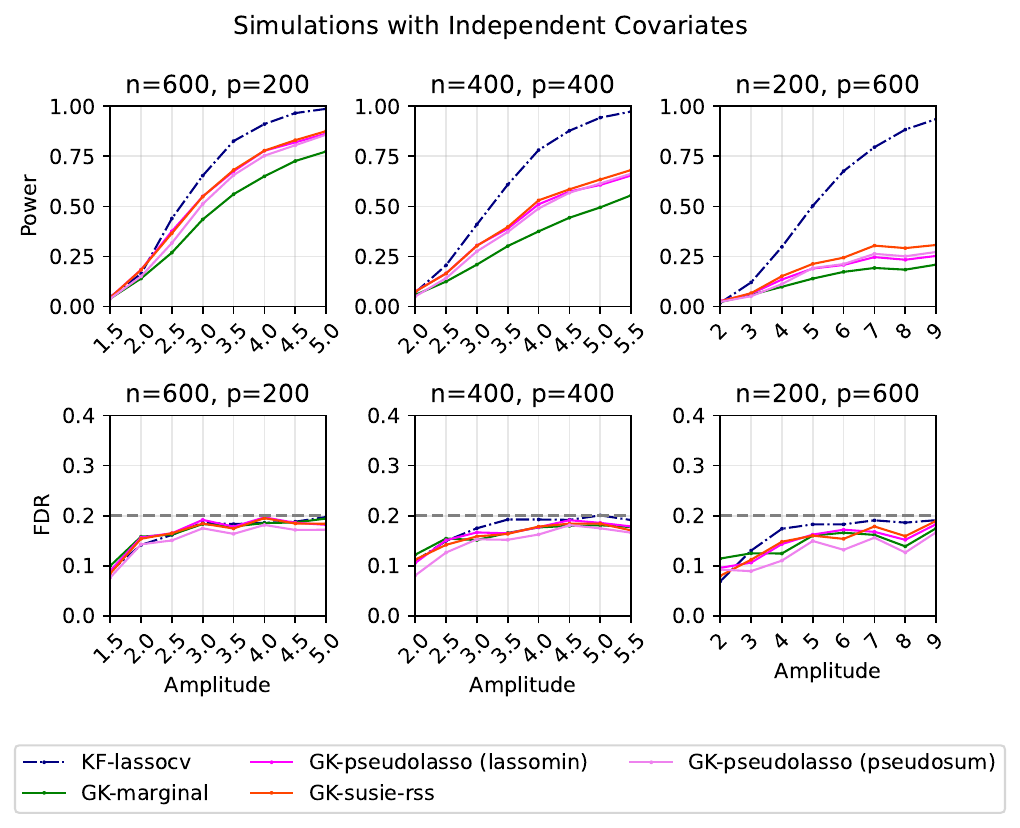}
  \caption{Power and FDR plots for independent features and a Gaussian linear model with varying dimensions. Each point is an average over 200 replications.} 
  \label{fig:full_ghost_indep}
\end{figure}

As shown in Figure \ref{fig:full_ghost_indep}, all three new methods (GK-pseudolasso with lasso-min/pseudo-sum and GK-susie-rss) consistently outperform GK-marginal, and the FDR is always controlled at the expected level, as theoretically guaranteed. As $n/p$ grows, we see that the three new methods have power closer to KF-lassocv. This is further demonstrated in additional simulations in Appendix \ref{app:add_full_gk_plots}.

\subsubsection{AR(1) features} \label{subsubsec:full_ghost_auto}
Figure \ref{fig:full_ghost_auto} shows the corresponding plots when the covariate matrix is generated from an AR(1) distribution. We found similar patterns to those with independent features. The power of all methods drops when the autocorrelation coefficient increases, as it is then harder to separate true signals from other variables.

\begin{figure}[htbp] 
  \centering
  \includegraphics[width=0.8\textwidth]{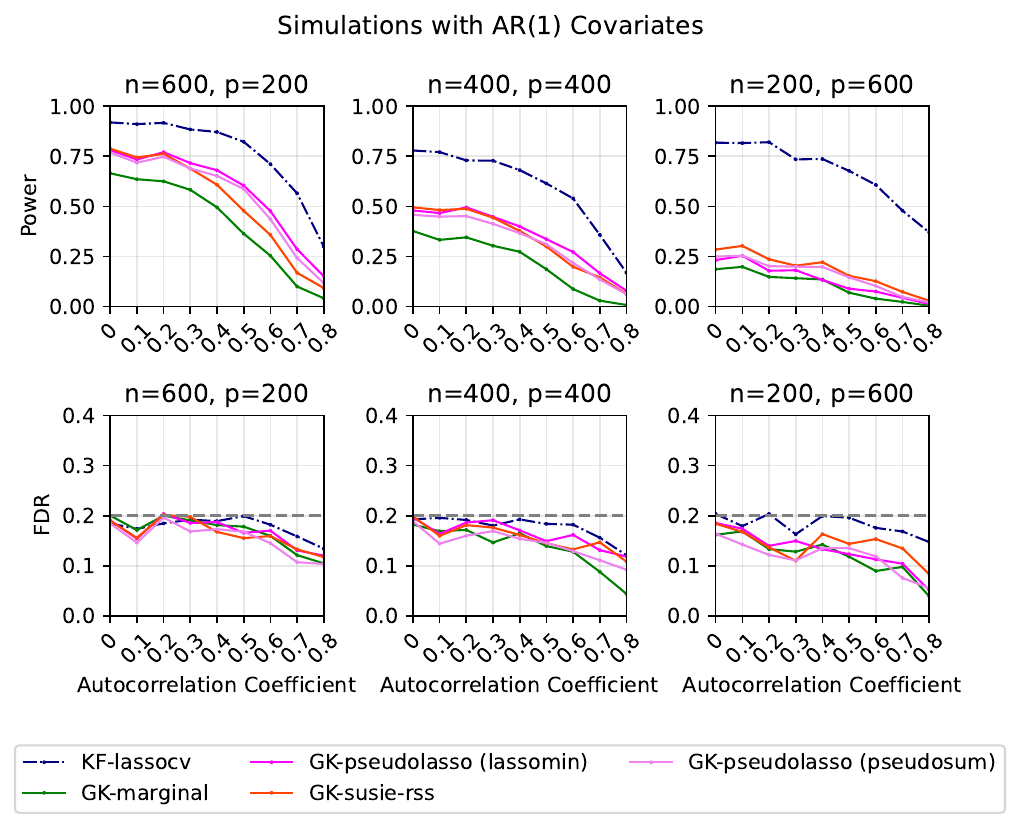}
  \caption{Power and FDR plots for AR(1) features and a Gaussian linear model with varying dimensions. Each point is an average over 200 replications.}
  \label{fig:full_ghost_auto}
\end{figure}

\section{Application to meta-analysis for Alzheimer's disease} \label{section:gwas}

To illustrate the empirical performance of the methods in detecting genetic variants associated with Alzheimer's disease (AD), we apply them to a meta-analysis of nine large-scale array-based genome-wide association and whole-exome/-genome sequencing studies for AD. We include the details of the nine studies in Appendix \ref{app:nine_studies}.

As all studies share the same focus on individuals with European ancestry, we perform a meta-analysis by aggregating their $Z$-scores and obtain the meta-analysis $Z$-score $\bZ_{\text{meta}}$ (see Appendix \ref{app:meta_z_scores} for details). In addition, we obtain the block-diagonal covariance matrix $\bSigma$ with respect to approximately independent linkage disequilibrium blocks provided by \citet{berisa2016approximately}. Within each block, we use the UK Biobank directly genotyped data as the reference panel and compute the covariance matrix via the Pan-UKB consortium (\url{https://pan.ukbb.broadinstitute.org}) with details in Appendix \ref{sec:LDmatrices}. To improve the power in the presence of tightly linked variants, we apply the group knockoffs construction on top of the GhostKnockoff algorithm, as detailed in Section \ref{sec:group_knockoff}. Finally, we implement GK-pseudolasso with tuning parameter chosen by the lasso-min method on the meta-analysis $Z$-score $\bZ_{\text{meta}}$ and the covariance matrix $\bSigma$. To stabilize the GhostKnockoffs procedures, we use $M=5$ multi-knockoffs as defined in Section \ref{section:multi-knockoffs}.

\begin{figure}[htbp]
  \centering
  \begin{subfigure}[b]{0.98\textwidth} % Adjust the width for each subfigure
    \includegraphics[width=\textwidth]{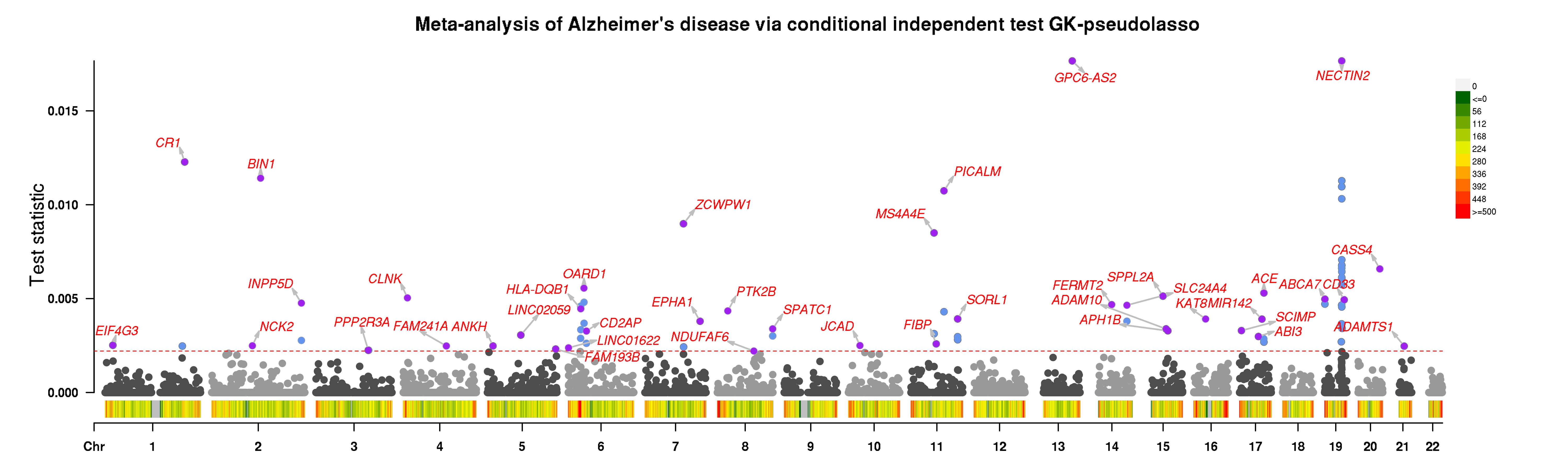}
  \end{subfigure}
  \caption{Graphical representation of the feature importance statistics after  applying the GK-pseudolasso on a meta-analysis of AD. Each point represents a group of genetic variants. With an target FDR level of 0.1, identified groups are highlighted in blue or purple. For each locus with at least one identified group, the name of the locus is presented at the variant group with the largest importance statistic (highlighted in purple). Variant density is shown at the bottom of plot (number of variants per 1Mb).}
  \label{fig:real_data_ghostlasso}
\end{figure}

Figure \ref{fig:real_data_ghostlasso} presents the result of the meta-analysis of the nine studies via our proposed method with target FDR level 0.1. Here, we specify loci based on variant groups and annotate two loci as different loci if they are 1 Mb away from each other. We adopt the most proximal gene’s name as the locus name.\footnote{Specifically, we consider the variant group with the largest group knockoff feature importance statistic within a locus, and then map the locus to the most proximal gene of the variant within the group that has the highest knockoff importance score.} As shown by Table \ref{table:real_data_table1} in Appendix \ref{app:gwas_add_tables}, GK-pseudolasso identifies variant groups in 42 and 63 loci when the target FDR level is 0.1 and 0.2 respectively, substantially more than GK-marginal (10 and 17 when the target FDR level is 0.1 and 0.2, respectively). This is consistent with our simulation results in Section \ref{sec:full_ghost_sim}. In addition, we observe from Table \ref{table:real_data_table1} that GK-susie-rss identifies fewer loci (35 and 47 when the target FDR level is 0.1 and 0.2, respectively), although it exhibits similar power in simulation studies. In Appendix \ref{app:gwas_add_figures}, we analogously visualize results of the meta-analysis via conventional marginal association test (with $p$-value cutoff $5\times 10^{-8}$), GK-marginal (with target FDR level 0.10), and GK-susie-rss (with target FDR level 0.10).

Table \ref{table:real_data_table2} in Appendix \ref{app:gwas_add_tables} shows the top variant with the largest feature importance statistic in each identified group. Most discoveries exhibit relatively strong marginal associations (marginal $p$-value $\le0.05$) in individual studies and the same direction of effects across all studies. Although some loci have an opposite direction of effect in one individual study, such effects are not significant. The consistency across individual studies supports the validity of the proposed method in discovering putative causal variants. In addition, we observe that all top variants of identified groups have small meta-analysis $p$-values (less than 0.05), though some are not smaller than the stringent genome-wide threshold ($5\times10^{-8}$) in marginal association tests with FWER control. 

To further investigate whether the identified groups are functionally enriched, we apply a SNP-to-gene linking strategy proposed by \citep{gazal2022combining} to link the top variants of identified groups to the genes that they potentially regulate. Out of 63 top variants, we find that 34 (54.0\%) can be mapped with functional evidence (e.g., being an expression quantitative trait locus, in a Hi-C linked enhancer region, near the exon of a gene, etc.), where the proportion is significantly higher than the average percentage of the background genome (28.6\%). In summary, the proposed method can identify functional genetic variants with weaker statistical effects missed by conventional association tests.

\section{Discussion} \label{section:discussion}

This paper introduced novel approaches for performing variable selection with FDR control on the basis of summary statistics. We proposed methods for testing conditional independence hypotheses from summary statistics alone. For the methods from Section \ref{section:fullghost}, all we need are essentially the marginal correlations between $X$ and $Y$,\footnote{Along with $\norm{\bY}^2$ and $n$.} which, at first sight, may appear surprising. Our arguments rely on the assumption that the covariates follow a Gaussian distribution, as well as on the linearity and rotational invariance of Gaussian distributions. Since our methods are based on the knockoffs procedure, they do not require any knowledge about the model of $Y$ given $X$. Our methods extend, and generally give better power than, the work by \cite{ghostknockoffs} by employing penalized regression to produce the measure of feature importance. The techniques employed in this paper provide a wrapper that can be combined with a variety of feature selection methods, yielding knockoffs versions that guarantee FDR control.

We applied our methods to genetic studies, in which summary statistics are typically available. Due to linkage disequilibrium, the application of our methods to individual genetic variants may yield conservative results. In a parallel work \cite{chu2023second}, we have developed tools for constructing group knockoffs efficiently and effectively. When combined, our methods offer a powerful new approach to controlled variable selection in GWAS. This is further supported in our companion work \cite{zihuai2023app}, where we see the methods in this paper led to significant scientific discoveries.

\section{Acknowledgement} \label{section:acknowledgement}
Z.C. would like to thank Kevin Guo and Amber Hu for helpful discussions. Z.C. was supported by the Simons Foundation under award 814641. Z.H. was supported by NIH/NIA award AG066206 and AG066515. T.M. was supported by a B.C. and E.J. Eaves Stanford Graduate Fellowship. C.S. was supported by the grants NIH R56HG010812 and NSF DMS2210392. E.J.C. was supported by the Office of Naval Research grant N00014-20-1-2157.

\newpage
\bibliography{refs}

\begin{thebibliography}{46}
\providecommand{\natexlab}[1]{#1}
\providecommand{\url}[1]{\texttt{#1}}
\expandafter\ifx\csname urlstyle\endcsname\relax
  \providecommand{\doi}[1]{doi: #1}\else
  \providecommand{\doi}{doi: \begingroup \urlstyle{rm}\Url}\fi

\bibitem[Barber and Cand{\`e}s(2015)]{fixed-x}
R.~F. Barber and E.~J. Cand{\`e}s.
\newblock {Controlling the false discovery rate via knockoffs}.
\newblock \emph{The Annals of Statistics}, 43\penalty0 (5):\penalty0 2055 --
  2085, 2015.
\newblock URL \url{https://doi.org/10.1214/15-AOS1337}.

\bibitem[Barber et~al.(2020)Barber, Cand{\`e}s, and Samworth]{barber2020robust}
R.~F. Barber, E.~J. Cand{\`e}s, and R.~J. Samworth.
\newblock Robust inference with knockoffs.
\newblock 2020.

\bibitem[Bates et~al.(2020)Bates, Sesia, Sabatti, and
  Cand{\`e}s]{bates2020causal}
S.~Bates, M.~Sesia, C.~Sabatti, and E.~Cand{\`e}s.
\newblock Causal inference in genetic trio studies.
\newblock \emph{Proceedings of the National Academy of Sciences}, 117\penalty0
  (39):\penalty0 24117--24126, 2020.

\bibitem[Belloni et~al.(2011)Belloni, Chernozhukov, and Wang]{sqrtlasso}
A.~Belloni, V.~Chernozhukov, and L.~Wang.
\newblock Square-root lasso: pivotal recovery of sparse signals via conic
  programming.
\newblock \emph{Biometrika}, 98\penalty0 (4):\penalty0 791--806, 2011.

\bibitem[Belloy et~al.(2022{\natexlab{a}})Belloy, Eger, Le~Guen, Damotte,
  Ahmad, Ikram, Ramirez, Tsolaki, Rossi, Jansen, et~al.]{belloy2022challenges}
M.~E. Belloy, S.~J. Eger, Y.~Le~Guen, V.~Damotte, S.~Ahmad, M.~A. Ikram,
  A.~Ramirez, A.~C. Tsolaki, G.~Rossi, I.~E. Jansen, et~al.
\newblock {Challenges at the APOE locus: a robust quality control approach for
  accurate APOE genotyping}.
\newblock \emph{Alzheimer's Research \& Therapy}, 14:\penalty0 22,
  2022{\natexlab{a}}.

\bibitem[Belloy et~al.(2022{\natexlab{b}})Belloy, Le~Guen, Eger, Napolioni,
  Greicius, and He]{belloy2022fast}
M.~E. Belloy, Y.~Le~Guen, S.~J. Eger, V.~Napolioni, M.~D. Greicius, and Z.~He.
\newblock {A Fast and Robust Strategy to Remove Variant-Level Artifacts in
  Alzheimer Disease Sequencing Project Data}.
\newblock \emph{Neurology Genetics}, 8\penalty0 (5):\penalty0 e200012,
  2022{\natexlab{b}}.

\bibitem[Belloy et~al.(2023)Belloy, Andrews, Le~Guen, Cuccaro, Farrer,
  Napolioni, and Greicius]{Belloy2023}
M.~E. Belloy, S.~J. Andrews, Y.~Le~Guen, M.~Cuccaro, L.~A. Farrer,
  V.~Napolioni, and M.~D. Greicius.
\newblock {APOE Genotype and Alzheimer Disease Risk Across Age, Sex, and
  Population Ancestry}.
\newblock \emph{JAMA Neurology}, 80\penalty0 (12):\penalty0 1284--1294, 2023.

\bibitem[Benjamini and Hochberg(1995)]{bh}
Y.~Benjamini and Y.~Hochberg.
\newblock {Controlling the False Discovery Rate: A Practical and Powerful
  Approach to Multiple Testing}.
\newblock \emph{Journal of the Royal Statistical Society: Series B
  (Methodological)}, 57\penalty0 (1):\penalty0 289--300, 1995.

\bibitem[Benjamini and Yekutieli(2001)]{by}
Y.~Benjamini and D.~Yekutieli.
\newblock The control of the false discovery rate in multiple testing under
  dependency.
\newblock \emph{Annals of statistics}, pages 1165--1188, 2001.

\bibitem[Berisa and Pickrell(2016)]{berisa2016approximately}
T.~Berisa and J.~K. Pickrell.
\newblock Approximately independent linkage disequilibrium blocks in human
  populations.
\newblock \emph{Bioinformatics}, 32\penalty0 (2):\penalty0 283--285, 2016.

\bibitem[Bis et~al.(2020)Bis, Jian, Kunkle, Chen, Hamilton-Nelson, Bush,
  Salerno, Lancour, Ma, Renton, et~al.]{bis2020whole}
J.~C. Bis, X.~Jian, B.~W. Kunkle, Y.~Chen, K.~L. Hamilton-Nelson, W.~S. Bush,
  W.~J. Salerno, D.~Lancour, Y.~Ma, A.~E. Renton, et~al.
\newblock {Whole exome sequencing study identifies novel rare and common
  Alzheimer’s-Associated variants involved in immune response and
  transcriptional regulation}.
\newblock \emph{Molecular psychiatry}, 25:\penalty0 1859--1875, 2020.

\bibitem[Boyd and Vandenberghe(2004)]{convex_opt}
S.~Boyd and L.~Vandenberghe.
\newblock \emph{Convex Optimization}.
\newblock Cambridge University Press, 2004.

\bibitem[Cand{\`e}s et~al.(2018)Cand{\`e}s, Fan, Janson, and Lv]{model-x}
E.~Cand{\`e}s, Y.~Fan, L.~Janson, and J.~Lv.
\newblock {Panning for Gold: ‘Model-X’ Knockoffs for High Dimensional
  Controlled Variable Selection}.
\newblock \emph{Journal of the Royal Statistical Society Series B: Statistical
  Methodology}, 80\penalty0 (3):\penalty0 551--577, 2018.

\bibitem[Chen et~al.(2013)Chen, Pollack, Hunter, Hirschhorn, Kraft, and
  Price]{chen2013improved}
C.-Y. Chen, S.~Pollack, D.~J. Hunter, J.~N. Hirschhorn, P.~Kraft, and A.~L.
  Price.
\newblock Improved ancestry inference using weights from external reference
  panels.
\newblock \emph{Bioinformatics}, 29\penalty0 (11):\penalty0 1399--1406, 2013.

\bibitem[Chu et~al.(2023)Chu, Gu, Chen, Morrison, Cand{\`e}s, He, and
  Sabatti]{chu2023second}
B.~B. Chu, J.~Gu, Z.~Chen, T.~Morrison, E.~Cand{\`e}s, Z.~He, and C.~Sabatti.
\newblock {Second-order group knockoffs with applications to GWAS}.
\newblock \emph{arXiv preprint arXiv:2310.15069}, 2023.

\bibitem[Dai and Barber(2016)]{group_knockoffs}
R.~Dai and R.~Barber.
\newblock {The knockoff filter for FDR control in group-sparse and multitask
  regression}.
\newblock In \emph{Proceedings of The 33rd International Conference on Machine
  Learning}, volume~48, pages 1851--1859. PMLR, 2016.

\bibitem[Dicker(2014)]{variance-estimation}
L.~H. Dicker.
\newblock Variance estimation in high-dimensional linear models.
\newblock \emph{Biometrika}, 101\penalty0 (2):\penalty0 269--284, 2014.

\bibitem[Gazal et~al.(2022)Gazal, Weissbrod, Hormozdiari, Dey, Nasser,
  Jagadeesh, Weiner, Shi, Fulco, O’Connor, et~al.]{gazal2022combining}
S.~Gazal, O.~Weissbrod, F.~Hormozdiari, K.~K. Dey, J.~Nasser, K.~A. Jagadeesh,
  D.~J. Weiner, H.~Shi, C.~P. Fulco, L.~J. O’Connor, et~al.
\newblock {Combining SNP-to-gene linking strategies to identify disease genes
  and assess disease omnigenicity}.
\newblock \emph{Nature Genetics}, 54:\penalty0 827--836, 2022.

\bibitem[Gimenez and Zou(2019)]{multipleknockoffs}
J.~R. Gimenez and J.~Zou.
\newblock {Improving the Stability of the Knockoff Procedure: Multiple
  Simultaneous Knockoffs and Entropy Maximization}.
\newblock In \emph{Proceedings of the Twenty-Second International Conference on
  Artificial Intelligence and Statistics}, volume~89, pages 2184--2192. PMLR,
  2019.

\bibitem[He et~al.(2021)He, Liu, Wang, Le~Guen, Lee, Gogarten, Lu, Montgomery,
  Tang, Silverman, et~al.]{he2021identification}
Z.~He, L.~Liu, C.~Wang, Y.~Le~Guen, J.~Lee, S.~Gogarten, F.~Lu, S.~Montgomery,
  H.~Tang, E.~K. Silverman, et~al.
\newblock Identification of putative causal loci in whole-genome sequencing
  data via knockoff statistics.
\newblock \emph{Nature Communications}, 12:\penalty0 3152, 2021.

\bibitem[He et~al.(2022)He, Liu, Belloy, Le~Guen, Sossin, Liu, Qi, Ma, Gyawali,
  Wyss-Coray, et~al.]{ghostknockoffs}
Z.~He, L.~Liu, M.~E. Belloy, Y.~Le~Guen, A.~Sossin, X.~Liu, X.~Qi, S.~Ma, P.~K.
  Gyawali, T.~Wyss-Coray, et~al.
\newblock Ghostknockoff inference empowers identification of putative causal
  variants in genome-wide association studies.
\newblock \emph{Nature Communications}, 13:\penalty0 7209, 2022.

\bibitem[He et~al.(2023)]{zihuai2023app}
Z.~He et~al.
\newblock In silico identification of putative causal genetic variants.
\newblock 2023.

\bibitem[Huang et~al.(2017)Huang, Marcora, Pimenova, Di~Narzo, Kapoor, Jin,
  Harari, Bertelsen, Fairfax, Czajkowski, et~al.]{huang2017common}
K.-l. Huang, E.~Marcora, A.~A. Pimenova, A.~F. Di~Narzo, M.~Kapoor, S.~C. Jin,
  O.~Harari, S.~Bertelsen, B.~P. Fairfax, J.~Czajkowski, et~al.
\newblock {A common haplotype lowers PU.1 expression in myeloid cells and
  delays onset of Alzheimer's disease}.
\newblock \emph{Nature Neuroscience}, 20:\penalty0 1052--1061, 2017.

\bibitem[Jansen et~al.(2019)Jansen, Savage, Watanabe, Bryois, Williams,
  Steinberg, Sealock, Karlsson, H{\"a}gg, Athanasiu, et~al.]{jansen2019genome}
I.~E. Jansen, J.~E. Savage, K.~Watanabe, J.~Bryois, D.~M. Williams,
  S.~Steinberg, J.~Sealock, I.~K. Karlsson, S.~H{\"a}gg, L.~Athanasiu, et~al.
\newblock {Genome-wide meta-analysis identifies new loci and functional
  pathways influencing Alzheimer’s disease risk}.
\newblock \emph{Nature Genetics}, 51:\penalty0 404--413, 2019.

\bibitem[Kunkle et~al.(2019)Kunkle, Grenier-Boley, Sims, Bis, Damotte, Naj,
  Boland, Vronskaya, Van Der~Lee, Amlie-Wolf, et~al.]{kunkle2019genetic}
B.~W. Kunkle, B.~Grenier-Boley, R.~Sims, J.~C. Bis, V.~Damotte, A.~C. Naj,
  A.~Boland, M.~Vronskaya, S.~J. Van Der~Lee, A.~Amlie-Wolf, et~al.
\newblock {Genetic meta-analysis of diagnosed Alzheimer’s disease identifies
  new risk loci and implicates A$\beta$, tau, immunity and lipid processing}.
\newblock \emph{Nature Genetics}, 51:\penalty0 414--430, 2019.

\bibitem[Le~Guen et~al.(2021)Le~Guen, Belloy, Napolioni, Eger, Kennedy, Tao,
  He, and Greicius]{le2021novel}
Y.~Le~Guen, M.~E. Belloy, V.~Napolioni, S.~J. Eger, G.~Kennedy, R.~Tao, Z.~He,
  and M.~D. Greicius.
\newblock {A novel age-informed approach for genetic association analysis in
  Alzheimer’s disease}.
\newblock \emph{Alzheimer's Research \& Therapy}, 13:\penalty0 72, 2021.

\bibitem[Leung et~al.(2019)Leung, Valladares, Chou, Lin, Kuzma, Cantwell, Qu,
  Gangadharan, Salerno, Schellenberg, et~al.]{leung2019vcpa}
Y.~Y. Leung, O.~Valladares, Y.-F. Chou, H.-J. Lin, A.~B. Kuzma, L.~Cantwell,
  L.~Qu, P.~Gangadharan, W.~J. Salerno, G.~D. Schellenberg, et~al.
\newblock {VCPA: genomic variant calling pipeline and data management tool for
  Alzheimer's Disease Sequencing Project}.
\newblock \emph{Bioinformatics}, 35\penalty0 (10):\penalty0 1768--1770, 2019.

\bibitem[Li and Cand{\`e}s(2021)]{crt_shuangning}
S.~Li and E.~J. Cand{\`e}s.
\newblock {Deploying the Conditional Randomization Test in High Multiplicity
  Problems}.
\newblock \emph{arXiv preprint arXiv:2110.02422}, 2021.

\bibitem[Mak et~al.(2017)Mak, Porsch, Choi, Zhou, and Sham]{lassosum}
T.~S.~H. Mak, R.~M. Porsch, S.~W. Choi, X.~Zhou, and P.~C. Sham.
\newblock Polygenic scores via penalized regression on summary statistics.
\newblock \emph{Genetic Epidemiology}, 41:\penalty0 469--480, 2017.

\bibitem[Pasaniuc and Price(2017)]{pasaniuc2017dissecting}
B.~Pasaniuc and A.~L. Price.
\newblock Dissecting the genetics of complex traits using summary association
  statistics.
\newblock \emph{Nature Reviews Genetics}, 18:\penalty0 117--127, 2017.

\bibitem[Qian et~al.(2020)Qian, Tanigawa, Du, Aguirre, Chang, Tibshirani,
  Rivas, and Hastie]{basil}
J.~Qian, Y.~Tanigawa, W.~Du, M.~Aguirre, C.~Chang, R.~Tibshirani, M.~A. Rivas,
  and T.~Hastie.
\newblock {A fast and scalable framework for large-scale and
  ultrahigh-dimensional sparse regression with application to the UK Biobank}.
\newblock \emph{PLoS Genetics}, 16\penalty0 (10):\penalty0 e1009141, 2020.

\bibitem[Sch{\"a}fer and Strimmer(2005)]{schafer2005shrinkage}
J.~Sch{\"a}fer and K.~Strimmer.
\newblock A shrinkage approach to large-scale covariance matrix estimation and
  implications for functional genomics.
\newblock \emph{Statistical Applications in Genetics and Molecular Biology},
  4:\penalty0 32, 2005.

\bibitem[Schwartzentruber et~al.(2021)Schwartzentruber, Cooper, Liu,
  Barrio-Hernandez, Bello, Kumasaka, Young, Franklin, Johnson, Estrada,
  et~al.]{schwartzentruber2021genome}
J.~Schwartzentruber, S.~Cooper, J.~Z. Liu, I.~Barrio-Hernandez, E.~Bello,
  N.~Kumasaka, A.~M. Young, R.~J. Franklin, T.~Johnson, K.~Estrada, et~al.
\newblock {Genome-wide meta-analysis, fine-mapping and integrative
  prioritization implicate new Alzheimer’s disease risk genes}.
\newblock \emph{Nature Genetics}, 53:\penalty0 392--402, 2021.

\bibitem[Serrano-Pozo et~al.(2021)Serrano-Pozo, Das, and
  Hyman]{SerranoPozo2021}
A.~Serrano-Pozo, S.~Das, and B.~T. Hyman.
\newblock {APOE and Alzheimer's disease: advances in genetics, pathophysiology,
  and therapeutic approaches}.
\newblock \emph{The Lancet Neurology}, 20\penalty0 (1):\penalty0 68--80, 2021.

\bibitem[Sesia et~al.(2021)Sesia, Bates, Cand{\`e}s, Marchini, and
  Sabatti]{sesia2021false}
M.~Sesia, S.~Bates, E.~Cand{\`e}s, J.~Marchini, and C.~Sabatti.
\newblock False discovery rate control in genome-wide association studies with
  population structure.
\newblock \emph{Proceedings of the National Academy of Sciences}, 118\penalty0
  (40):\penalty0 e2105841118, 2021.

\bibitem[Spector and Janson(2022)]{mvrknockoffs}
A.~Spector and L.~Janson.
\newblock Powerful knockoffs via minimizing reconstructability.
\newblock \emph{The Annals of Statistics}, 50\penalty0 (1):\penalty0 252--276,
  2022.

\bibitem[{The 1000 Genomes Project Consortium}(2015)]{10002015global}
{The 1000 Genomes Project Consortium}.
\newblock A global reference for human genetic variation.
\newblock \emph{Nature}, 526:\penalty0 68--74, 2015.

\bibitem[Tian et~al.(2018)Tian, Loftus, and Taylor]{sqrtlasso2}
X.~Tian, J.~R. Loftus, and J.~E. Taylor.
\newblock Selective inference with unknown variance via the square-root lasso.
\newblock \emph{Biometrika}, 105\penalty0 (4):\penalty0 755--768, 2018.

\bibitem[Tibshirani et~al.(2012)Tibshirani, Bien, Friedman, Hastie, Simon,
  Taylor, and Tibshirani]{strong_rules}
R.~Tibshirani, J.~Bien, J.~Friedman, T.~Hastie, N.~Simon, J.~Taylor, and R.~J.
  Tibshirani.
\newblock {Strong Rules for Discarding Predictors in Lasso-Type Problems}.
\newblock \emph{Journal of the Royal Statistical Society Series B: Statistical
  Methodology}, 74\penalty0 (2):\penalty0 245--266, 2012.

\bibitem[Wang et~al.(2020)Wang, Sarkar, Carbonetto, and Stephens]{susie}
G.~Wang, A.~Sarkar, P.~Carbonetto, and M.~Stephens.
\newblock {A Simple New Approach to Variable Selection in Regression, with
  Application to Genetic Fine Mapping}.
\newblock \emph{Journal of the Royal Statistical Society Series B: Statistical
  Methodology}, 82\penalty0 (5):\penalty0 1273--1300, 2020.

\bibitem[Wang and Janson(2021)]{wang}
W.~Wang and L.~Janson.
\newblock A high-dimensional power analysis of the conditional randomization
  test and knockoffs.
\newblock \emph{Biometrika}, 109\penalty0 (3):\penalty0 631--645, 2021.

\bibitem[Weinstein et~al.(2020)Weinstein, Su, Bogdan, Barber, and
  Cand{\`e}s]{lcd_power}
A.~Weinstein, W.~J. Su, M.~Bogdan, R.~F. Barber, and E.~J. Cand{\`e}s.
\newblock {A Power Analysis for Model-X Knockoffs with $\ell_p$-Regularized
  Statistics}.
\newblock \emph{arXiv preprint arXiv:2007.15346}, 2020.

\bibitem[Willer et~al.(2010)Willer, Li, and Abecasis]{willer2010metal}
C.~J. Willer, Y.~Li, and G.~R. Abecasis.
\newblock {METAL: fast and efficient meta-analysis of genomewide association
  scans}.
\newblock \emph{Bioinformatics}, 26\penalty0 (17):\penalty0 2190--2191, 2010.

\bibitem[Witten and Tibshirani(2009)]{scout}
D.~M. Witten and R.~Tibshirani.
\newblock Covariance-regularized regression and classification for high
  dimensional problems.
\newblock \emph{Journal of the Royal Statistical Society Series B: Statistical
  Methodology}, 71\penalty0 (3):\penalty0 615--636, 2009.

\bibitem[Zhang et~al.(2021)Zhang, Priv{\'e}, Vilhj{\'a}lmsson, and
  Speed]{pseudo_summary}
Q.~Zhang, F.~Priv{\'e}, B.~Vilhj{\'a}lmsson, and D.~Speed.
\newblock Improved genetic prediction of complex traits from individual-level
  data or summary statistics.
\newblock \emph{Nature Communications}, 12:\penalty0 4192, 2021.

\bibitem[Zou et~al.(2022)Zou, Carbonetto, Wang, and Stephens]{susie_rss}
Y.~Zou, P.~Carbonetto, G.~Wang, and M.~Stephens.
\newblock {Fine-mapping from summary data with the “Sum of Single Effects”
  model}.
\newblock \emph{PLoS Genetics}, 18\penalty0 (7):\penalty0 e1010299, 2022.

\end{thebibliography}
\bibliographystyle{abbrvnat}

\newpage
\appendix

\section{Computation of free parameters $\bs$} 
\label{app:s}
In this paper, we use the semidefinite program (SDP) construction of second-order knockoffs \cite{model-x}. Without loss of generality, we assume that columns of the data matrix $\bX$ have been standardized with mean 0 and variance 1 such that diagonal entries $\bSigma$ are 1. As a result, $\bs$ is the solution of the convex optimization problem.
\begin{equation}
\label{eqn:sdp}
\begin{aligned}
\textrm{minimize} \quad & \sum_{j=1}^p |1-s_j|\\
\textrm{subject to} \quad &s_j \ge 0, \quad  \ 1\le j \le p,\\
  &\text{diag}\{\bs\} \preceq 2\bSigma.    \\
\end{aligned}
\end{equation}

Other methods to compute $\bs$ include the minimum variance-based reconstructability (MVR) construction \citep{mvrknockoffs} and maximum entropy (ME) construction \citep{multipleknockoffs,mvrknockoffs}, which are all compatible with our methods in this paper.

\section{Equivalence of GhostKnockoffs and the Gaussian knockoff sampler in sampling the knockoff $Z$-score $\widetilde{\bZ}_s$} \label{ghostknockoffs_proof}

In this section, we summarize the proof of \citet{ghostknockoffs} that $\widetilde{\bZ}_s$ computed by \eqref{eqn:Z_s} satisfies \eqref{eqn:Z_s_required} as follows.

\begin{lemma}{\rm \citep{ghostknockoffs}}
For any $\bP$ and $\bV$ computed in step 3 of Algorithm \ref{alg:gaussiankf}, we have $$\widetilde{\bZ}_s\mid \bX, \bY \stackrel{d}{=}\widetilde{\bX}^\top \bY\mid \bX, \bY,$$
where $\widetilde{\bZ}_s$ is computed by \eqref{eqn:Z_s} and $\widetilde{\bX}$ is the output of Algorithm \ref{alg:gaussiankf}.
\end{lemma}

\begin{proof}
By step 5 of Algorithm \ref{alg:gaussiankf}, we have $\widetilde{\bX}=\bX \bP+\bE\bV^{1/2}$, where $\bE$ is an $n$ by $p$ matrix with i.i.d.~standard Gaussian entries, independent of $\bX$. Therefore,
\begin{align*}
\widetilde{\bX}^\top \bY\mid\bX,\bY &\stackrel{}{=} \bP^\top \bX^\top \bY+\bV^{1/2}\bE^\top \bY\mid \bX,\bY.
\end{align*}
Because $\bE^\top \bY\mid \bX,\bY\sim \mathcal{N}(\bzero,||\bY||_2^2\bI_p)$, we have
$$\bE^\top \bY\mid \bX,\bY\stackrel{d}{=}||\bY||_2\bS\mid \bX,\bY,\quad \text{where} \;\bS \sim \mathcal{N}(\bzero,\bI_p) \;\text{is independent of} \;\bX \;\text{and} \;\bY$$

Thus, we have
\begin{align*}
\widetilde{\bX}^\top \bY\mid\bX,\bY &\stackrel{d}{=} \bP^\top \bX^\top \bY + ||\bY||_2\bV^{1/2}\bS\mid \bX,\bY\\
&\stackrel{d}{=}  \bP^\top \bX^\top \bY + ||\bY||_2\bZ\mid \bX,\bY \quad\text{where} \;\bZ \sim \mathcal{N}(\bzero,\bV) \;\text{is independent of} \;\bX \;\text{and} \;\bY.\\
&\stackrel{}{=}  \widetilde{\bZ}_s\mid \bX, \bY .
\end{align*}
\end{proof}

\section{Proof of Proposition \ref{prop:partial_ghost}}\label{proof:partial_ghost}

To prove Proposition \ref{prop:partial_ghost}, we need to first prove Lemma \ref{lm:orthogonal}.

\begin{lemma}\label{lm:orthogonal}
Let $\bZ_1$ and $\bZ_2$ be two real $n$ by $p$ matrices. For any $n$ and $p$, if $\bZ_1^\top \bZ_1 = \bZ_2^\top \bZ_2$, there must exists an orthogonal matrix $\bQ \in \mathbb{R}^{p\times p}$ such that $\bZ_1 = \bQ\bZ_2$.
\end{lemma}

\begin{proof}
Suppose $\bZ_1^\top \bZ_1 = \bZ_2^\top \bZ_2 = \bU\bLambda \bU^\top $, where $\bU\in R^{p\times r}$ is an orthogonal matrix such that $\bU^\top \bU=\bI_r$,  $\bLambda\in R^{r\times r}$ is diagonal with positive entries and $r$ is the rank of $\bZ_1^\top \bZ_1$. In other words, we perform eigen-decomposition of $\bZ_1^\top \bZ_1 = \bZ_2^\top \bZ_2$ and remove all zero eigenvalues and their corresponding eigenvectors. Note that $\bU\bU^\top $ is a projection matrix that projects any vector onto $\textit{colspace}(\bU)$, the column space of $\bU$.

 It is clear that
$$\textit{colspace}(\bU\Lambda \bU^\top ) \subseteq \textit{colspace}(\bU).$$ Because $\bU = (\bU\bLambda \bU^\top )\bU\bLambda^{-1}$, we also have $$\textit{colspace}(\bU) \subseteq \textit{colspace}(\bU\Lambda \bU^\top ).$$ As a result, we have $\textit{colspace}(\bU\Lambda \bU^\top ) =\textit{colspace}(\bU)$.

Thus, for $k=1,2$, 
$\bU\bU^\top $ is a projection matrix that projects any vector onto the column space of $\bU\bLambda \bU^\top  = \bZ_k^\top \bZ_k$. Because $\textit{colspace}(\bZ_k^\top \bZ_k)=\textit{rowspace}(\bZ_k)$, we have $$\bZ_k =\bZ_k\bU\bU^\top=\bZ_k\bU\bLambda^{-1/2}\bLambda^{1/2}\bU^\top.$$
Let $\bQ_k = \bZ_k\bU\bLambda^{-1/2}$, we have $\bZ_k=\bQ_k\bLambda^{1/2}\bU^\top$ 
and 
$$\bQ_k^\top \bQ_k = \bLambda^{-1/2}\bU^\top \bZ_k^\top \bZ_k \bU\bLambda^{-1/2}= \bLambda^{-1/2}\bU^\top \bU\bLambda \bU^\top \bU\bLambda^{-1/2}=\bI_r, \quad (k=1,2).$$
Thus, we have 
\begin{align*}
\bZ_1 &= \bQ_1\bLambda^{1/2}\bU^\top  = \bQ_1\bQ_2^\top \bQ_2\bLambda^{1/2}\bU^\top  = \bQ_1\bQ_2^\top \bZ_2,\\
\bZ_2 &= \bQ_2\bLambda^{1/2}\bU^\top  = \bQ_2\bQ_1^\top \bQ_1\bLambda^{1/2}\bU^\top  = \bQ_2\bQ_1^\top \bZ_1.
\end{align*}

Because $\bQ_1^\top \bQ_1=\bQ_2^\top \bQ_2=\bI_r$, there exist $\bQ_1^\perp,\bQ_2^\perp\in R^{p\times (p-r)}$ such that  $\bV_1 = \begin{bmatrix}
\bQ_1 & \bQ_1^{\perp} 
\end{bmatrix}$ and $\bV_2 = \begin{bmatrix}
\bQ_2 & \bQ_2^{\perp} 
\end{bmatrix}$ are both orthogonal matrices. Thus, we have
\begin{align*}
    \bZ_1 = \bQ_1\bQ_2^\top \bZ_2=(\bV_1\bV_2^\top  - \bQ_1^{\perp}(\bQ_2^{\perp})^\top )\bZ_2\\
    \bZ_2 = \bQ_2\bQ_1^\top \bZ_1=(\bV_2\bV_1^\top  - \bQ_2^{\perp}(\bQ_1^{\perp})^\top )\bZ_1
\end{align*}
Substituting $\bZ_1 = \bQ_1\bQ_2^\top \bZ_2$ in $\bZ_2 = \bQ_2\bQ_1^\top \bZ_1$, we have
$$\bZ_2 = \bQ_2\bQ_1^\top \bQ_1\bQ_2^\top \bZ_2 = \bQ_2\bQ_2^\top \bZ_2$$
and thus 
$$\bQ_1^{\perp}(\bQ_2^{\perp})^\top )\bZ_2=\bQ_2\bQ_2^\top \bZ_2=\mathbf{0}.$$
Thus, these exists an orthogonal matrix $\bQ = \bV_1\bV_2^\top $ such that $\bZ_1 = \bQ\bZ_2$.

\end{proof}

We can then prove Proposition \ref{prop:partial_ghost} as follows.
By Lemma \ref{lm:orthogonal}, since $[\widecheck{\bX}\ \widecheck{\bY}]^\top [\widecheck{\bX}\ \widecheck{\bY}] = [\bX\ \bY]^\top [\bX\ \bY]$, we know that $[\widecheck{\bX}\ \widecheck{\bY}] = \bQ^\top [\bX \ \bY]$ for some orthogonal matrix $\bQ$. 

Let $\bE \in \bbR^{n \times p}$ be a matrix with i.i.d.~standard Gaussian entries, we have $\bQ\bE$ is also a matrix with i.i.d.~standard Gaussian entries (i.e. $\bE\stackrel{d}{=}\bQ\bE$) and \begin{align*}
(\bE^\top [\widecheck{\bX}\ \widecheck{\bY}],\bE^\top \bE)\mid \bX,\bY &\stackrel{}{=} (\bE^\top \bQ^\top [\bX\ \bY],\bE^\top \bQ^\top\bQ\bE)\mid \bX,\bY \\ &\stackrel{d}{=} (\bE^\top [\bX\ \bY],\bE^\top \bE)\mid \bX,\bY.
\end{align*}

By the construction of $[\widecheck{\bX}\ \widecheck{\bY}]^\top [\widecheck{\bX}\ \widecheck{\bY}]$, we have that $\widecheck{\bX}^\top \widecheck{\bX}=\bX^\top \bX$, $\widecheck{\bX}^\top \widecheck{\bY}=\bX^\top \bY$ and $\norm{\widecheck{\bY}}_2 =\norm{\bY}_2$. Therefore, we focus on the third, fifth and sixth arguments of $\mathcal{T}$ where
\begin{align*}
&(\widetilde{\bX}^\top \bY,\bX^\top \bX,\widetilde{\bX}^\top \widetilde{\bX})\mid \bX,\bY \\ \stackrel{}{=}& \; (\bP^\top \bX^\top \bY+\bV^{1/2}\bE^\top \bY,\bP^\top \bX^\top \bX+\bV^{1/2}\bE^\top \bX, \bP^\top \bX^\top \bX \bP+\bV^{1/2}\bE^\top \bE\bV^{1/2}+\\& \bP^\top \bX^\top \bE\bV^{1/2}+\bV^{1/2}\bE^\top \bX \bP)\mid \bX,\bY \\ \stackrel{d}{=}& \; (\bP^\top \widecheck{\bX}^\top \widecheck{\bY}+\bV^{1/2}\bE^\top \widecheck{\bY},\bP^\top \widecheck{\bX}^\top \widecheck{\bX}+\bV^{1/2}\bE^\top \widecheck{\bX}, \bP^\top \widecheck{\bX}^\top \widecheck{\bX}\bP+\bV^{1/2}\bE^\top \bE\bV^{1/2}+\\& \bP^\top \widecheck{\bX}^\top \bE\bV^{1/2}+\bV^{1/2}\bE^\top \widecheck{\bX}\bP)\mid \bX,\bY \\ \stackrel{}{=}& \; (\widetilde{\widecheck{\bX}}^\top \widecheck{\bY},\widetilde{\widecheck{\bX}}^\top \widecheck{\bX},\widetilde{\widecheck{\bX}}^\top \widetilde{\widecheck{\bX}})\mid \bX,\bY. 
\end{align*} 
Hence,
$$\mathcal{T}(\bX,\widetilde{\bX},\bY)\mid \bX,\bY  \stackrel{d}{=} \mathcal{T}(\widecheck{\bX},\widetilde{\widecheck{\bX}},\widecheck{\bY})\mid \bX,\bY.$$

\section{Construction of $[\widecheck{\bX} \ \widecheck{\bY}]$ via eigen-decomposition} \label{construction}

In this section, we give details on how to construct $[\widecheck{\bX} \ \widecheck{\bY}]$ such that $[\widecheck{\bX} \ \widecheck{\bY}]^\top [\widecheck{\bX} \ \widecheck{\bY}] = [\bX \ \bY]^\top [\bX \ \bY]$ using eigen-decomposition,
$$[\bX \ \bY]^\top [\bX \ \bY]=\bU\bD\bU^\top,\quad \text{where }\bU=[\textbf{u}_1\ \ldots\ \textbf{u}_{p+1}]\text{ is an orthogonal matrix, }\bD=\text{diag}(d_1,\ldots,d_{p+1}),$$
with $d_1\geq\cdots\geq d_{p+1}$.
% . Let $\bZ = [\bX \ \bY]$, $\widetilde{\bZ} = [\widecheck{\bX} \ \widecheck{\bY}]$ be our target, and $\bZ^\top \bZ = \bU\bD\bU^\top $ be the eigen-decomposition of $\bZ^\top \bZ$.
We consider two cases as follows. \\ \\
\noindent
\textbf{Case 1} ($n < p+1$): Since $\text{rank}([\bX \ \bY]^\top [\bX \ \bY]) \leq n$, we have $d_{n+1}=\cdots=d_{p+1}=0$ and 
$$[\bX \ \bY]^\top [\bX \ \bY] = \bU_1\bD_n\bU_1^\top ,$$
where $\bU_1=[\textbf{u}_1\ \ldots\ \textbf{u}_{n}]$, and $\bD_n=\text{diag}(d_1,\ldots,d_{n})$. Under this case, we let $[\widecheck{\bX} \ \widecheck{\bY}]=\bD_n^{1/2}\bU_1^\top$ such that $[\widecheck{\bX} \ \widecheck{\bY}]^\top [\widecheck{\bX} \ \widecheck{\bY}] = [\bX \ \bY]^\top [\bX \ \bY]$ is satisfied.
\\ \\
\noindent
\textbf{Case 2} ($n \ge p+1$): Under this case, we let $$[\widecheck{\bX} \ \widecheck{\bY}] = 
\begin{bmatrix} 
\bD^{1/2}\bU^\top  \\ \bzero_{(n-p-1) \times (p+1)}
\end{bmatrix}$$
such that
$$[\widecheck{\bX} \ \widecheck{\bY}]^\top [\widecheck{\bX} \ \widecheck{\bY}] = \bU\bD\bU^\top  = [\bX \ \bY]^\top [\bX \ \bY]$$ 
is satisfied.

\section{Computation of the tuning parameter $\lambda$ for the lasso-min method} 
\label{app:tuning_parameter}

Suppose we had access to individual level data such that Gaussian knockoffs $\widetilde{\bX}$ can be constructed, we can follow the method of \citet{variance-estimation} to estimate the noise level $\sigma$ by \begin{equation}
\label{eqn:sigma_hat}
\widehat{\sigma}_0=\sqrt{\text{max}\Lp\frac{2p+n+1}{n(n+1)}\norm{\bY}_2^2-\frac{\bY^\top \begin{bmatrix} \bX & \widetilde{\bX} \end{bmatrix}\bG^{-1}\begin{bmatrix} \bX & \widetilde{\bX} \end{bmatrix}^\top \bY}{n(n+1)},0\Rp},\quad \text{ where } \bG=\begin{bmatrix}
\bSigma & \bSigma-\bD 
\\ \bSigma-\bD & \bSigma
\end{bmatrix}.\end{equation} We could then compute $\lambda = \kappa \cdot \frac{\widehat{\sigma}_0}{n} \cdot \mathbb{E}[\norm{\bR^\top \bepsilon}_\infty]$, where $\bR \in \mathbb{R}^{n\times 2p}$ is a data matrix whose rows are i.i.d. samples from $\mathcal{N}(\bzero,\bG)$, and $\bepsilon \sim \mathcal{N}(\bzero,\bI_n)$ is independent of $\bR$. In the summary statistics setting, we replace $\widetilde{\bX}^\top \bY$ in \eqref{eqn:sigma_hat} by $\bP^\top \bX^\top \bY+\norm{\bY}_2\bZ$, where $\bP$ and $\bZ$ are obtained in Algorithm \ref{alg:fullghost}.

The expectation $\mathbb{E}[\norm{\bR^\top \bepsilon}_\infty]$ can be computed using Monte Carlo integration. However, when both $n$ and $p$ are very large, sampling $\bR$ and $\bepsilon$ becomes too time-consuming. Observing that $$\bR^\top \bepsilon=\bR^\top \frac{\bepsilon}{\norm{\bepsilon}_2}\norm{\bepsilon}_2$$ where 
$\bR^\top \frac{\bepsilon}{\norm{\bepsilon}_2}\sim \mathcal{N}(\bzero,\bG)$ and $\bepsilon$ are independent, we have $$\mathbb{E}[\norm{\bR^\top \bepsilon}_\infty]=\mathbb{E}[\norm{N(\bzero,\bG)}_\infty]\mathbb{E}[\norm{\bepsilon}_2]=\mathbb{E}[\norm{N(\bzero,\bG)}_\infty]\cdot \sqrt{2}\frac{\Gamma\{(n+1)/2\}}{\Gamma(n/2)}.$$ By Stirling's formula that
$${\displaystyle \Gamma (z)={\sqrt {\frac {2\pi }{z}}}\,{\left({\frac {z}{e}}\right)}^{z}\left(1+O\left({\frac {1}{z}}\right)\right),}$$
we have $$\sqrt{2}\frac{\Gamma\{(n+1)/2\}}{\Gamma(n/2)} \sim \sqrt{n}.$$
Therefore, we may approximate 
$$\mathbb{E}[\norm{\bR^\top \bepsilon}_\infty]\approx \sqrt{n}\;\mathbb{E}[\norm{\bL\bZ}_\infty],$$ 
where $\bL$ is the Cholesky decomposition of $\bG$ and $\bZ\sim \mathcal{N}(\bzero,\bI_{2p})$.

In practice, the simulated $\norm{\bL\bZ}_\infty$ usually concentrates around its mean as shown in Figure \ref{fig:concentration}. Thus, only several Monte Carlo samples are needed to accurately estimate $\mathbb{E}[\norm{\bL\bZ}_\infty]$, and we draw $10$ Monte Carlo samples throughout numerical experiments of this paper.

\begin{figure}
\centering
\includegraphics[width=0.8\textwidth]{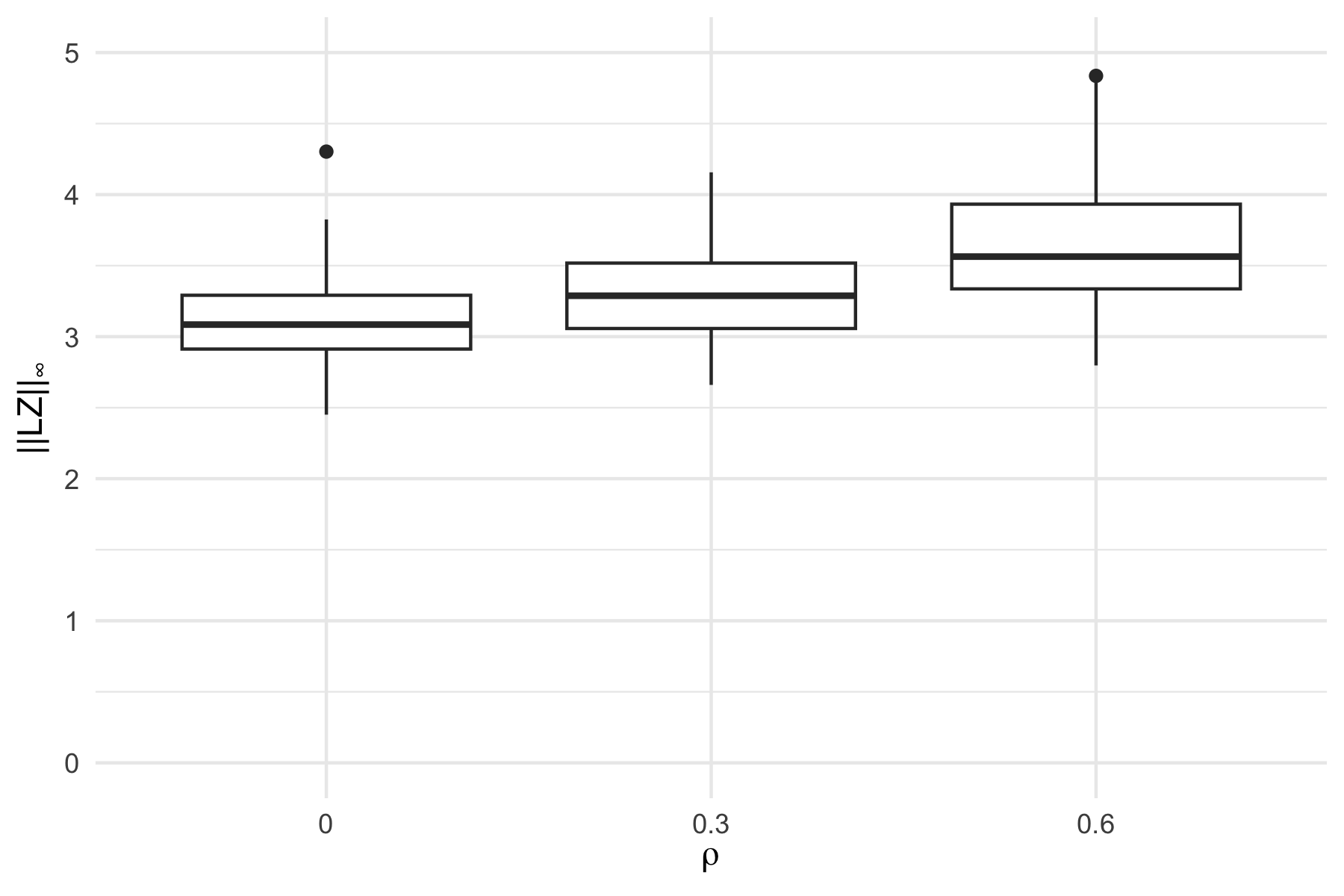}
\caption{Boxplots of 100 simulated samples $\norm{\bL\bZ}_\infty$ with $p=200$, $\bSigma_{i,j} = \rho^{|i-j|}$ and $\bD$ obtained via \eqref{eqn:sdp} for different $\rho$ values.}
\label{fig:concentration}
\end{figure}

Next, we prove that Algorithm \ref{alg:fullghost} maintains FDR control when $\lambda$ is computed as described in Section \ref{section:full_tuning}. By Proposition \ref{proposition:full_ghost}, it suffices to show that \eqref{eqn:fullghost} with the computed $\lambda$  produces feature importance statistics that satisfy the flip sign property. By $\lambda \approx \kappa \cdot\frac{\widehat{\sigma}_0}{n}\cdot\mathbb{E}[\norm{\bR^\top \bepsilon}_\infty]$, it suffices to show that $\hat{\sigma}_0$ is invariant to swapping variables with their knockoffs \citep{model-x}.

Let $\bPi_j \in \mathbb{R}^{2p\times 2p}$ be the permutation matrix that swaps the $j$-th column with the $(j+p)$-th column of a matrix. Thus, we have $\bPi_j^{-1}=\bPi_j^\top=\bPi_j$. 
This leads to 
\begin{align*}
    \bY^\top \begin{bmatrix} \bX & \widetilde{\bX} \end{bmatrix}\bPi_j(\bG)^{-1}(\begin{bmatrix} \bX & \widetilde{\bX} \end{bmatrix}\bPi_j)^\top \bY =& \bY^\top \begin{bmatrix} \bX & \widetilde{\bX} \end{bmatrix}\bPi_j(\bG)^{-1}\bPi_j^\top \begin{bmatrix} \bX & \widetilde{\bX} \end{bmatrix}^\top \bY\\=& \bY^\top \begin{bmatrix} \bX & \widetilde{\bX} \end{bmatrix}(\bPi_j\bG\bPi_j)^{-1} \begin{bmatrix} \bX & \widetilde{\bX} \end{bmatrix}^\top \bY\\=& \bY^\top \begin{bmatrix} \bX & \widetilde{\bX} \end{bmatrix}(\bG)^{-1}\begin{bmatrix} \bX & \widetilde{\bX} \end{bmatrix}^\top \bY,
\end{align*}
suggesting $\hat{\sigma}_0$ is invariant to swapping variables with their knockoffs \citep{model-x}. Therefore, the FDR of Algorithm \ref{alg:fullghost} is controlled when $\lambda$ is computed as described in Section \ref{section:full_tuning}.

Since all variables in $\widetilde{\bX}$ are null, in practice we may replace $\begin{bmatrix} \bX & \widetilde{\bX} \end{bmatrix}$ by $\bX$, $\bG$ by $\bSigma$ and $2p+n+1$ by $p+n+1$ in \eqref{eqn:sigma_hat} to reduce the dimension when estimating $\sigma$. Although this would, in theory, break the flip-sign property required for FDR control, no FDR inflation is observed in our simulations.

\section{Connection with the scout procedure} \label{app:scout}

In this section, we explain the connection of the feature importance statsitic defined in Algorithm \ref{alg:fullghost} and the scout procedure  \citep{scout}.

For covariates $X \in \bbR^p$ and response $Y\in \bbR$, \cite{scout} assume that $\begin{bmatrix}
X \\ Y
\end{bmatrix} \sim \mathcal{N}(\bzero,\bSigma_{X,Y})$. The population linear regression coefficient of $Y$ on $X$, which induces a linear predictor that achieves the minimal mean squared prediction error, is given by $\bbeta = -\bTheta_{XY}/\bTheta_{YY}$, where $\bTheta=\begin{bmatrix} \bTheta_{XX} & \bTheta_{XY} \\ \bTheta_{YX} & \bTheta_{YY} \end{bmatrix} = \bSigma_{X,Y}^{-1}$ is the precision matrix. Let $\bS$ be the empirical covariance matrix of $X$ and $Y$, they consider the following covariance-regularized regression approach to estimate $\bbeta$,

\begin{enumerate}
\item Compute $\hat{\bTheta}_{XX}$ to maximize $\log\{\det(\bTheta_{XX})\} - \text{tr}(\bS_{XX}\bTheta_{XX}) - J_1(\bTheta_{XX})$
\item Compute $\hat{\bTheta}$ to maximize $\log\{\det(\bTheta)\} - \text{tr}(\bS\bTheta) - J_2(\bTheta)$ subject to $\bTheta_{XX} = \hat{\bTheta}_{XX}$ obtained from Step 1.
\item Compute $\hat{\bbeta}=-\hat{\bTheta}_{XY}/\hat{\Theta}_{YY}$.
\item Compute $\hat{\bbeta}^*=c\hat{\bbeta}$ where $c$ is the regression coefficient of $\bY$ onto $\bX \hat{\bbeta}$.
\end{enumerate}
Here, $J_1$ and $J_2$ are two penalty functions. The first two steps are to appropriately separate true conditional correlations from those purely due to noise. As shown in \cite{scout}, when $J_2(\Theta)=\lambda_2{\norm{\bTheta}_1}$ (resp. $\lambda_2{\norm{\bTheta}_2^2}$), the solution to step 3 is proportional to the solution  of $$\hat{\bbeta} = \argmin_{\bbeta} \bbeta^\top \bG_{XX}\bbeta-2\bS_{XY}^\top \bbeta+\lambda_2\norm{\bbeta}_1 \; (\text{resp.} \; \lambda_2\norm{\bbeta}_2^2),$$
where $\bG_{XX}$ is the inverse of the solution $\hat{\bTheta}_{XX}$ from step 1. In other words, the Lasso corresponds to the setting that $J_1=0$ and $\bG_{XX}=\bS_{XX}$. \cite{scout} consider various settings in which they demonstrate the superiority of the scout procedure over the Lasso, Ridge and Elastic Net. 
In the setting of Section \ref{section:fullghost}, we have $\text{cov}(X,\widetilde{X})=\begin{bmatrix}
\bSigma & \bSigma-\bD \\
\bSigma-\bD & \bSigma
\end{bmatrix}$. Therefore, the objective function \eqref{eqn:fullghost} corresponds to the case that the true $\bTheta_{XX}$ is used in step 1 (here we include both $X$ and $\widetilde{X}$ as explanatory variables). 

\section{Construction of group knockoffs and examples of importance scores at the group level}\label{sec:group_knockoff_appendix}

For group knockoffs, we test the group conditional independence hypothesis:
\begin{align*}
H_{\gamma}^0: X_{\gamma} \indep Y \mid X_{-\gamma}
\end{align*} 
where $\gamma \in \{1,...,g\}$ denotes a group and $X_\gamma$ is the vector of features in group $\gamma$.

 In addition to the 
 conditional independence \eqref{Conditional independence}, group knockoffs $\widetilde{\bX}$ must satisfy 
 the group exchangeability condition that
$$\textbf{(Group exchangeability):} \; (\bX_\gamma, \widetilde{\bX}_\gamma, \bX_{-\gamma}, \widetilde{\bX}_{-\gamma})\stackrel{d}{=} (\widetilde{\bX}_\gamma, \bX_\gamma, \bX_{-\gamma}, \widetilde{\bX}_{-\gamma}),\; \forall \; \gamma \in \{1,...,g\}.$$
No exchangeability property is required for features within the same group, which allows greater flexibility in the construction knockoffs.

When $X \sim \mathcal{N}(\mathbf{0}, \bSigma)$, the group exchangeability condition allows the diagonal matrix $\bD = \text{diag}(\bs)$ described in Sections \ref{section:model-x} and \ref{section:fullghost} becomes a block-diagonal matrix $\bD = \text{diag}(\bS_1,...,\bS_g)$, where $\bS_\gamma$ is a symmetric matrix whose dimension equals the number of variables in group $\gamma$ ($\gamma \in \{1,...,g\}$). With the block-diagonal matrix $\bD$ obtained following the SDP construction of \cite{chu2023second} in step 3, Algorithm \ref{alg:gaussiankf} can construct valid group knockoffs $\widetilde{\bX}$ of $\bX$ with respect to $g$ feature groups. Analogously, Algorithms \ref{alg:ghostkf}-\ref{alg:fullghost} can also be modified correspondingly to perform inference of $H_{\gamma}^0$'s.

Although it is conceptually straightforward to modify $\bD$ from a diagonal to a block-diagonal matrix, note that doing so introduces significantly more optimization variables. To reduce computational burden in practice, we exploit a form of conditional independence across groups, described in section 4 of \citep{chu2023second}. The main idea is to select a few key variables in each group that capture most within-group variations, and perform a reduced optimization problem only on the key variables. In the real data analysis result, we defined groups via average-linkage hierarchical clustering with correlation cutoff $0.5$, selected representatives within groups via Algorithm A1 of \cite{chu2023second} with $c=0.5$, and replaced objective \eqref{eqn:sdp} by the maximum entropy (ME) objective, which has improved power over SDP constructions in simulations. 

In this paper, we use $M$ multi-knockoffs. To define the importance score for group $\gamma$ and its knockoffs, we sum the effect for variants in each group. With $M$ knockoff copies, we explicitly compute $Z_{\gamma} = \sum_{i \in \mathcal{A}_{\gamma}} |\beta_i|$ and $\widetilde{Z}_{\gamma}^{(\ell)} = \sum_{i \in \mathcal{A}_{\gamma}} |\widetilde{\beta}^{(\ell)}_{i}| $ for $\ell= 1,...,M$, where $\bbeta = (\bbeta, \widetilde{\bbeta}^1,...,\widetilde{\bbeta}^M)$ is the estimated effect sizes from step 4 of Algorithm \ref{alg:fullghost}. One may use other choices of feature importance such as the $l_2$ norm. The group-wise Lasso coefficient difference is then defined as 
\begin{align*}
    W_\gamma = (Z_\gamma - \operatorname{median}(\widetilde{Z}_\gamma^{(1)},...,\widetilde{Z}_\gamma^{(M)}))I_{Z_\gamma \ge \operatorname{max}(\widetilde{Z}^{(1)}_\gamma,...,\widetilde{Z}^{(M)}_\gamma)}
\end{align*}
and groups with $W_\gamma > \tau$ are selected, where $\tau$ is calculated from the multiple knockoff filter \citep{multipleknockoffs}. Note that $W_\gamma$ is the feature importance statistic first introduced in \cite{he2021identification}. 

\section{Ghostknockoffs for CRT (\textit{GhostCRT})} \label{ghostcrt}

Let $\bX\in \mathbb{R}^{n\times p}$ and $\bY\in\mathbb{R}^n$ be the covariate matrix and the response vector respectively. Recall that in the conditional randomization test, to test $H_j:X_j\indep Y\mid X_{-j}$, \cite{model-x} draw i.i.d. samples $\widetilde{\bX}_j^1,\ldots,\widetilde{\bX}_j^B\sim \mathcal{L}(\bX_j|\bX_{-j})$ ($\bX_j$ is the $j$-th column of the covariate matrix $\bX$) and compute the CRT $p$-value as \begin{equation}\label{CRT_pvalue}
    p_j =\frac{1}{B+1}\left[1+\sum_{b=1}^B\mathbbm{1}_{T(\widetilde{\bX}_j^b,\bX_{-j},\bY)\ge T(\bX_j,\bX_{-j},\bY) }\right],
\end{equation} for some feature importance function $T$. 

Under the assumption that rows of $\bX$ are i.i.d. samples from $\mathcal{N}(\bzero,\bSigma)$, we can generate $\widetilde{\bX}_j^1,\ldots,\widetilde{\bX}_j^B$ by \begin{equation}\label{CRT_j}
    \widetilde{\bX}_j^b=\bX_{-j}\bgamma_j+v_j^{1/2} \bE_j^b,
\end{equation} where $\bgamma_j=\bSigma_{-j,-j}^{-1}\bSigma_{-j,j}\in\mathbb{R}^{p-1}$, $v_j=\Sigma_{j,j}-\bSigma_{j,-j}\bSigma_{-j,-j}^{-1}\bSigma_{-j,j}$, and $\bE_j^1,\ldots,\bE_j^B \stackrel{iid}{\sim} \mathcal{N}(0,\bI_n)$ are independent of everything else.
% and are independent of each other for all $b \in \{1,2,...,B\}$.
Utilizing the analogy between \eqref{CRT_all}
and \eqref{CRT_j}, we develop the \textit{GhostCRT} with counterparts of Algorithms \ref{alg:ghostkf}-\ref{alg:partialghost} as follows, while the counterpart Algorithm \ref{alg:fullghost} is derived in the similar way.

\begin{algorithm}[H]
\caption{GhostKnockoffs with Marginal Correlation Difference Statistic for CRT}\label{alg:ghostkf_GhostCRT}
\begin{algorithmic}[1]
\STATE \textbf{Input}: $\bX^\top \bY$, $||\bY||_2^2$, and $\bSigma$.
\FOR{$j=1,\ldots,p$}
\STATE Compute $\bgamma_j=\bSigma_{-j,-j}^{-1}\bSigma_{-j,j}\in\mathbb{R}^{p-1}$ and $v_j=\Sigma_{j,j}-\bSigma_{j,-j}\bSigma_{-j,-j}^{-1}\bSigma_{-j,j}$.
\FOR{$b=1,\ldots,B$}
\STATE Generate $\widetilde{Z}^b_j = \bgamma_j^\top\bX_{-j}^\top \bY + ||\bY||_2Z_j^b$ where $Z_j^b\sim \mathcal{N}(0,v_j)$ and is independent of everything else.
\ENDFOR
\STATE Compute the CRT $p$-value $p_j$ via \eqref{CRT_pvalue} with $T(\bX_j,\bX_{-j},\bY)=|\bX_j^\top \bY|$ and $T(\widetilde{\bX}_j^b,\bX_{-j},\bY)=\widetilde{Z}^b_j$.
\ENDFOR
\STATE \textbf{Output}: Selection set by conducting existing multiple testing procedures on CRT $p$-values $p_1,\ldots,p_p$.
\end{algorithmic}
\end{algorithm}

\begin{algorithm}[H]
\caption{GhostKnockoffs with Penalized Regression for CRT: Known Empirical Covariance}\label{alg:partialghost_GhostCRT}
\begin{algorithmic}[1]
\STATE \textbf{Input}: $\bX^\top \bX, \bX^\top \bY, ||\bY||_2^2$, $\bSigma$, and $n$. \vspace{1mm}
\STATE Find $\widecheck{\bX}$ and $\widecheck{\bY}$ such that $[\widecheck{\bX}\ \widecheck{\bY}]^\top [\widecheck{\bX}\ \widecheck{\bY}] = [\bX\ \bY]^\top [\bX\ \bY]$ by eigen-decomposition or Cholesky decomposition.
\FOR{$j=1,\ldots,p$}
\STATE Compute $\bgamma_j=\bSigma_{-j,-j}^{-1}\bSigma_{-j,j}\in\mathbb{R}^{p-1}$ and $v_j=\Sigma_{j,j}-\bSigma_{j,-j}\bSigma_{-j,-j}^{-1}\bSigma_{-j,j}$.
\FOR{$b=1,\ldots,B$}
\STATE Generate $\widetilde{\widecheck{\bX}}_j^b$ via \eqref{CRT_j}
using $\widecheck{\bX}_{-j}$ as input.
\ENDFOR
\STATE Compute the CRT $p$-value $p_j$ via \eqref{CRT_pvalue} and replacing $\widetilde{\bX}_j^b$ by $\widetilde{\widecheck{\bX}}_j^b$ with feature importance statistic defined by
$T(\bX_j,\bX_{-j},\bY)=|\hat{\beta}_j|,$ where 
$$(\hat{\beta}_j,\hat{\bbeta}_{-j}) = \argmin_{(\beta_j,\bbeta_{-j}) \in \mathbb{R}^p}\frac{1}{2}||\bY-\bX_j{\beta}_j-\bX_{-j}{\bbeta}_{-j}||_2^2+\lambda||(\beta_j,\bbeta_{-j})||_1.$$
% $T(\bX_j,\bX_{-j},\bY)=|\bX_j^\top \bY|$ and $T(\widetilde{\bX}_j^b,\bX_{-j},\bY)=|\widetilde{Z}^b_j|$ ($b=1,\ldots,B$)
\ENDFOR
\STATE \textbf{Output}: Selection set by conducting existing multiple testing procedures on CRT $p$-values $p_1,\ldots,p_p$.
\end{algorithmic}
\end{algorithm}

As \eqref{CRT_j} is a special case of \eqref{CRT_all} where \begin{itemize}
    \item $\textbf{P}$ is obtained by substituting the $(j,j)$-entry and other entries in the $j$-th column of $\bI_p$ by $0$ and $\bgamma_j$ respectively;
    \item $\textbf{V}$ is a matrix of zeros expect the $(j,j)$-entry equals $v_j$,
\end{itemize}
all theoretical results in Sections \ref{section:model-x}-\ref{section:fullghost} remain true for the \textit{GhostCRT}.

\begin{remark}
In Algorithm \ref{alg:partialghost_GhostCRT}, the tuning parameter $\lambda$ is allowed to depend on $\bX^\top \bX$, $\bX^\top \bY$, $\bY^\top \bY$ and $n$. We may also use the square-root Lasso or the Lasso-max importance statistic as outlined in Sections \ref{GK-sqrtlasso} and \ref{GK-lassomax}.
\end{remark}

\section{Additional results for Section \ref{subsubsec:full_ghost_indep}} \label{app:add_full_gk_plots}
To further demonstrate the effect of sample size on the new GhostKnockoffs methods in comparison to the individual level knockoffs with (cross-validated) Lasso coefficient difference, we consider additional experiments with $p=600$ and $n=600/1800/3000$ under the same setting of Section \ref{subsubsec:full_ghost_indep}. Note that the noise level scales in the order of $\sqrt{n}$ such that the signal to noise ratio does not change dramatically. From Figure \ref{fig:full_ghost_increased_n_indep}, we observe that as $n$ increases, all three new methods proposed in this paper have comparable power with KF-lassocv and outperform GK-marginal \citep{ghostknockoffs} consistently, with FDR controlled in all cases.

\begin{figure}[htbp] 
  \centering
  \includegraphics[width=0.8\textwidth]{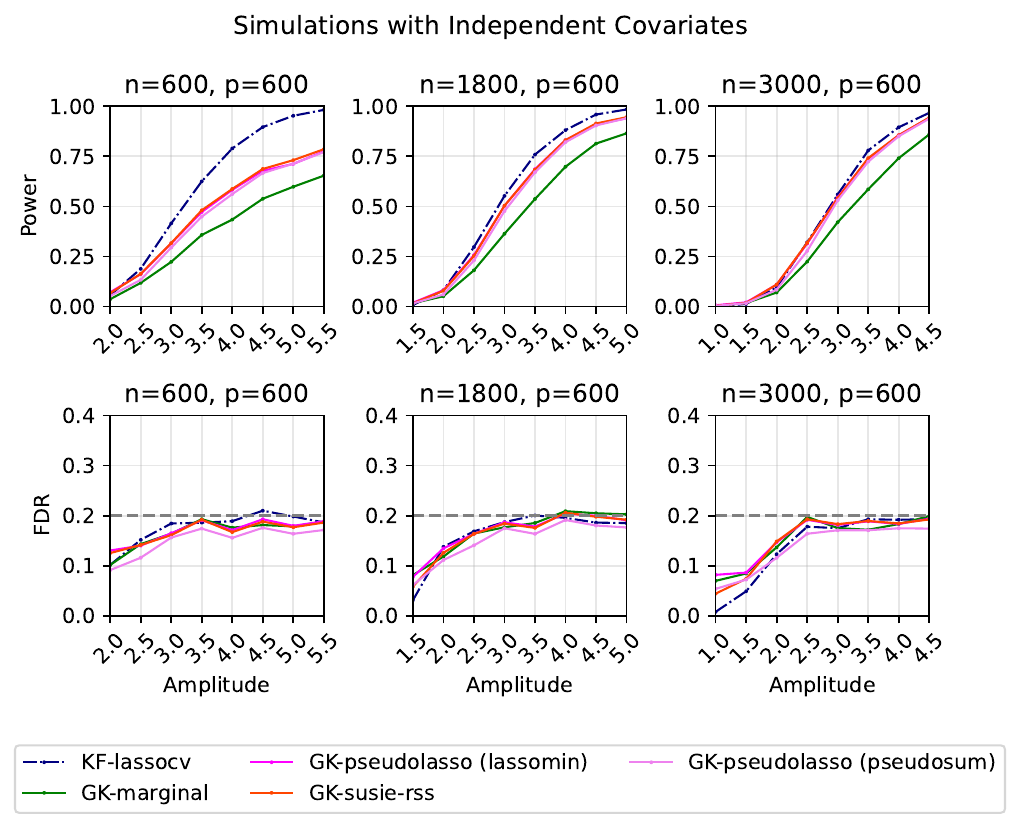}
  \caption{Power and FDR plots for independent features and a Gaussian linear model with varying sample sizes and a fixed feature dimension. Each point shown is an average over 200 replications.}
  \label{fig:full_ghost_increased_n_indep}
\end{figure}

\section{Supplementary plots for Section \ref{subsubsec:full_ghost_real}} \label{app:supplementary_semi_plots}

Let $\bZ = \bX^\top \bY$ and $\tilde{\bZ} = \bP^\top \bX^\top \bY+\norm{\bY}_2\bZ$ be defined as in Algorithm \ref{alg:fullghost}. Note that for Algorithm \ref{alg:fullghost} to control the FDR, it suffices to require that swapping the $j-$th entry of $\bZ$ with the $j-$th entry of $\tilde{\bZ}$ does not change the joint distribution of $(\bZ, \tilde{\bZ})$ for each null $j$.

By the Central Limit Theorem, short sets of entries (e.g. single entries, pairs, and triples etc.) of $(\bZ, \tilde{\bZ})$ are approximately Gaussian. Additionally, in Figure \ref{fig:exchangeability_Z}, we show empirically that the covariance of $(\bZ, \tilde{\bZ})$ (approximately) satisfies the required swap-invariance for null positions. These approximations, coupled with the robustness of the knockoff framework, empirically yield the FDR control. This is similar to the robustness of second-order knockoffs observed empirically in \citet{model-x}.

In the setting from Section \ref{subsubsec:full_ghost_real}, Figure \ref{fig:normality_Z} depicts the ordered empirical values of $Z_j$ (respectively $\tilde{Z}_j$) plotted against an equal-size ordered random sample from a Gaussian distribution with matching mean and variance as the empirical mean and variance of $Z_j$ (respectively $\tilde{Z}_j$), for three randomly selected indices. This comparison is based on the 1000 sub-sampled data replications from Section \ref{subsubsec:full_ghost_real}. In Figure \ref{fig:normality_Z_all}, we overlay the plots for all indices. We observe that $Z_j$ and $\tilde{Z}_j$ approximately follow Gaussian distributions. In Figure \ref{fig:exchangeability_Z}, we present the scatter plots of relevant empirical covariances. We observe that all the points roughly concentrate around the $y=x$ line. This shows the approximate swap-invariance of $\bZ$ and $\tilde{\bZ}$ (for null indices).

\begin{figure}[htbp] 
  \centering
  \includegraphics[width=0.8\textwidth]{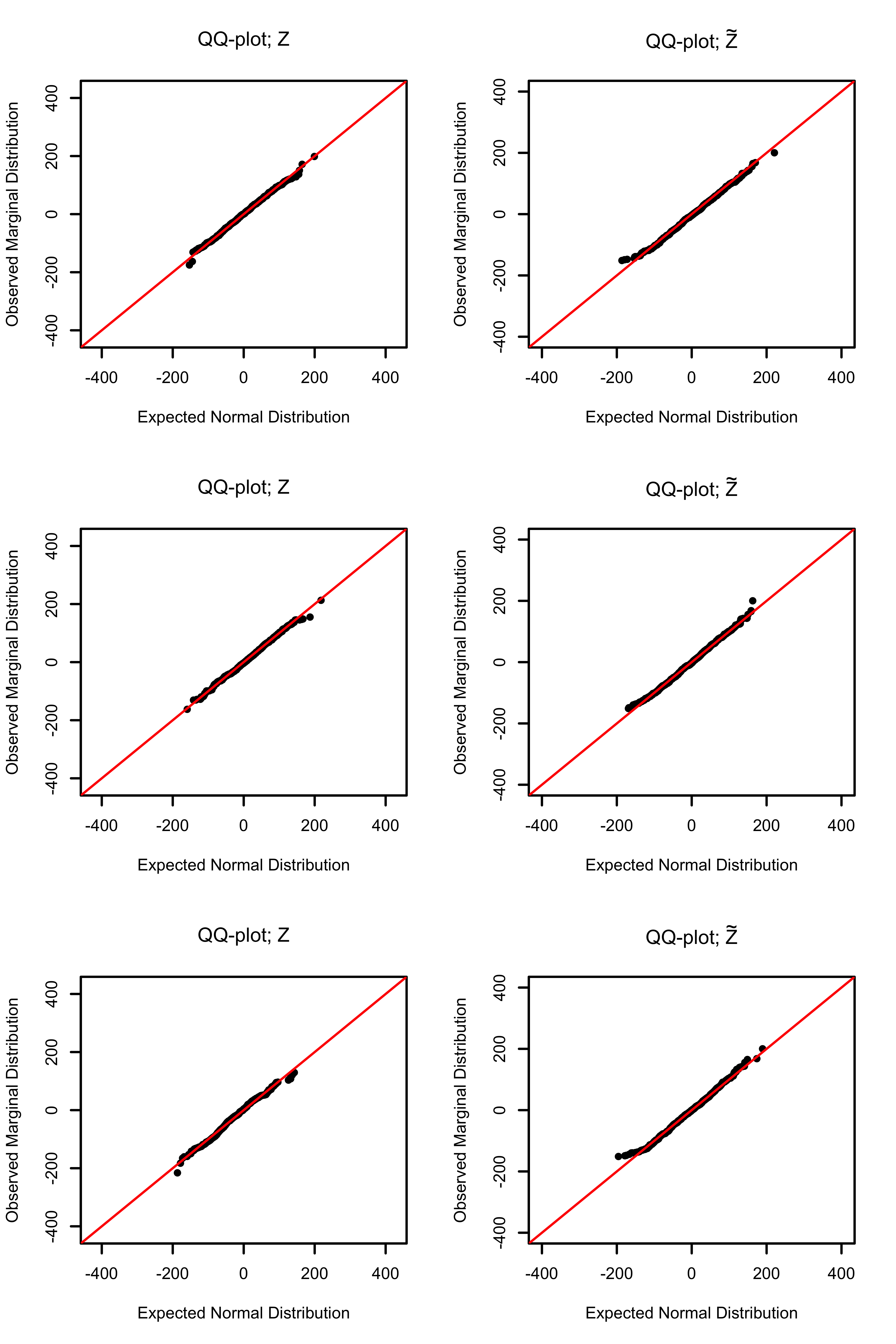}
  \caption{QQ plots of $Z_j$'s (left) and $\tilde{Z}_j$'s (right) against Gaussian samples with matching mean and variance for three randomly sampled indices.}
  \label{fig:normality_Z}
\end{figure}

\begin{figure}[htbp] 
  \centering
  \includegraphics[width=0.8\textwidth]{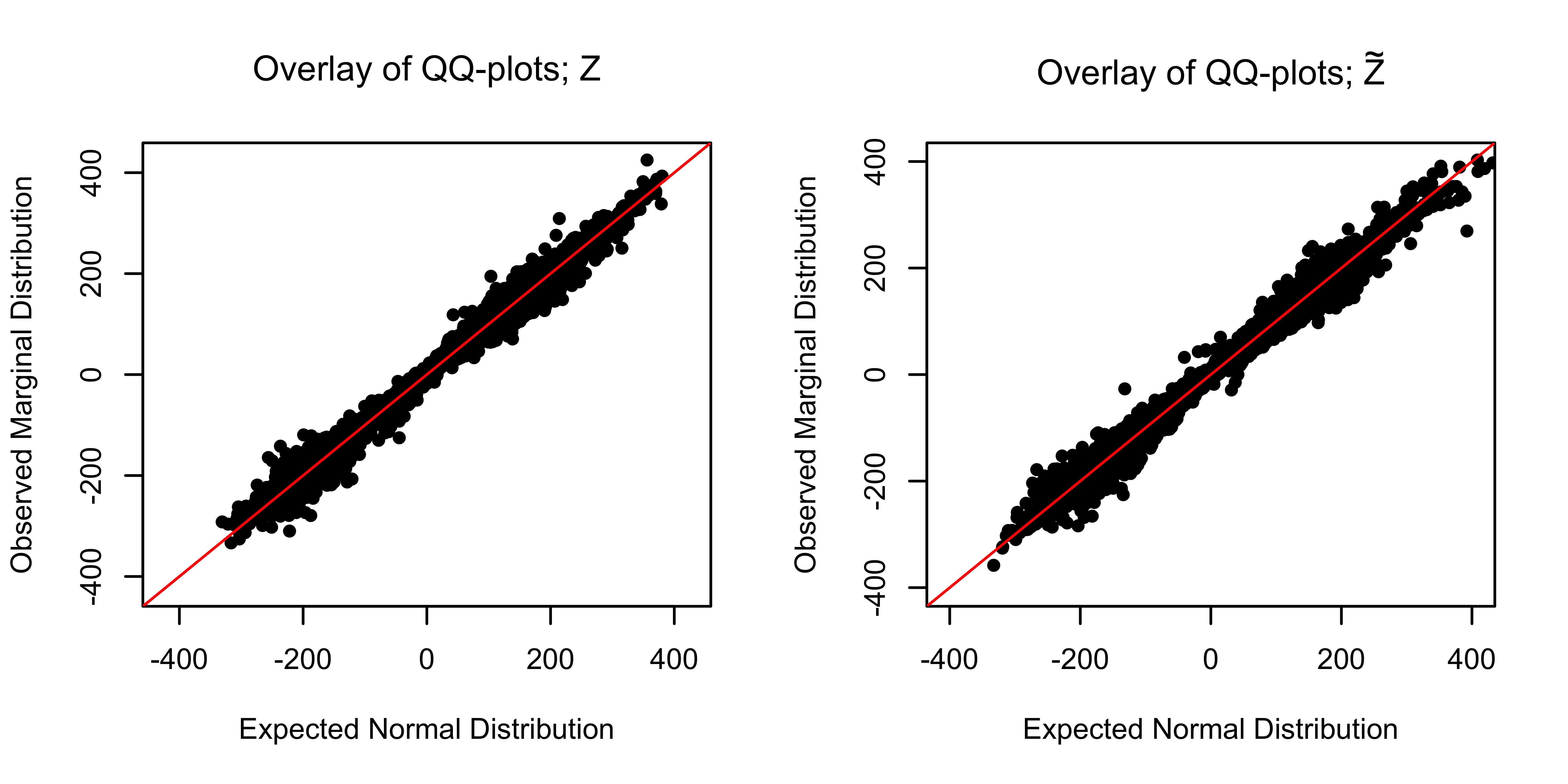}
  \caption{QQ plots of $Z_j$'s (left) and $\tilde{Z}_j$'s (right) against Gaussian samples with matching mean and variance for all indices overlaid.}
  \label{fig:normality_Z_all}
\end{figure}

\begin{figure}[htbp] 
  \centering
  \includegraphics[width=0.8\textwidth]{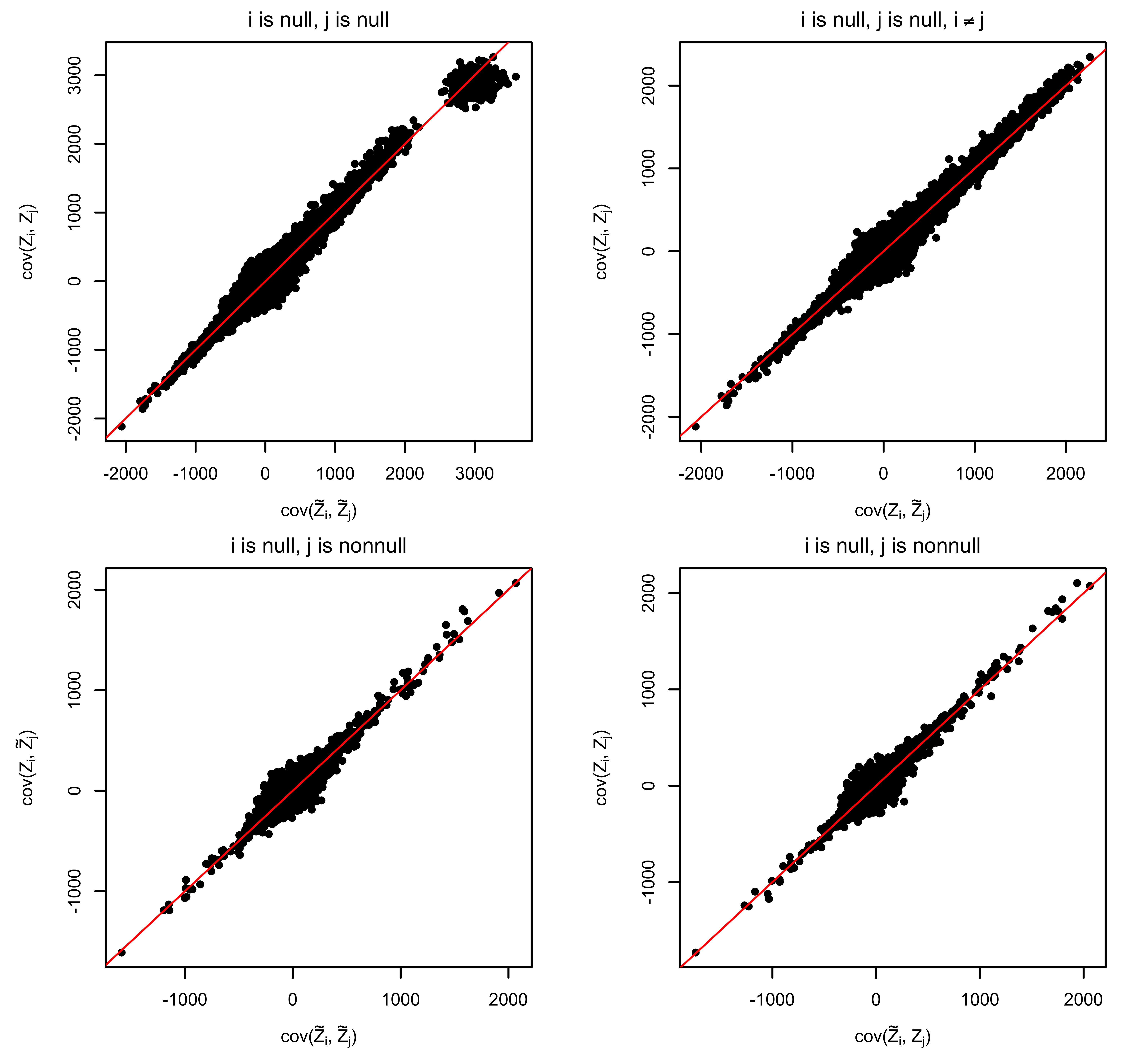}
  \caption{Scatter plots for relevant empirical covariances. Each correlation is estimated from 1000 samples drawn as described in Section \ref{subsubsec:full_ghost_real}.}
  \label{fig:exchangeability_Z}
\end{figure}

\section{Details of the nine studies in Section \ref{section:gwas}}
\label{app:nine_studies}

Section \ref{section:gwas} considers the following nine studies for Alzheimer's disease:
\begin{enumerate}
\item[1.] The genome-wide association study performed by \citet{huang2017common}.
\item[2.] The genome-wide meta-analysis of clinically diagnosed AD and AD-by-proxy by \citet{jansen2019genome}.
\item[3.] The genome-wide meta-analysis of clinically diagnosed AD by \citet{kunkle2019genetic}.
\item[4.] The genome-wide meta-analysis by \citet{schwartzentruber2021genome}.
\item[5.] In-house genome-wide associations study of 15,209 cases and 14,452 controls aggregating 27 cohorts across 39 SNP array data sets, imputed using the TOPMed reference panels \citep{belloy2022challenges}.
\item[6.] A whole-exome sequencing analyse of data from ADSP by \citet{bis2020whole}.
\item[7.] A whole-exome sequencing analyse of data from ADSP by \citet{le2021novel}.
\item[8.] In-house whole-exome sequencing analysis of ADSP (6155 cases, 5418 controls). 
% (\citep{belloy2022fast} and \citep{belloy2022challenges}).
\item[9.] In-house whole-genome sequencing analysis of the 2021 ADSP release (3584 cases, 2949 controls) \citep{belloy2022fast}. 
\end{enumerate}

\section{Calculation of meta-analysis Z-score} \label{app:meta_z_scores}

Based on $Z$-scores $\bZ_1,\bZ_2,...,\bZ_K$ from $K$ studies, we adopt the definition in \citet{ghostknockoffs} that meta-analysis $Z$-score with overlapping samples is $$\bZ_{\text{meta}} = \bH\sum_{k=1}^Kw_k\bC_k\bZ_k.$$ Specifically,
\begin{itemize}
    \item  optimal weights $w_1,\ldots,w_K$ are obtained by solving the optimization problem $$\begin{aligned}
\textrm{minimize }  & \sum_{k=1}^{K}\sum_{l=1}^{K} w_kw_lcor.S_{kl}\quad
\textrm{subject to }\sum_{k=1}^{K}w_k\sqrt{n_k}=1,\text{ }w_1,\ldots,w_K \ge 0;\\
\end{aligned}$$
\item for $k=1,\ldots,K$, $\bC_k = \text{diag}\{c_{k1},...,c_{kp}\}$  is a diagonal matrix where $c_{kj}=1$ if $Z$-score of the $j$-th variant is observed in the $k$-th study and $c_{kj}=0$ otherwise ($j=1,\ldots,p$);
\item $\bH=\text{diag}\{h_1,...,h_p\}$ is a diagonal matrix where $h_j=(\sum_{k}\sum_{l} w_kw_lc_{kj}c_{lj}cor.S_{kl})^{-1/2}$ ($j=1,\ldots,p$);
\item $cor.S_{kl}$ is the study correlation between the $k$-th study and the $l$-th study.
\end{itemize}

In practice, when calculating $cor.S_{kl}$, we only use variants whose $Z$-scores are bounded in $[-1.96,1.96]$ in both the $k$-th study and the $l$-th study to eliminate the impact of polygenic effects. This meta-analysis approach is a generalization of the METAL method proposed by \citet{willer2010metal}.

\section{Obtaining the covariance matrix $\bSigma$ in meta-analysis for AD}\label{sec:LDmatrices}

To perform meta-analysis for AD, we need a suitable estimate of the covariance matrix $\bSigma$. In this paper, we adopt strategies in \citet{zihuai2023app} and \citet{chu2023second} as follows.

We first download the covariance matrix from the Pan-UKB consortium (\url{https://pan.ukbb.broadinstitute.org}), which contains about $24$ million variants across the human genome derived from about $ 500,000$ British samples. We then extract $p=650,576$ variants which satisfy the following three conditions: (a) the variant is recorded in the UK Biobank genotype array, (b) its MAF exceeds 0.01, (c) its reference/alternate allele pair matches with the ones listed in all the nine studies in meta-analysis. Based on the covariance matrix of extracted variants, we further partition extracted variants into 1703 quasi-independent blocks using the partition given by  \citet{berisa2016approximately}. Finally, we compute the block-diagonal covariance matrix
$$\bSigma=\begin{bmatrix}
        \bSigma_1 & & \\
        & \ddots & \\
        & & \bSigma_{1703}
    \end{bmatrix},$$
where $\bSigma_l$ is the shrinkage estimator of the covariance matrix of variants in the $l$-th block using the R package \texttt{corpcor} \citep{schafer2005shrinkage}. To ensure that all blocks $\bSigma_1,...,\bSigma_{1703}$ are positive definite, we perform eigen-decomposition and increase all their eigenvalues not larger than $ 10^{-5}$ to $10^{-5}$.

\section{Supplementary tables of meta-analysis for AD} \label{app:gwas_add_tables}

Tables \ref{table:real_data_table1} and \ref{table:real_data_table2} provide more details of the meta-analysis for AD in Section \ref{section:gwas}. Specifically, Table \ref{table:real_data_table1} presents the number of loci, average signals per locus, standard deviation of the number of signals per locus, average groups per locus, and standard deviation of the number of groups per locus identified by conventional marginal association test, GK-marginal, GK-pseudolasso, and GK-susie-rss. Here, the $p$-value threshold of the conventional marginal association test is $5\times 10^{-8}$, and GK-pseudolasso uses the tuning parameter chosen by the lasso-min method in Section \ref{section:full_tuning}. For GK-marginal, GK-pseudolasso, and GK-susie-rss, we display results with respect to target FDR levels 0.05, 0.1, and 0.2. Table \ref{table:real_data_table2} provides details of top variants of identified loci given by GK-pseudolasso (target FDR level: 0.20), including their positions (columns ``Chr." and ``SNP"), their reference alleles (column ``Ref.") and alternative alleles (column ``Alt."), genes that they potentially regulate (column ``TopS2GGene"), their closest genes, their $Z$-scores from different individual studies, their meta-analysis $Z$-scores, their feature importance scores ($W$), and the marginal $p$-values obtained from meta-analysis $Z$-scores. 

\vspace{2mm}
\begin{table}[h]
\caption{Summary of results by applying different methods on meta-analysis for AD}
\resizebox{\columnwidth}{!}{\begin{tabular}{lrrrrrr}
 \hline
 \multirow{2}*{Method}&Target&Number of &Average signals& SD of signals & Average groups& SD of groups \\
 &FDR level&identified loci&per locus &per locus & per locus &per locus\\
 \hline
 \multicolumn{2}{l}{Marginal association test} &29 & 15.517 & 32.231 & 1.000 & 0.000\\
 \hline
 \multirow{3}*{GK-marginal}&0.05&3&107.667&97.027&3.667&3.055\\
 &0.10&10 & 95.600 & 214.568 & 2.700 & 3.622\\
 &0.20&17 & 76.176 & 174.714 & 2.412 & 3.318\\
 \hline
 \multirow{3}*{GK-pseudolasso}&0.05&30 & 21.500 & 44.323 & 2.333 & 4.722\\
 &0.10&42 & 17.214 & 38.019 & 2.024 & 4.015\\
 &0.20&63 & 15.889 & 47.074 & 1.794 & 3.561\\
 \hline
 \multirow{3}*{GK-susie-rss}&0.05&21 & 14.619 & 27.902 & 1.286 & 0.644\\
 &0.10&35 & 12.000 & 23.506 & 1.257 & 0.657\\
 &0.20&47 & 10.191 & 20.591 & 1.319 & 0.695\\
 \hline
\end{tabular}}
\label{table:real_data_table1}
\end{table}

\title{}
	%\date{\today}
	%\date{\today}
%	\author[1]{A}
%	\author[1]{B}
%	\author[1]{C}
%	\affil[1]{Affiliation
%	}
	\date{}                     %% if you don't need date to appear
	\maketitle
		\renewcommand{\arraystretch}{1}
	\begin{table}[h]
		\centering
		\caption{Details of top variants of identified loci given by GK-pseudolasso (target FDR level: 0.20).}
		\resizebox{\columnwidth}{!}{
			\begin{tabular}{rrrrrrrrrrrrrrrrrr}
				\toprule
				\multirow{2}{*}{Chr.}&\multirow{2}{*}{SNP}&\multirow{2}{*}{Ref.}&\multirow{2}{*}{Alt.}&\multirow{2}{*}{TopS2GGene}&\multirow{2}{*}{Closest gene}&\multicolumn{9}{c}{$Z$-scores from different individual studies}&Meta-analysis&\multirow{2}{*}{$W$}&Marginal\\
				\cline{7-15}
				&&&&&&Study 1&Study 2&Study 3&Study 4&Study 5&Study 6&Study 7&Study 8&Study 9&$Z$-scores&&$p$-values\\
				\midrule

1&20853688&C&T&EIF4G3&EIF4G3&2.72&3.93&3.29&4.29&2.89&--&--&0.68&1.65&4.95&2.516$\times10^{-3}$&3.697$\times10^{-7}$\\
1&200984367&A&G&KIF21B&KIF21B&-2.21&-3.73&-4.33&-3.60&-2.92&--&--&--&-0.67&-4.36&1.837$\times10^{-3}$&6.490$\times10^{-6}$\\
1&207611623&A&G&CR1&CR1&-4.84&-8.81&-7.97&-9.65&-6.37&--&--&--&-3.96&-10.97&1.228$\times10^{-2}$&2.802$\times10^{-28}$\\
2&37270395&G&A&--&NDUFAF7&1.50&4.02&3.73&4.03&2.05&--&--&--&--&4.62&1.998$\times10^{-3}$&1.953$\times10^{-6}$\\
2&44026309&T&C&--&LRPPRC&-1.23&-3.80&-1.88&-3.55&-2.26&--&--&--&-3.64&-4.37&2.089$\times10^{-3}$&6.208$\times10^{-6}$\\
2&65409567&G&A&--&SPRED2&--&-3.93&-2.41&-4.05&-0.23&--&--&--&0.15&-4.44&1.993$\times10^{-3}$&4.538$\times10^{-6}$\\
2&105805908&T&C&--&NCK2&0.10&-3.94&-2.80&-4.67&-2.08&--&--&--&--&-4.72&2.490$\times10^{-3}$&1.185$\times10^{-6}$\\
2&127136908&A&T&BIN1&BIN1&3.77&10.94&8.68&11.95&8.74&--&--&--&4.90&13.36&1.141$\times10^{-2}$&5.466$\times10^{-41}$\\
2&233117202&G&C&NGEF&INPP5D&2.03&6.15&5.16&6.42&2.29&--&--&--&2.40&7.21&4.762$\times10^{-3}$&2.826$\times10^{-13}$\\
3&136105288&G&A&SLC35G2&PPP2R3A&-1.24&-3.66&-2.03&-4.64&-1.98&--&--&--&-3.04&-4.84&2.250$\times10^{-3}$&6.607$\times10^{-7}$\\
4&11024404&A&G&--&CLNK&-2.47&-6.00&-4.06&-6.50&-4.32&--&--&--&-2.52&-7.32&5.039$\times10^{-3}$&1.275$\times10^{-13}$\\
4&71303158&G&A&--&SLC4A4&-2.60&-3.77&-3.65&-3.34&-1.49&--&--&--&-1.53&-4.28&1.817$\times10^{-3}$&9.466$\times10^{-6}$\\
4&112082387&A&C&--&FAM241A&2.47&4.82&1.76&3.36&0.40&--&--&--&-0.71&4.68&2.478$\times10^{-3}$&1.418$\times10^{-6}$\\
4&143428212&C&T&--&GAB1&--&-3.68&-2.84&-3.95&-1.56&--&--&--&-1.36&-4.37&2.017$\times10^{-3}$&6.081$\times10^{-6}$\\
4&158808801&G&A&RAPGEF2&FNIP2&-2.43&-3.95&-2.39&-3.82&-2.40&--&--&--&-1.91&-4.66&1.894$\times10^{-3}$&1.553$\times10^{-6}$\\
5&4068226&C&T&--&IRX1&--&4.53&1.20&3.62&0.34&--&--&--&-0.91&4.50&2.144$\times10^{-3}$&3.323$\times10^{-6}$\\
5&14707491&C&T&ANKH&ANKH&-3.92&-3.20&-3.95&-4.36&-3.60&--&--&--&0.60&-4.66&2.480$\times10^{-3}$&1.602$\times10^{-6}$\\
5&86923485&A&G&--&LINC02059&2.49&4.70&2.45&3.76&3.49&--&--&--&2.49&5.12&3.059$\times10^{-3}$&1.517$\times10^{-7}$\\
5&177559423&G&A&RAB24&FAM193B&1.96&3.85&3.99&4.16&2.48&--&--&--&1.38&4.71&2.313$\times10^{-3}$&1.248$\times10^{-6}$\\
5&179373099&C&T&--&ADAMTS2&1.52&2.54&4.29&4.80&3.12&--&--&--&2.35&4.36&1.938$\times10^{-3}$&6.512$\times10^{-6}$\\
6&935171&T&C&--&LINC01622&-2.80&-3.20&-3.33&-4.55&-3.37&--&--&--&-2.17&-4.75&2.380$\times10^{-3}$&1.040$\times10^{-6}$\\
6&32686937&T&C&HLA-DQA2&HLA-DQB1&-3.88&-6.46&-4.86&-7.53&-2.29&--&--&--&-1.10&-8.13&4.461$\times10^{-3}$&2.090$\times10^{-16}$\\
6&41066261&G&C&OARD1&OARD1&2.69&3.78&6.91&7.12&4.06&--&--&--&--&6.37&5.558$\times10^{-3}$&9.364$\times10^{-11}$\\
6&47484147&C&T&CD2AP&CD2AP&2.95&5.74&5.21&6.10&5.33&--&--&--&2.24&7.05&3.271$\times10^{-3}$&8.942$\times10^{-13}$\\
7&1543652&A&G&TMEM184A&MAFK&2.33&4.06&2.93&3.64&2.36&--&--&--&0.33&4.54&1.810$\times10^{-3}$&2.868$\times10^{-6}$\\
7&37842715&G&A&--&NME8&2.95&4.15&3.81&3.74&3.20&--&--&--&1.13&4.79&2.045$\times10^{-3}$&8.230$\times10^{-7}$\\
7&100406823&C&T&--&ZCWPW1&4.25&7.53&4.01&8.41&5.04&--&--&3.59&1.29&9.35&8.987$\times10^{-3}$&4.266$\times10^{-21}$\\
7&143410495&G&T&EPHA1-AS1&EPHA1&1.19&6.56&4.37&6.81&2.70&--&--&--&1.63&7.52&3.795$\times10^{-3}$&2.751$\times10^{-14}$\\
8&27362470&C&T&PTK2B&PTK2B&3.84&6.79&6.12&7.94&5.19&--&--&--&2.12&8.70&4.345$\times10^{-3}$&1.668$\times10^{-18}$\\
8&95041772&C&T&--&NDUFAF6&4.06&3.96&4.03&4.50&2.81&--&--&--&0.36&5.17&2.207$\times10^{-3}$&1.172$\times10^{-7}$\\
8&97359646&A&G&--&SNORD3H&2.70&3.01&3.70&3.99&2.42&--&--&--&1.30&4.25&1.767$\times10^{-3}$&1.067$\times10^{-5}$\\
8&102564430&G&A&--&ODF1&1.72&4.00&2.66&3.53&1.42&--&--&--&-0.48&4.29&1.855$\times10^{-3}$&8.825$\times10^{-6}$\\
8&111515902&C&T&--&LINC02237&--&4.13&0.44&3.77&-0.41&--&--&--&--&4.40&2.051$\times10^{-3}$&5.387$\times10^{-6}$\\
8&144042819&T&C&PARP10&SPATC1&0.17&4.66&2.47&4.57&3.68&--&--&--&--&5.16&3.389$\times10^{-3}$&1.210$\times10^{-7}$\\
10&29966853&G&A&--&JCAD&--&3.72&2.05&4.56&1.31&--&--&--&0.48&4.68&2.501$\times10^{-3}$&1.443$\times10^{-6}$\\
10&42722997&T&C&--&LOC283028&0.39&4.79&2.57&4.34&1.14&--&--&--&0.25&5.02&2.128$\times10^{-3}$&2.616$\times10^{-7}$\\
10&59962515&T&G&--&LINC01553&1.43&3.63&3.30&5.14&3.48&--&--&--&3.17&5.18&2.031$\times10^{-3}$&1.130$\times10^{-7}$\\
10&80494228&C&T&TSPAN14&TSPAN14&3.23&3.22&4.17&5.83&2.03&--&--&--&--&5.35&2.041$\times10^{-3}$&4.295$\times10^{-8}$\\
11&60254475&G&A&--&MS4A4E&-5.74&-7.97&-8.27&-9.09&-6.66&--&--&--&-3.32&-10.30&8.499$\times10^{-3}$&3.570$\times10^{-25}$\\
11&65888811&G&A&FIBP&FIBP&-2.13&-4.59&-1.22&-3.57&-1.62&--&--&--&-0.38&-4.74&2.589$\times10^{-3}$&1.070$\times10^{-6}$\\
11&86156833&A&G&PICALM&PICALM&6.78&8.67&8.07&10.55&5.11&--&--&--&3.08&11.50&1.074$\times10^{-2}$&6.418$\times10^{-31}$\\
11&121578263&T&C&--&SORL1&-3.10&-4.40&-3.82&-5.59&-3.38&--&--&--&-0.52&-5.90&3.920$\times10^{-3}$&1.768$\times10^{-9}$\\
13&43679792&C&T&--&ENOX1&0.19&3.79&1.22&4.28&0.01&--&--&--&-1.03&4.30&1.865$\times10^{-3}$&8.441$\times10^{-6}$\\
13&93594511&A&T&--&GPC6-AS2&--&-0.04&-1.09&-0.85&-2.34&--&--&--&-0.57&-0.62&7.282$\times10^{-2}$&2.672$\times10^{-1}$\\
14&32478306&T&C&AKAP6&AKAP6&-1.45&-4.35&-1.77&-3.63&-0.28&--&--&--&0.77&-4.44&1.869$\times10^{-3}$&4.449$\times10^{-6}$\\
14&52924962&A&G&--&FERMT2&4.68&4.58&4.97&6.27&2.90&--&--&--&1.32&6.58&4.682$\times10^{-3}$&2.429$\times10^{-11}$\\
14&92470949&C&T&--&SLC24A4&-3.83&-6.10&-5.16&-6.67&-2.90&--&--&--&-2.58&-7.57&4.647$\times10^{-3}$&1.836$\times10^{-14}$\\
15&50735410&C&T&HDC&SPPL2A&-3.16&-4.81&-4.09&-6.02&-2.45&--&--&--&0.09&-6.29&5.133$\times10^{-3}$&1.547$\times10^{-10}$\\
15&58753575&A&G&--&ADAM10&-2.86&-5.90&-4.16&-5.97&-2.81&--&--&--&-2.16&-6.94&3.385$\times10^{-3}$&1.910$\times10^{-12}$\\
15&63277703&C&T&APH1B&APH1B&1.20&5.52&3.68&5.72&2.58&2.46&1.61&0.98&2.05&6.45&3.285$\times10^{-3}$&5.482$\times10^{-11}$\\
16&31120929&A&G&KAT8&KAT8&-2.28&-5.50&-2.72&-5.84&-2.89&--&--&--&-1.45&-6.56&3.913$\times10^{-3}$&2.702$\times10^{-11}$\\
17&5233752&G&A&SCIMP&SCIMP&3.30&6.04&3.82&5.48&1.93&--&--&--&2.40&6.79&3.297$\times10^{-3}$&5.560$\times10^{-12}$\\
17&7581494&G&A&CD68&LOC100996842&-1.82&-3.60&-1.57&-3.49&-3.37&-1.95&-1.61&-2.72&-3.18&-4.42&1.933$\times10^{-3}$&4.941$\times10^{-6}$\\
17&49219935&T&C&ABI3&ABI3&--&-4.94&--&--&-4.75&-2.68&0.20&--&-2.61&-5.25&2.982$\times10^{-3}$&7.430$\times10^{-8}$\\
17&58331728&G&C&BZRAP1&MIR142&-1.00&-4.94&-5.09&-5.12&-3.81&--&--&--&-1.35&-5.75&3.909$\times10^{-3}$&4.412$\times10^{-9}$\\
17&63482562&C&T&ACE&ACE&2.73&5.07&3.54&5.25&3.92&1.93&2.67&2.09&2.45&6.32&5.299$\times10^{-3}$&1.268$\times10^{-10}$\\
19&1058177&A&G&--&ABCA7&-0.93&-4.61&-2.73&-4.94&-3.96&-1.16&-1.48&-0.38&0.52&-5.45&4.973$\times10^{-3}$&2.534$\times10^{-8}$\\
19&6876985&T&C&VAV1&ADGRE1&1.05&3.04&3.58&4.42&1.59&--&--&--&0.42&4.21&2.119$\times10^{-3}$&1.254$\times10^{-5}$\\
19&44888997&C&T&PVRL2&NECTIN2&20.83&51.85&--&--&--&--&--&--&--&53.66&8.573&0.000\\
19&51224706&C&A&CD33&CD33&-3.40&-5.84&-5.09&-5.69&-3.76&--&--&--&-3.97&-6.x96&4.936$\times10^{-3}$&1.696$\times10^{-12}$\\
19&54664811&A&G&LILRB4&LILRB4&-2.61&-3.61&-3.13&-3.89&-1.05&--&--&--&0.54&-4.37&1.958$\times10^{-3}$&6.300$\times10^{-6}$\\
20&56409712&G&T&CASS4&CASS4&-3.82&-5.84&-4.56&-6.07&-5.14&--&--&--&--&-7.12&6.582$\times10^{-3}$&5.526$\times10^{-13}$\\
21&26775872&C&T&ADAMTS1&ADAMTS1&-1.60&-2.90&-5.17&-5.54&-3.39&--&--&--&-0.22&-4.87&2.469$\times10^{-3}$&5.668$\times10^{-7}$\\

				\bottomrule
			\end{tabular}
		}
            \label{table:real_data_table2}
	\end{table}

\section{Supplementary figures of meta-analysis for AD} \label{app:gwas_add_figures}

Analogous to Figure \ref{fig:real_data_ghostlasso} in Section \ref{section:gwas}, Figures \ref{fig:real_data_gwas},  \ref{fig:real_data_marginal} and \ref{fig:real_data_susie} respectively present Manhattan plots of the meta-analysis of the nine studies via conventional marginal association test (with $p$-value cutoff $5\times 10^{-8}$), GK-marginal (with target FDR level 0.10), and GK-susie-rss (with target FDR level 0.10).

\begin{figure}[htbp]
  \centering
  \begin{subfigure}[b]{0.98\textwidth} 
    \includegraphics[width=\textwidth]{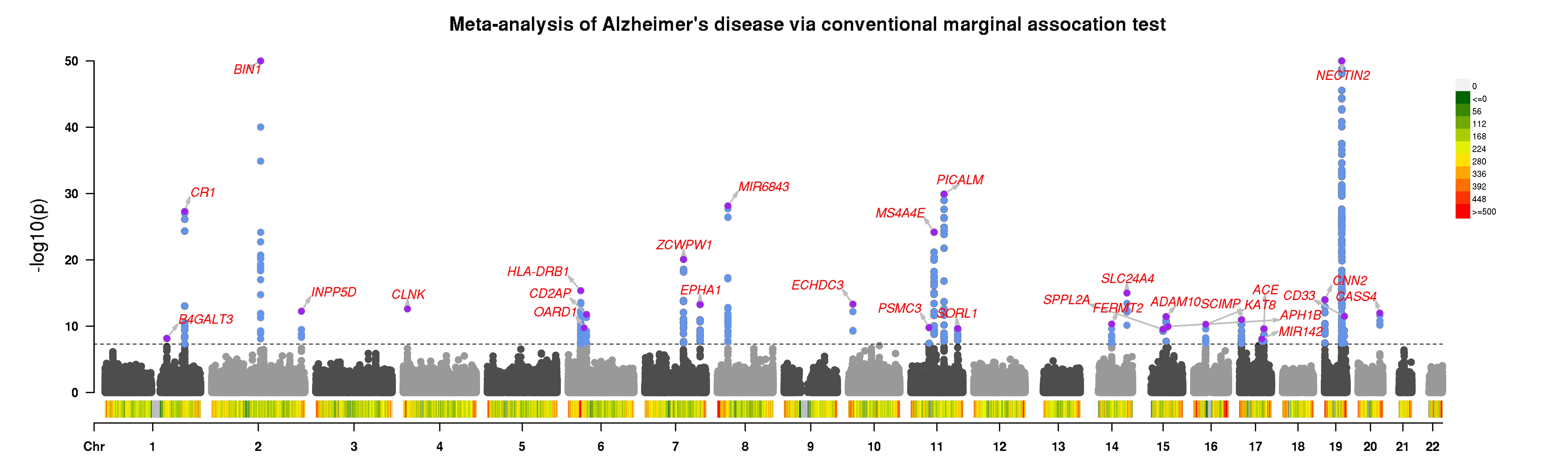}
  \end{subfigure}
  \caption{Graphical illustration of the result by applying conventional marginal association test on meta-analysis for AD. The dotted line represents the conventional genome-wide p-value threshold of $5\times10^{-8}$. P-values are truncated at $10^{-50}$ for better visualization. The results are obtained from the meta-analysis $p$-values calculated based on Section \ref{app:meta_z_scores}. Variant density is shown at the bottom of plot (number of variants per 1Mb).}
  \label{fig:real_data_gwas}
\end{figure}

\begin{figure}[htbp]
  \centering
  \begin{subfigure}[b]{0.98\textwidth} 
    \includegraphics[width=\textwidth]{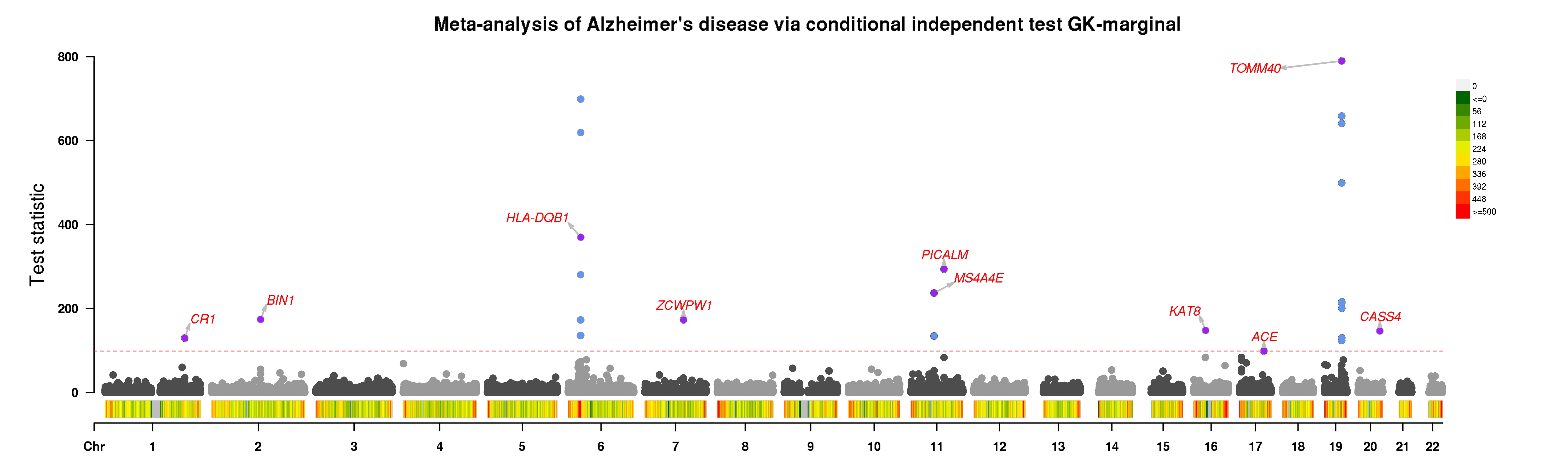}
  \end{subfigure}
  \caption{Graphical illustration of the result by applying the GK-marginal on meta-analysis for AD. Each point represents a group of genetic variants. With respect to the target FDR level 0.1, points of identified groups are highlighted in blue or purple. For each locus with at least one identified group, the name of the locus is presented at the variant group with the largest importance statistic (highlighted in purple). Variant density is shown at the bottom of plot (number of variants per 1Mb).}
  \label{fig:real_data_marginal}
\end{figure}

\begin{figure}[htbp]
  \centering
  \begin{subfigure}[b]{0.98\textwidth} 
    \includegraphics[width=\textwidth]{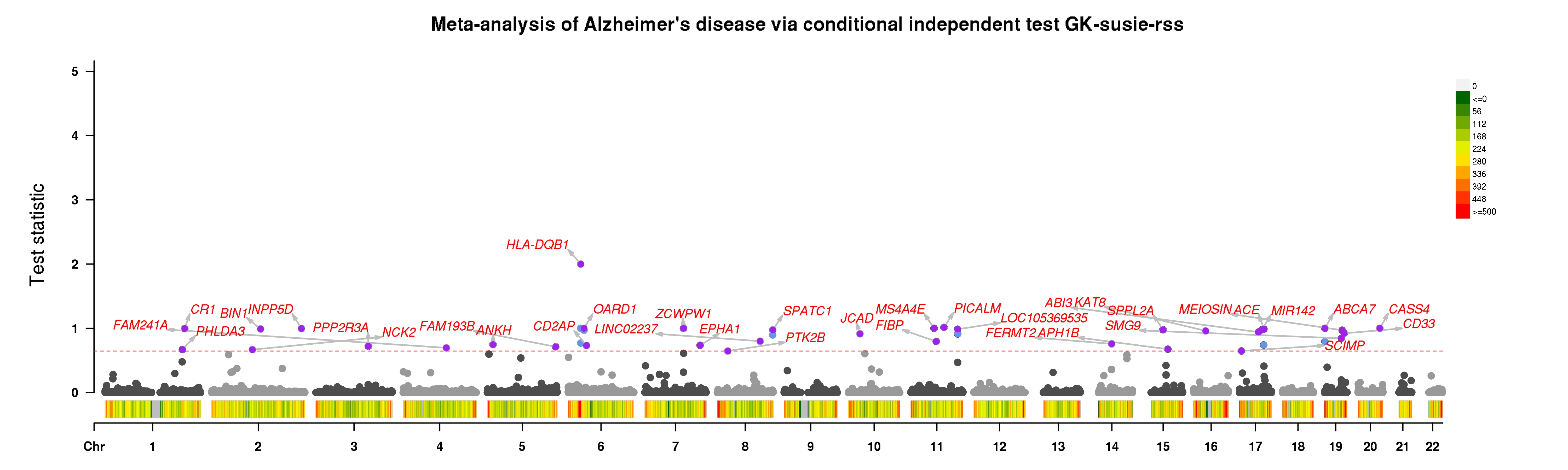}
  \end{subfigure}
  \caption{Graphical illustration of the result by applying the GK-susie-rss on meta-analysis for AD. Each point represents a group of genetic variants. With respect to the target FDR level 0.1, points of identified groups are highlighted in blue or purple. For each locus with at least one identified group, the name of the locus is presented at the variant group with the largest importance statistic (highlighted in purple). Variant density is shown at the bottom of plot (number of variants per 1Mb).}
  \label{fig:real_data_susie}
\end{figure}

The conventional marginal association test selects many feature groups because it focuses on marginal correlations between feature groups and the response while ignoring spurious correlation induced by linkage disequilibrium. This is shown in Figure \ref{fig:real_data_gwas}, where the conventional marginal association test tends to select many nearby loci. This issue is alleviated by the GhostKnockoffs approach that tests conditional independence as seen in Figures \ref{fig:real_data_ghostlasso}, \ref{fig:real_data_marginal}, and \ref{fig:real_data_susie}. 

\section{Running Lasso on binary responses} \label{app:binary response}

In genetic datasets, the response $Y$ is often binary. Performing Lasso or Lasso-type regressions on binary response may sound unreasonable since it violates the usual linear model assumption. One might assume that utilizing penalized logistic regression to generate feature importance statistics would be much more effective. However, a bit surprisingly, we demonstrate that this intuition may not be correct through the following two simulations.

For the first column of Figure \ref{fig:binary_response}, we generate $X_i \stackrel{iid} {\sim} \mathcal{N}(\bzero,\frac{1}{\sqrt{n}}\bI_p)$, and, conditional on $X_i$, $\bbP(Y_i=1)=\frac{1}{1+e^{-\bbeta^\top X_i}}$ and $\bbP(Y_i=0)=1-\bbP(Y_i=1)$, where $n=1000$ and $p=300$. We create $\bbeta$ by uniformly randomly selecting 30 coordinates to be non-zero. The signs of these non-zero coordinates are assigned to be either positive or negative with equal probability. The dark curve represents the knockoffs procedure with Lasso coefficient difference statistic (with tuning parameter chosen by cross-validation), i.e., KF-lassocv. The red curve represents the knockoffs procedure with coefficient difference statistic generated by $L_1$-penalized logistic regression. We vary the signal amplitudes such that we observe relatively complete power profiles below. The target FDR is $0.1$. Each point on the curves represents an average over 200 replications. For the second column of Figure \ref{fig:binary_response}, we show the result for AR(1) features. Here, $n=600$, $p=200$ and the signal amplitude (i.e., the magnitude of non-zero $\beta$ values) is fixed to be 0.5. Otherwise, the simulation setting is exactly the same as the independent case. We observe that the two methods considered have almost the same power and FDR, so the use of penalized logistic regression does not meaningfully affect the results.

\begin{figure}[htbp] 
  \centering
  \includegraphics[width=0.8\textwidth]{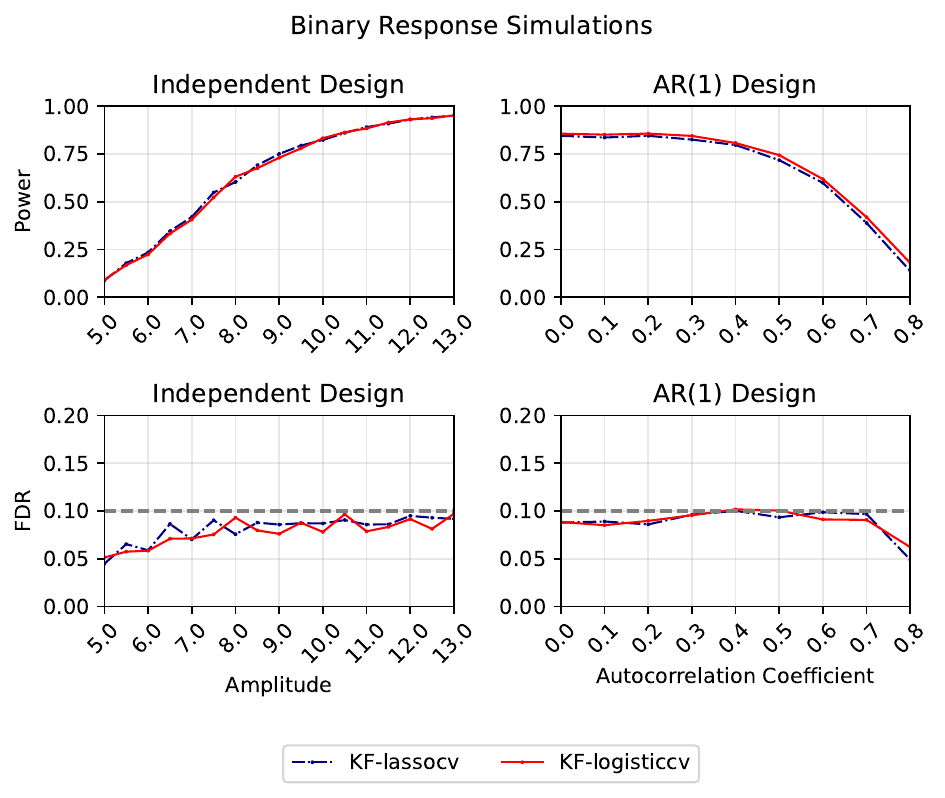}
  \caption{Power and FDR plots when the response is generated by a logistic regression model. Each point is an average over 200 replications.}
  \label{fig:binary_response}
\end{figure}

\end{document}